\documentclass[11pt,times,letter]{article}
\usepackage{amsmath,amsfonts,amsthm}
\usepackage[margin=1in]{geometry}
\usepackage[normalem]{ulem}
\usepackage{amssymb}
\usepackage{natbib}
\usepackage{float}
\usepackage{color}

\usepackage[normalem]{ulem}

\definecolor{linkcolor}{rgb}{0.1,0,0.7}
\definecolor{urlcolor}{rgb}{1,0,0}

\usepackage[unicode=true,colorlinks,citecolor=linkcolor,linkcolor=linkcolor,urlcolor=blue]{hyperref}

\usepackage{graphicx}
\usepackage{verbatim}
\usepackage{enumitem}
\usepackage{color}
\usepackage{tikz}
\usetikzlibrary{decorations.fractals}
\usetikzlibrary{patterns,arrows,shapes,decorations.pathreplacing}
\usepackage{bm}

\usepackage{subfigure}
\usepackage{soul}

\numberwithin{equation}{section}

\theoremstyle{plain}

\newcommand {\R}{\mathbb{R}}
\newcommand {\N}{\mathbb{N}}
\newcommand {\Ex}{\mathbb{E}}

\newtheorem{lemma}{Lemma}[section]
\newtheorem{theorem}{Theorem}[section]
\newtheorem{definition}{Definition}[section]

\newtheorem{remark}{Remark}[section]
\allowdisplaybreaks[4]

\def \E{\mathbb{E}}

\def \Fx{\mathbb{F}}

\makeatletter
\def\@setcopyright{}
\def\serieslogo@{}
\makeatother

\title{A Mean Field Game Approach to Equilibrium Consumption under External Habit Formation}

\author{Lijun Bo \thanks{Email: lijunbo@ustc.edu.cn, School of Mathematical Sciences, University of Science and Technology of China, Hefei, Anhui Province, 230026, China.}
\and
Shihua Wang \thanks{Email: wangshihua@xidian.edu.cn, School of Mathematics and Statistics, Xidian University, Xi'an, 710126, China.}
\and
Xiang Yu \thanks{Email: xiang.yu@polyu.edu.hk, Department of Applied Mathematics, The Hong Kong Polytechnic University, Kowloon, Hong Kong.}
}
\date{\vspace{-1cm}}

\begin{document}

\maketitle
\begin{abstract}
This paper studies the equilibrium consumption under external habit formation in a large population of agents. We first formulate problems under two types of conventional habit formation preferences, namely linear and multiplicative external habit formation, in a mean field game framework. In a log-normal market model with the asset specialization, we characterize one mean field equilibrium in analytical form in each problem, allowing us to understand some quantitative properties of the equilibrium strategy and conclude some financial implications caused by consumption habits from a mean-field perspective. In each problem with $n$ agents, we construct an approximate Nash equilibrium for the $n$-player game using the obtained mean field equilibrium when $n$ is sufficiently large. The explicit convergence order in each problem can also be obtained.

\vspace{0.4 cm}

\noindent{\textbf{Mathematics Subject Classification (2020)}: 49N80, 91A15, 91B42, 91B50, 91B10}
\vspace{0.2 cm}

\noindent{\textbf{Keywords}:} Catching up with the Joneses, linear habit formation, multiplicative habit formation, mean field equilibrium, approximate Nash equilibrium
\end{abstract}

\section{Introduction}
To reconcile the observed equity premium puzzle, the time non-separable habit formation preference has been proposed (see \cite{constantinides1990habit}) as a new paradigm for measuring individual's consumption performance and risk aversion over the past decades. The dependence of the utility on the past consumption path can partially explain why consumers' reported sense of well-being often seems more related to recent changes instead of the absolute levels. The time non-separable structure can also better explain the well documented smoothness in consumption data. Some recent studies on internal habit formation for an individual agent can be found in \cite{detemple1992}, \cite{englezos2009utility}, \cite{schroder2002}, \cite{Yu15}, \cite{Yu17}, \cite{GLY20}, \cite{vanBil20}, \cite{Yu22}, \cite{Bah2022} among others.

Another research direction with fruitful outcomes is to extend the previous framework to the study of equilibrium consumption behavior for a group of interacting agents, where each individual's habit level depends on the average of consumption habits from all peers in the economy; see \cite{Abel90}, \cite{detemple1991}, \cite{Abel99}, \cite{Camp99} and many subsequent studies. The so-called \textit{catching up with the Joneses} has been widely used to refer to the external habit formation and depict the flavor of competition in the equilibrium problem as each agent chooses the relative consumption by competing with the historical consumption from others. In the literature with $n$ agents, both the linear external habit formation and the multiplicative external habit formation have attracted a lot of interested thanks to their mathematical tractability and financial interpretations. The linear external habit formation preference (see, for example, \cite{constantinides1990habit}, \cite{detemple1991}) measures the difference between the current consumption rate and the average of the aggregate consumption from all agents under the CRRA utility, featuring the \textit{addictive} consumption habits in the sense that each agent can not tolerate the consumption to fall below the external habit level induced by the infinite marginal utility. On the other hand, the multiplicative external habit formation preference (see, for example, \cite{Abel90},  \cite{Camp99}, \cite{Carroll20}) is defined on the ratio of the current consumption and the average of the aggregate consumption, which is conventionally referred to \textit{non-addictive} habit formation as the agent can bear the consumption plan to be lower than the habit level from time to time and may strategically suppress the consumption temporarily to accumulate higher wealth from the financial market.

In this paper, we revisit these two types of external habit formation preferences in the literature, however, from the mean field game (MFG) point of view. We aim to investigate the equilibrium consumption behavior with a continuum of agents when each agent focuses on the investment on the individual asset class. In particular, we first study the equilibrium consumption as a MFG problem when the market is populated by infinitely many agents. We then establish some connections to the model with $n$ agents by constructing and verifying the approximated Nash equilibrium when $n$ is sufficiently large. Our contributions are two-fold:

 (i) In contrast to conventional studies on equilibrium consumption under external habit formation preferences, we do not characterize the excessive return of the risky asset as the equilibrium output (see, for example, \cite{Abel90},  \cite{detemple1991}, \cite{Abel99}). Instead, we regard the external habit formation as the relative performance benchmark and choose to study the associated $n$-player game and MFG in the same spirit of \cite{lacker2019} and \cite{lacker2020}. We differ the equilibrium from the excessive return in order to avoid the additional technicality in verifying the consistency condition, which usually also requires the market clearing assumption. In the mean-field model with infinitely many agents, it is more natural to consider the aggregate habit formation process to define and verify the NE condition. We therefore can take advantage of the tractability in the MFG formulation when the influence of each agent on the population is negligible; see \cite{huang2006large} and \cite{lasry2007mean}. In each MFG problem, we can notably obtain one mean field equilibrium in the analytical form, allowing us to investigate some impacts on the equilibrium consumption by model parameters and the competition nature from external habit preference.

(ii) Our work is also an important add-on to the literature of relative performance by featuring the path-dependent benchmark. The research on $n$-player games and MFGs under the relative performance has been active in recent years. To name a few, we refer to \cite{espinosa2015optimal}, \cite{FreiDoesReis2011}, \cite{lacker2019}, \cite{FuZhou2023}, \cite{HuZ21}, \cite{Bo2021} among others. However, only a handful of studies has incorporated the consumption control into the MFG formulation. \cite{lacker2020} extended the formulation in \cite{lacker2019} by considering the relative consumption where the constant Nash equilibrium (NE) is obtained when model parameters are time-independent. \cite{ReisPla2022} study the portfolio-consumption MFG under the power type forward performance processes (FPP). \cite{Fu2023} recently establishes a one-to-one correspondence between NE of the MFG and the solution to some FBSDE, however the general well-posedness of the FBSDE therein remains an open problem. Only when market parameters do not depend on the common Brownian motion, \cite{Fu2023} can obtain the mean-field NE in a closed-form. It is noted that the methodology in the aforementioned studies can not be applied directly to tackle the new challenge caused by the path-dependent benchmark process, which is generated by the consumption control. We resort to the PDE method and work carefully with the fixed point result arising from the consistency condition on the aggregated habit formation process. Moreover, we also study the approximate NE in the n-player game using the obtained the mean field equilibrium. In particular, for the linear habit formation, some technical efforts are needed to guarantee that the constructed solution indeed satisfies the addictive habit constraint in the n-player game setting. In response, we consider some auxiliary state processes in the construction of the candidate NE control (see definition in \eqref{eq:pi-c-star-i-two}) on the strength of the simple structure of the mean-field NE process and the geometric Brownian motion property of the auxiliary process. Some technical arguments are also required to derive some estimations and to show the convergence results of our constructed approximate NE for both linear and multiplicative habit formation preferences.

The rest of the paper is organized as follows. In Section \ref{sec:introduc}, we introduce the $n$-player game problems under linear and multiplicative habit formation preferences when the asset specialization is applied to each agent. In Section \ref{sec:MFG}, we formulate two MFG problems under two types of external habit formation with infinitely many agents. A mean field equilibrium in each problem is established in analytical form. Some numerical illustrations and sensitivity analysis of the mean field equilibrium as well as their financial implications are presented in Section \ref{sec:num}. In Section \ref{sec:approx}, we construct and verify an approximate Nash equilibrium in each $n$-player game when $n$ is sufficiently large using the mean field equilibrium and derive the explicit order of the approximation error. Some conclusions and future research directions are given in Section \ref{sec:con}.

\section{The Market Model}\label{sec:introduc}

Fix a finite time horizon $T>0$, let $(\Omega, \mathcal{F}, \mathbb{F}, \mathbb{P})$ be a filtered probability space, where the filtration $\mathbb{F}=(\mathcal{F}_t)_{t\in[0,T]}$ satisfies the usual conditions. We consider a market model consisting of one riskless bond and $n$ risky assets, in which there are $n$ heterogeneous agents who dynamically invest and consume up to the finite horizon $T$. Without loss of generality, the interest rate of the riskless bond is assumed be $r=0$ by changing of num\'{e}raire. 

Similar to \cite{lacker2019},  the asset specialization to each agent is assumed that the agent $i$ can only invest in the risky asset $S^i=(S^i_t)_{t\in[0,T]}$ whose price process follows
\begin{align}\label{eq:stockSi}
dS^i_t=S^i_t\mu_i dt+S^i_t\sigma_i dW^i_t,\ \ i=1,\ldots, n,
\end{align}
where $(W^1,\ldots,W^n)=(W_t^1,\ldots, W_t^n)_{t\in[0,T]}$ is an $n$-dimensional $\mathbb{F}$-adapted standard Brownian motion. For $i=1,\ldots,n$, let $(\pi^i, C^i)=(\pi^i_t, C_t^i)_{t\in[0,T]}$ be an $\Fx$-adapted process, where $\pi^i$ represents the dynamic proportion of wealth that the agent $i$ allocates in the risky asset $S^i$ and $C^i$ represents the consumption rate process of agent $i$. We also denote the consumption-to-wealth proportion $c_t^i:= C^i_t/X^i_t$ if $X_t^i>0$. The resulting self-financing wealth process $X^i=(X^i_t)_{t\in[0,T]}$ of agent $i$ is governed by
\begin{align}\label{eq:X-i}
\frac{dX^i_t}{X^i_t} = \pi^i_t \mu_i dt + \pi^i_t \sigma_i dW^i_t - c^i_tdt,\quad t\in[0,T],
\end{align}
with the initial wealth $X^i_0=x^i_0>0$.

The so-called \textit{habit formation} process $Z^i=(Z_t^i)_{t\in[0,T]}$ of agent $i$ generated by the consumption rate process $C^i=(C_t^i)_{t\in[0,T]}$ is defined by
\begin{align*}
  d Z^i_t = -\delta_i(Z^i_t - C^i_t)dt,\quad Z_0^i=z^i_0>0,
\end{align*}
where $z^i_0$ stands for the initial habit. It follows that
\begin{align}\label{eq:solZti}
Z^i_t=e^{-\delta_i t}\left(z^i_0+\int_0^t \delta_i e^{\delta_i s}C^i_sds\right),\quad t\in[0,T].
\end{align}
Here, the habit intensity parameter $ \delta_i> 0$ depicts how much the habit is influenced by the recent consumption path comparing with the initial habit level.

For the group of $n$ agents in the financial market, let us define their average habit formation process by
\begin{align}\label{eq:barZn}
\bar{Z}_t^n := \frac{1}{n}\sum_{i=1}^{n}Z_t^i,\quad t\in[0,T],
\end{align}
which depicts the average of aggregate consumption trend in the economy. We adopt two well-studied external habit formation preferences in the literature, namely, the {\it linear external habit formation} (see  \cite{constantinides1990habit} and \cite{detemple1991}) and the {\it multiplicative external habit formation} (see \cite{Camp99} and \cite{Carroll20}). That is, each agent's utility function is measured by the distance between his current consumption rate and the benchmark process described by the average habit formation process from all $n$ peers. Therefore, other agent's historical consumption pattern directly affects the satisfaction and risk aversion of the $i$-th agent.

Mathematically speaking, for the $i$-th agent, the optimal relative consumption problem under the {\it linear external habit formation} is defined by
\begin{align}
\max_{(\pi^i,c^i)\in\mathcal{A}^{l,i}(x^i_0)}J_i^{l}(\bm{\pi},\bm{c}) &:= \max_{(\pi^i,c^i)\in\mathcal{A}^{l,i}(x^i_0)}\mathbb{E}\left[\int_0^T U_i\left(c^i_tX_t^i - \alpha_i\bar{Z}^n_t\right)dt + U_i(X_T^i) \right],\label{eq:Objective-i-two}
\end{align}
and the optimal consumption problem under the {\it multiplicative external habit formation} is defined by
\begin{align}
\max_{(\pi^i,c^i)\in\mathcal{A}^{m,i}(x^i_0)}J_i^{m}(\bm{\pi},\bm{c}) &:= \max_{(\pi^i,c^i)\in\mathcal{A}^{m,i}(x^i_0)}\mathbb{E}\left[\int_0^T U_i\left(\frac{c^i_tX_t^i}{(\bar{Z}_t^n)^{\alpha_i}}\right)dt + U_i(X_T^i) \right],\label{eq:Objective-i}
\end{align}
where $(\bm{\pi},\bm{c})=((\pi^1,c^1),\ldots, (\pi^n,c^n))$, $\alpha_i\in(0,1]$ represents the habit persistence that can also be understood as the competition level of the relative performance, and $U_i:\R_+\to\R_+$ ($i=1,\ldots,n$) is the power utility of agent $i$ that
\begin{align}\label{eq:Uix}
U_i(y)=\frac{1}{p_i}y^{p_i},\quad p_i\in (0,1),\quad y\geq 0.
\end{align}

\begin{remark}
We emphasize that our external habit formation preferences in \eqref{eq:Objective-i-two} and \eqref{eq:Objective-i} are exactly from some existing studies such as     \cite{constantinides1990habit}, \cite{detemple1991}), \cite{Camp99} and \cite{Carroll20}. However, as opposed to these papers, we do not employ the external habit formation to investigate the consumption-based equilibrium pricing. Therefore, our focus is not the equilibrium mean return of the underlying risky asset. In the present paper, we would like to study 
problems \eqref{eq:Objective-i-two} and \eqref{eq:Objective-i} in the limiting model with infinitely many agents as some MFG problems and examine the existence of the mean field equilibrium through the aggregated average habit formation process. Later, building upon our obtained mean field equilibrium, we will also construct and verify the approximate Nash equilibrium in the n-player game problems.    
\end{remark}

For two types of external habit formation preferences, we stress that the admissible control sets are different. In problem \eqref{eq:Objective-i-two} under the linear habit formation, the external consumption habits are \textit{addictive} in the sense that $C^i_t=c^i_tX^i_t\geq \alpha_i \bar{Z}^n_t$ for $t\in[0,T]$ a.s. because of the infinite marginal utility. Therefore, we define ${\cal A}^{l,i}(x^i_0)$ as the set of $\Fx$-adapted consumption-portfolio pairs $(\pi^i,c^i)$ such that $c^i_tX^i_t\geq \alpha_i\bar{Z}^n_t$ and no bankruptcy condition holds that $X^i_t>0$ a.s. for $t\in[0,T]$. To ensure that the admissible set ${\cal A}^{l,i}(x^i_0)$ is non-empty, we additionally require that $x^i_0 > \alpha_i z^i_0T$ such that the initial wealth $x^i_0$ of the agent $i$ is sufficiently large to support the consumption under addictive habit constraint.

On the other hand, in view of the ratio form in problem \eqref{eq:Objective-i} under multiplicative habit formation, the consumption can fall below the habit level and the external habit is \textit{non-addictive}. Therefore, we define ${\cal A}^{m,i}(x^i_0)$ as the set of $\Fx$-adapted consumption-portfolio pairs $(\pi^i,c^i)$ such that $c^i_t\geq 0$, a.s. and no bankruptcy is allowed that $X^i_t>0$ a.s. for $t\in[0,T]$. For the well-posedness of the problem, it is additionally assumed in ${\cal A}^{m,i}(x^i_0)$ that the uniform boundedness condition holds that $\sup_{i\geq1}\sup_{t\in[0,T]}|\pi_t^i|\vee|c_t^i|<\infty$, a.s. and the initial habit is strictly positive that $z^i_0>\epsilon$ for some constant $\epsilon>0$.

For technical convenience and ease of presentation, we make the following assumption throughout the paper.

\begin{itemize}
\item[] $\bm{(A_{h})}$: Assume that all agents are homogenous in their initial wealth, the initial habit, the habit discounting factor and the habit persistence level such that $(x^i_0, z^i_0, \delta_i, \alpha_i)=(x_0,z_0,\delta, \alpha)\in\R_+^3\times{(0,1]}$, $i=1,\ldots, n$. Additionally, in problem \eqref{eq:Objective-i-two}, it is assumed that $x_0>\alpha z_0 T$; In problem \eqref{eq:Objective-i}, it is assumed that $z_0>\epsilon$ for some constant $\epsilon>0$.
\end{itemize}

Note that the heterogeneity of $n$ agents in the present paper is captured via their different type vectors $o_i:=(\mu_i, \sigma_i, p_i)\in{\cal O}:=\R\times\R_+\times (0,1)$.

\section{Mean Field Game Problems}\label{sec:MFG}

We now proceed to formulate the mean field games under linear and multiplicative external habit formation when the number of agents grows to infinity. The type vector $o_i=(\mu_i, \sigma_i, p_i)$, $i=1,\ldots,n$, induces an empirical measure on the type space ${\cal O}$ given by
\begin{align*}
{\rm m}_n(A):=\frac{1}{n}\sum_{i=1}^{n}\delta_{o_i}(A)=\frac{1}{n}\sum_{i=1}^{n}\mathbf{1}_A(o_i),
\end{align*}
for Borel sets $A\subset {\cal O}$ (i.e., $A\in{\cal B}({\cal O})$). The following assumption is needed to formulate the MFG problem:
\begin{itemize}
\item[] $\bm{(A_{o})}$: there exists a ${\cal O}$-valued random variable $\xi$ under the probability space $(\Omega, \mathbb{F}, \mathbb{P})$ that is independent of Brownian motions $(W^1,\ldots,W^n)$ in \eqref{eq:stockSi} with the law ${\rm m}\in{\cal P}({\cal O})$ such that $\int_{{\cal O}}fd{\rm m}_n\rightarrow \int_{{\cal O}}fd{\rm m}$, as $n\to\infty$, for every bounded and continuous function $f$ on ${\cal O}$ (i.e., $f\in C_b({\cal O})$).
 \end{itemize}
When the type vector $o_i=(\mu_i, \sigma_i, p_i)\to o$ as $i\to\infty$ for some $o=(\mu,\sigma,p)\in{\cal O}$, the random variable $\xi$ satisfies $\mathbb{P}(\xi=o)=1$.

For a given type vector $o=(\mu,\sigma,p)\in{\cal O}$, the wealth process of a representative agent is governed by
\begin{equation}\label{eq:wealth-X}
\frac{dX_t}{X_t}=\pi_t\mu dt+\pi_t\sigma dW_t - c_tdt,\quad X_0=x_0.
\end{equation}
Here, $W=(W_t)_{t\in[0,T]}$ is a scalar Brownian motion under the probability space $(\Omega, \mathbb{F}, \mathbb{P})$ that is independent of the type vector $\xi$ and the Brownian motions $(W^1,\ldots,W^n)$ in \eqref{eq:stockSi}. For $n$ sufficiently large, we may approximate $\bar{Z}^n=(\bar{Z}^{n}_t)_{t\in[0,T]}$ by a deterministic function $\bar{Z}=(\bar{Z}_t)_{t\in[0,T]}$, and this can be heuristically justified by the law of large numbers as long as the individual controls satisfy some mild conditions. To this purpose, let the deterministic function $\bar{Z} = (\bar{Z}_t)_{t\in[0,T]}\in{\cal C}_T:=C([0,T];\mathbb{R}_{+})$ denote the approximation of the average habit formation process $\bar{Z}^n=(\bar{Z}^{n}_t)_{t\in[0,T]}$ as $n\to\infty$.

\subsection{Mean field equilibrium under linear habit formation}

In this section, we formulate and study the MFG problem under linear external habit formulation associated to the $n$-player problem considered in \eqref{eq:Objective-i-two}. Given a deterministic function $\bar{Z}=(\bar{Z}_t)_{t\in[0,T]}$ as the approximation of $\bar{Z}^n=(\bar{Z}^{n}_t)_{t\in[0,T]}$ when $n\to\infty$, the dynamic version of the objective function for a representative agent under linear external habit formulation is defined by
\begin{align}\label{eq:Objective-two}
\bar{J}^l((\pi,c), t,x; \bar{Z}) :=\mathbb{E}_{t,x} \left[\int_{t}^{T} \frac{1}{p}(c_sX_s - \alpha \bar{Z}_s)^pds + \frac{(X_T)^p}{p} \right]
\end{align}
with $\mathbb{E}_{t,x}=\mathbb{E}[\cdot | X_t=x]$. Let $o=(\mu, \sigma, p)$ be a deterministic sample from its distribution. Accordingly, the optimal control problem is given by
\begin{align}\label{eq:value-function-two}
 \sup_{(\pi,c)\in{\cal A}^l(x)}\bar{J}^l((\pi,c), t,x; \bar{Z})=\Ex[V^l(t,x;\xi)]=\int_{{\cal O}}V^l(t,x;o){\rm m}(do),
\end{align}
where ${\cal A}^l(x)$ is the dynamic admissible control set of $\mathbb{F}$-adapted consumption-portfolio pairs $(\pi,c)$ such that $c_sX_s\geq \alpha \bar{Z}_s$, $s\in[t,T]$, and no bankruptcy condition holds that $X_s>0$ a.s. for $s\in [t,T]$. Here, $V^l(t,x,o)$ is the optimal value function under the realization $o$ of the random type vector $\xi$.

Next, in order to study the existence of the mean field equilibrium for the MFG problem under the linear habit formation and the addictive habit constraint, we first introduce 

\begin{align}\label{CTx}
{\cal C}_{T, x_0}:=\left\{\bar{Z}=(\bar{Z}_t)_{t\in[0,T]}\in{\cal C}_{T}: \alpha \int_0^T\bar{Z}_tdt<x_0\right\}.
\end{align}

We next give the definition of the mean field equilibrium when the deterministic $\bar{Z}_t$ is restricted to the set ${\cal C}_{T, x_0}$.
\begin{definition}\label{mfg-def-1}
For a given deterministic $\bar{Z}=(\bar{Z}_t)_{t\in[0,T]}\in{\cal C}_{T, x_0}$, a strategy pair $(\pi^{*,\bar{Z}},c^{*,\bar{Z}})\in{\cal A}^l(x_0)$ is called the {\it best response} strategy to the stochastic control problem (3.3) if $(\pi^{*,\bar{Z}},c^{*,\bar{Z}})$ is an optimal feedback control for the representative agent such that the optimal value function is attained, i.e. $\bar{J}^l(\pi^{*,\bar{Z}},c^{*,\bar{Z}}, t,x; \bar{Z})=\int_{{\cal O}}V^l(t,x,o){\rm m}(do)$. The strategy $(\pi^{l}, c^{l}):=(\pi^{*,\bar{Z}^{l}},c^{*,\bar{Z}^{l}})$ is called a mean field equilibrium if it is the best response to itself in the sense that $\bar{Z}^{l}_t=z_0e^{-\delta t}+\int_0^t \delta e^{\delta(s-t)} \mathbb{E}[c_s^{l} X_s^{l,\bar{Z}^{l}}]ds$, $t\in[0,T]$, where $X^{l, \bar{Z}^{l}}=(X_t^{l, \bar{Z}^{l}})_{t\in[0,T]}$ is the wealth process under the best response control $(\pi^{l}, c^{l})$ with $X_0^{l,\bar{Z}^l}=x_0$.
\end{definition}

By Definition \ref{mfg-def-1}, we first find the best response strategy to the stochastic control problem \eqref{eq:value-function-two} for a given function $\bar{Z}=(\bar{Z}_t)_{t\in[0,T]}\in{\cal C}_{T, x_0}$. Using dynamic program arguments, we can derive the associated HJB equation of the value function $V^l(t,x):=V^l(t,x,o)$ on the effective domain $\{(t,x)\in [0,T]\times \mathbb{R}_+: x>\int_t^T \alpha \bar{Z}_sds\}$ that
\begin{align}\label{eq:V-PDE-two}
\partial_tV^l  +  \sup_{\pi\in\R}\left(\mu\pi x\partial_xV^l  + \frac{\sigma^2}{2}\pi^2x^2\partial_{x}^2V^l\right) +\sup_{c\geq0}\left(\frac{1}{p} (cx-\alpha \bar{Z}_t)^p-cx\partial_xV^l\right)=0,
\end{align}
with the terminal condition $V^l(T,x) = \frac{1}{p}x^p$ for $x>0$.
\begin{lemma}\label{lem:solHJBlimit-two}
Let $\bar{Z}=(\bar{Z}_t)_{t\in[0,T]}\in{\cal C}_{T, x_0}$. The classical solution of the HJB equation \eqref{eq:V-PDE-two} on the effective domain $\{(t,x)\in [0,T]\times \mathbb{R}_+: x>\int_t^T \alpha \bar{Z}_sds\}$ admits the closed-form that
\begin{align}\label{eq:V-solution-thm-2}
   V^l(t,x)  = \frac{1}{p}\left(x - \int_{t}^{T}\alpha \bar{Z}_sds\right)^p g^l(t),
\end{align}
where
\begin{equation}\label{eq:solution-g-2}
  g^l(t)
  :=  \left[\left(1+\frac{1}{a}\right)e^{a(T-t)}-\frac{1}{a} \right]^{1-p},
\end{equation}
and $a:= \frac{p\mu^2}{2(1-p)^2\sigma^2}$.
The feedback functions of the optimal investment and consumption to the problem \eqref{eq:value-function-two} from the initial time are given by
\begin{align}\label{eq:picstar-2}
\pi^{*,\bar{Z}}(t,x) &= \frac{\mu}{(1-p)\sigma^2x}\left(x-\int_{t}^{T}\alpha \bar{Z}_s ds\right),\quad c^{*,\bar{Z}}(t,x)= \frac{1}{x}\left[\alpha \bar{Z}_t +\left(x-\int_{t}^{T} \alpha\bar{Z}_s ds\right)g^l(t)^{\frac{1}{p-1}}\right],
\end{align}
and the controlled optimal wealth process satisfies $X^{*,\bar{Z}}_t>\int_t^T \alpha \bar{Z}_sds$, a.s., for all $t\in[0,T]$.
\end{lemma}

\begin{proof}
Suppose that there exists a classical solution $V$ to the HJB equation~\eqref{eq:V-PDE-two} on the effective domain $\{(t,x)\in [0,T]\times \mathbb{R}_+: x>\int_t^T \alpha \bar{Z}_sds\}$ that is strictly concave (i.e., $\partial_{x}^2V^l<0$) and $\partial_x V^l>0$. The first-order condition gives the optimal (feedback) investment-consumption strategies that
\begin{align*}
  \pi^{*,\bar{Z}}(t,x) &= -\frac{\mu}{\sigma^2} \frac{\partial_xV^l(t,x)}{x \partial_{x}^2V^l(t,x)},\quad
  c^{*,\bar{Z}}(t,x) = \frac{1}{x}\left(\alpha \bar{Z}_t+(\partial_xV^l)^{\frac{1}{p-1}}\right).
\end{align*}
Plugging them into the HJB equation~\eqref{eq:V-PDE-two}, we obtain that
\begin{align*}
 \partial_tV^l -\frac{\mu^2}{2\sigma^2} \frac{(\partial_xV^l)^2}{\partial_{x}^2V^l}
 -\alpha\bar{Z}_t\partial_xV^l+\frac{1-p}{p} (\partial_xV^l)^{\frac{p}{p-1}}=0.
\end{align*}
We conjecture that the value function satisfies the form $V^l(t,x)=\frac{(x-\alpha f(t))^p}{p}g^l(t)$, where $f(t):=\int_t^T\bar{Z}_sds$ and $t {\color{red}\mapsto } g^l(t)$ is a positive function satisfying $g^l(T)=1$. Plugging the expression of $V^l(t,x)$ into the HJB equation, we arrive at
\begin{align}\label{eq:odegtbernoulli}
\frac{1}{p}(g^l)'(t)  -\frac{\mu^2}{2\sigma^2(p-1)} g^l(t)+\frac{1-p}{p}(g^l(t))^{\frac{p}{p-1}}=0.
\end{align}
Note that \eqref{eq:odegtbernoulli} is a Bernoulli ODE. To solve this ODE, let us consider $g^l(t) = (u(t))^{1-p}$ for $t\in[0,T]$. Then $u(t)$ satisfies the following linear ODE:
\begin{equation}\label{eq:ODE-u}
  u^{\prime}(t)+ \frac{\mu^{2} p}{2 \sigma^{2}(1-p)^{2}}u(t)+1 = 0,\quad u(T)=1.
\end{equation}
This yields that
\begin{equation*}
  u(t) = \left(1+\frac{1}{a}\right)e^{a(T-t)}-\frac{1}{a},
\end{equation*}
Here $a = \frac{p\mu^2}{2(1-p)^2\sigma^2} $. Therefore, we can find an explicit classical solution to the HJB equation on the effective domain $\{(t,x)\in [0,T]\times \mathbb{R}_+: x>\int_t^T \alpha \bar{Z}_sds\}$ that
\begin{equation*}
  V^l(t,x) = \frac{(x-\alpha f(t))^p}{p}g^l(t)
  = \frac{1}{p}\left( x-\int_{t}^{T}\alpha \bar{Z}_sds\right)^p \left[\left(1+\frac{1}{a}\right)e^{a(T-t)}-\frac{1}{a} \right]^{1-p},
\end{equation*}
which satisfies $\partial_x^2 V^l<0$ and $\partial_x V^l>0$ .

We can then follow some standard arguments to prove the verification theorem and conclude that the optimal controls $(\pi^{x,\bar{Z}}, c^{*,\bar{Z}})$ of the problem \eqref{eq:value-function-two} are given in feedback form by \eqref{eq:picstar-2} as long as we can show that the resulting wealth process $X^{*,\bar{Z}}$ under $(\pi^{*,\bar{Z}}, c^{*,\bar{Z}})$ satisfies the constraint $X_t^{*,\bar{Z}}>\int_t^T\alpha \bar{Z}_sds$, $t\in[0,T]$, such that $(\pi^{*,\bar{Z}}(t,X_t^{*,\bar{Z}}), c^{*,\bar{Z}}(t,X_t^{*,\bar{Z}})$ is an admissible control. That is, we need to show the existence of a strong solution to the SDE
\begin{equation}\label{eq:wealth-XstarbarZ-2}
\frac{dX_t^{*,\bar{Z}}}{X_t^{*,\bar{Z}}}=\pi^{*,\bar{Z}}_t\mu dt+\pi^{*,\bar{Z}}_t\sigma dW_t - c_t^{*,\bar{Z}}dt,\quad X_0^{*,\bar{Z}}=x_0,
\end{equation}
which evolves in the effective domain $\{(t,x)\in [0,T]\times \mathbb{R}_+: x>\int_t^T \alpha \bar{Z}_sds\}$. To this end, let us consider $Y_t^{*,\bar{Z}} := X_t^{*,\bar{Z}} - \int_{t}^{T}\alpha \bar{Z}_s ds$, $t\in [0,T]$. We deduce from \eqref{eq:wealth-XstarbarZ-2} that
\begin{equation}\label{eq:wealth-YstarbarZ}
 \frac{dY_t^{*,\bar{Z}}}{Y_t^{*,\bar{Z}}} = \left[\frac{\mu^2}{(1-p)\sigma^2}- g^l(t)^{\frac{1}{p-1}}\right] dt + \frac{\mu}{(1-p)\sigma} dW_t,\quad Y_0^{*,\bar{Z}}= x_0 -\int_{0}^{T}\alpha \bar{Z}_sds.
\end{equation}
It follows that $Y_t^{*,\bar{Z}}$ is a GBM, and the SDE \eqref{eq:wealth-XstarbarZ-2} admits a strong solution. Moreover, $X_t^{*,\bar{Z}}>\int_{t}^{T}\alpha\bar{Z}_s ds$ indeed holds thanks to the condition $\int_0^T\alpha\bar{Z}_sds< x_0$  using the fact that $\bar{Z}\in{\cal C}_{T, x_0}$.
\end{proof}

We next examine the fixed point problem from the consistence condition in Definition \ref{mfg-def-1} that
\begin{align}\label{eq:fixed-point-2}
 \bar{Z}^{l}_{t} = e^{-\delta t}\left\{z_0 +\int_{0}^{t}\delta e^{\delta s}\Ex[c_s^{l}X_s^{l,\bar{Z}^{l}}]ds\right\},\quad t\in[0,T].
\end{align}
For $t\in[0,T]$, recall that $Y_t^{*,\bar{Z}}= X_t^{*,\bar{Z}} - \int_{t}^{T}\alpha\bar{Z}_s ds$ is a GBM that satisfies \eqref{eq:wealth-YstarbarZ}. We have that
\begin{equation}\label{eq:Ex-Ystar-barZ}
  \Ex\left[ Y_t^{*,\bar{Z}}\right] =  \left( x_0 - \int_{0}^{T}\alpha\bar{Z}_sds\right)\exp\left\{ \int_{0}^{t} \left( \frac{\mu^2}{(1-p)\sigma^2}- g^l(s)^{\frac{1}{p-1}} \right)ds\right \}.
\end{equation}
With the help of \eqref{eq:Ex-Ystar-barZ}, the consistency condition \eqref{eq:fixed-point-2} for $\bar{Z}=(\bar{Z}_t)_{t\in[0,T]}$ can be written as
\begin{align}\label{eq:fixed-point-3-2}
 d\bar{Z}_{t}&= -\delta\bar{Z}_{t}dt + \delta\alpha\bar{Z}_{t}dt + \delta\Ex\left[X_t^{*,\bar{Z}} -\int_{t}^{T}\alpha\bar{Z}_s ds \right]g^l(t)^{\frac{1}{p-1}} dt  \\
 &= \delta(\alpha-1)\bar{Z}_{t}dt + \delta \left(x_0-\int_{0}^{T}\alpha\bar{Z}_sds  \right)\exp\left\{ \int_{0}^{t} \left( \frac{\mu^2}{(1-p)\sigma^2}- g^l(s)^{\frac{1}{p-1}} \right)ds\right \} g^l(t)^{\frac{1}{p-1}} dt \nonumber
\end{align}
with $\bar{Z}_0=z_0$.

We then have the next main result.
\begin{theorem}\label{lem:fixedpoint-two}
There exists a unique deterministic fixed point $\bar{Z}^{l}=(\bar{Z}^{l}_t)_{t\in[0,T]}\in{\cal C}_{T, x_0}$ to Eq.~\eqref{eq:fixed-point-3-2}, and hence $(\pi^{l}, c^{l})=(\pi^{*,\bar{Z}^{l}}, c^{*,\bar{Z}^{l}})$ is a mean field equilibrium, where $(\pi^{*,\bar{Z}}, c^{*,\bar{Z}})$ is defined in \eqref{eq:picstar-2} for a given $\bar{Z}\in{\cal C}_{T, x_0}$. 
\end{theorem}

\begin{proof}
Let us define that, for all $(t,Z)\in[0,T]\times{\cal C}_T$,
\begin{align}\label{eq:PHI-2}
 \Phi(t,Z)
 & := z_0 +\delta(\alpha-1)\int_{0}^{t}Z_{s}ds + \delta\left(x_0-\int_{0}^{T}\alpha Z_udu  \right)\int_0^t {\phi}(s)ds,
\end{align}
where ${\phi}(t):=\exp\{\int_{0}^{t}(\frac{\mu^2}{(1-p)\sigma^2}- g^l(u)^{\frac{1}{p-1}})du\} g^l(t)^{\frac{1}{p-1}}$. In light of \eqref{eq:PHI-2}, for any $Z^1,Z^2\in{\cal C}_{T}$, we have that
\begin{align}\label{eq:esti-phi}
 \left\|\Phi(\cdot,Z^1)-\Phi(\cdot,Z^2)\right\|_T &\leq \delta\int_0^T |Z^1_u-Z_u^2|du  +  \delta \left(\int_0^T |Z^1_u-Z_u^2|du\right) \left(\int_0^T\phi(t)dt\right)\nonumber\\
 &\leq C(T) \|Z^1-Z^2\|_{T},
\end{align}
where $C(T):=\delta T(\int_0^T\phi(t)dt+1)$ and $\|Z\|_T:=\sup_{t\in[0,T]}|Z_t|$ for $ Z\in \mathcal{C}_T $. Note that $T \mapsto C(T)$ is continuous on $\R_+$ and it satisfies $\lim_{T\to0}C(T)=0$. Then, we can choose $t_1\in(0,T]$ small enough such that $C(t_1)\in(0,1)$. Thus,  ${\Phi}$ is a contraction map on ${\cal C}_{t_1}$ by \eqref{eq:esti-phi}, and hence there exists a unique fixed point of $\Phi$ on $[t_0,t_1]$ with $t_0=0$. Note that $T  \mapsto C(T)$ is independent of $z$. Then, we can apply this similar argument to conclude that there exists a unique fixed point of $\Phi$ on $[t_1,t_2]$ for some $t_2>t_1$ small enough. Repeating this procedure, we can conclude the existence of a unique fixed point $\bar{Z}^{l}$ of $\Phi$ on $[0,T]$.

We next verify that the fixed point $\bar{Z}^{l}$ of $\Phi$ on $[0,T]$ satisfies $\int_0^T \alpha\bar{Z}^{l}_tdt<x_0$. As $\bar{Z}^{l}$ is the unique fixed point of $\Phi$ (i.e., $\bar{Z}_t=\Phi(t,\bar{Z})$ for $t\in[0,T]$), we deduce from \eqref{eq:PHI-2} that
\begin{align*}
x_0-\int_0^T\alpha\bar{Z}^{l}_tdt = x_0- \alpha z_0  \int_{0}^{T}e^{\delta(\alpha-1)u}du - \delta\alpha\left(x_0-\int_{0}^{T}\alpha\bar{Z}^{l}_udu  \right)\int_0^T\int_0^t e^{\delta(\alpha-1)(s-t)}\phi(s)dsdt.
\end{align*}

This yields that $x_0-\alpha\int_0^T\bar{Z}^{l}_tdt=(x_0-\alpha z_0\int_{0}^{T}e^{\delta(\alpha-1)u}du)/(1+\delta \alpha\int_0^T\int_0^t e^{\delta(\alpha-1)(s-t)}\phi(s)dsdt)$. Note that $x_0-\alpha z_0\int_{0}^{T}e^{\delta(\alpha-1)u}du > x_0-\alpha z_0T>0$ by the assumption $\bm{(A_{h})}$, we hence conclude that $x_0-\int_0^T\alpha\bar{Z}^{l}_tdt>0$, which completes the proof.
\end{proof}

\subsection{Mean field equilibrium under multiplicative habit formation}\label{sec:mfg-multiplicative}

This section formulates and studies the MFG problem under the multiplicative external habit formation associated to the $n$-player game problem defined in \eqref{eq:Objective-i}. The dynamic version of the objective function of a representative agent is defined by
\begin{align}\label{eq:Objective-1}
\bar{J}^m((\pi,c), t,x; \bar{Z}) :=  \mathbb{E}_{t,x} \left[\int_{t}^{T} \frac{ (c_sX_s)^p}{p(\bar{Z}_s)^{\alpha p}}ds + \frac{(X_T)^p}{p}  \right].
\end{align}

The stochastic control problem is given by
\begin{align}\label{eq:value-function}
\sup_{(\pi,c)\in{\cal A}^m(x)}\bar{J}^m((\pi,c), t,x; \bar{Z})=\int_{\mathcal{O}} V^m(t,x,o){\rm m}(do),
  \end{align}
where ${\cal A}^m(x)$ is the admissible control set for the MFG problem that is defined similar to ${\cal A}^{m,i}(x_0)$, and $V^m(t,x,o)$ is the value function associated with the objective functional \eqref{eq:Objective-1} when the random type vector $\xi=o\in{\cal O}$.

Note that $z_0>\epsilon>0$ in the assumption $\bm{(A_{h})}$. For $\beta:=\epsilon^{\frac{1}{1-p}}$, let us denote
\begin{align}\label{CTbeta}
{\cal C}_{T,\beta}:=\{ \bar{Z}=(\bar{Z}_t)_{t\in[0,T]}\in{\cal C}_T: \bar{Z}_t\geq \beta,\ \forall t\in[0,T] \}.
\end{align}

Recall that there is no habit constraint under the multiplicative habit formation, we next give the definition of the mean field equilibrium when the deterministic $\bar{Z}_t$ is restricted to the set ${\cal C}_{T,\beta}$.

\begin{definition}\label{mfg-def-2}
For a given deterministic function $\bar{Z}=(\bar{Z}_t)_{t\in[0,T]}\in{\cal C}_{T,\beta}$, let $(\pi^{*,\bar{Z}},c^{*,\bar{Z}})\in{\cal A}^m(x_0)$ be the best response strategy to the stochastic control problem \eqref{eq:value-function}. The strategy $(\pi^{m}, c^{m}):=(\pi^{*,\bar{Z}^m},c^{*,\bar{Z}^m})$ is called a mean field equilibrium if it is the best response to itself in the sense that $\bar{Z}^m_t=z_0e^{-\delta t}+\int_0^t \delta e^{\delta(s-t)} \mathbb{E}[c_s^{m} X_s^{m,\bar{Z}^m}]ds$, $t\in[0,T]$, where $X^{m, \bar{Z}^m}=(X_t^{m, \bar{Z}^m})_{t\in[0,T]}$ is the wealth process under the best response control
$(\pi^{m}, c^{m})$ with $X_0^{m, \bar{Z}^m}=x_0$.
\end{definition}

Similarly, we first solve the stochastic control problem~\eqref{eq:value-function} with a deterministic sample $o=(\mu, \sigma, p)$ from its distribution. The associated HJB equation on the domain $(t,x)\in [0,T]\times\mathbb{R}_+$ is given by
\begin{align}\label{eq:V-PDE}
\partial_tV^m  +  \sup_{\pi\in\R}\left(\mu\pi x\partial_xV^m  + \frac{\sigma^2}{2}\pi^2x^2\partial_{x}^2V^m\right) +\sup_{c\geq0}\left(-cx\partial_xV^m  + \frac{1}{p} c^p x^p \bar{Z}_t^{-\alpha p}\right)=0
\end{align}
with the terminal condition $V^m(T,x) = \frac{1}{p}x^p$ for all $x>0$. The best response control is given in the next result.

\begin{lemma}\label{lem:solHJBlimit}
Given $\bar{Z} = (\bar{Z}_t)_{t\in[0,T]}\in{\cal C}_{T,\beta}$, the classical solution to the HJB equation \eqref{eq:V-PDE} admits the following closed-form that
\begin{align}\label{eq:V-solution-thm}
   V^m(t,x)  = \frac{1}{p}x^p g^m(t),\quad t\in[0,T],
\end{align}
where
\begin{equation}\label{eq:solution-g}
 g^m(t) := \left( e^{b(t-T)}+ e^{bt} \int_{t}^{T} e^{-bs}(\bar{Z}_s)^{\frac{\alpha p}{p-1}}ds \right)^{1-p},\quad  b:= -\frac{\mu^2}{2\sigma^2}\frac{p}{(p-1)^2}.
\end{equation}
The feedback functions of the optimal investment and consumption to the problem \eqref{eq:value-function} from the initial time are given by
\begin{align}\label{eq:picstar}
\pi^{*,\bar{Z}}(t,x) &\equiv \frac{\mu}{(1-p)\sigma^2},\quad
c^{*,\bar{Z}}(t,x)=\bar{Z}_t^{\frac{\alpha p}{p-1}}g^m(t)^{\frac{1}{p-1}},\quad t\in[0,T].
\end{align}
\end{lemma}

\begin{proof}
Let us first assume that the classical solution $V^m$ is strictly concave (i.e., $\partial_{x}^2V<0$). Then, the first-order condition gives the optimal (feedback) strategies that, for $(t,x)\in[0,T]\times\R_+$,
\begin{align}\label{eq:theta-star-0}
  \pi^{*,\bar{Z}}(t,x) &= -\frac{\mu}{\sigma^2} \frac{\partial_xV^m(t,x)}{x \partial_{x}^2V^m(t,x)},\quad
  c^{*,\bar{Z}}(t,x) = \left( x^{1-p}\bar{Z}_t^{\alpha p}  \partial_x V^m(t,x) \right)^{\frac{1}{p-1}}.
\end{align}
Plugging the optimal (feedback) strategies \eqref{eq:theta-star-0} into Eq.~\eqref{eq:V-PDE}, we have that
\begin{align}\label{eq:V-PDE-0}
0 = &\partial_tV^m - \frac{\mu^2}{2\sigma^2} \frac{(\partial_xV^m)^2}{\partial_{x}^2V^m}
 +\frac{1-p}{p}\left(\bar{Z}_t\right)^{\frac{\alpha p}{p-1}}\left(\partial_xV^m\right)^{\frac{p}{p-1}}.
\end{align}
To solve \eqref{eq:V-PDE-0}, we make the ansatz that
\begin{align}\label{eq:Vtxgt}
  V^m(t,x) = \frac{1}{p} g^m(t)x^p,\quad (t,x)\in[0,T]\times\R_+.
\end{align}
Substituting \eqref{eq:Vtxgt} into \eqref{eq:V-PDE-0}, we get that
\begin{align*}
0= & (g^m)'(t)\frac{1}{p} x^p  - \frac{\mu^2}{2(p-1)\sigma^2}x^pg^m(t) + \left(\frac{1-p}{p}\left(\bar{Z}_t\right)^{\frac{\alpha p}{p-1}}x^{p}\right)g^m(t)^{\frac{p}{p-1}}.
\end{align*}
We can obtain the ODE for $g^m(t)$ that
\begin{equation}\label{eq:ODE-g}
  \left\{
  \begin{aligned}
  (g^m)'(t) & =\frac{\mu^2}{2\sigma^2}\frac{p}{p-1}g^m(t) -(1-p)\left(\bar{Z}_t\right)^{\frac{\alpha p}{p-1}}g^m(t)^{\frac{p}{p-1}}, \\[0.4em]
  g^m(T) & = 1.
  \end{aligned}
  \right.
\end{equation}
We then consider
\begin{align}\label{eq:ht}
h(t) = g^m(t)^{\frac{1}{1-p}},\quad t\in[0,T].
\end{align}
Consequently, $h'(t)=\frac{1}{1-p}h(t)^p(g^m)'(t)$, and it follows from \eqref{eq:ODE-g} that
\begin{align}\label{eq:hODE}
h'(t)&= bh(t)-\left(\bar{Z}_t\right)^{\frac{\alpha p}{p-1}},\quad h(T)=1.
\end{align}
with $b = - \frac{\mu^2}{2\sigma^2}\frac{p}{(p-1)^2}$. Then, $h(t)=e^{b(t-T)}+ e^{bt} \int_{t}^{T} e^{-bs} (\bar{Z}_s)^{\frac{\alpha p}{p-1}}ds$ for $t\in[0,T]$. The solution \eqref{eq:solution-g} follows from \eqref{eq:ht}. Finally, it follows from \eqref{eq:V-solution-thm} that $\partial_x^2V<0$ indeed holds. Following some standard verification arguments, the optimal feedback controls to the problem \eqref{eq:value-function}  are given by \eqref{eq:picstar}.
\end{proof}


Let $X^{*,\bar{Z}}=(X_t^{*,\bar{Z}})_{t\in[0,T]}$ be the wealth process under the optimal investment-consumption control in \eqref{eq:picstar} that
\begin{equation}\label{eq:wealth-XstarbarZ}
\frac{dX_t^{*,\bar{Z}}}{X_t^{*,\bar{Z}}}=\pi_t^{*,\bar{Z}}\mu dt+\pi_t^{*,\bar{Z}}\sigma dW_t - c_t^{*,\bar{Z}}dt,\quad X_0^{*,\bar{Z}}=x_0>0.
\end{equation}
Note that $c_t^{*,\bar{Z}}$ in \eqref{eq:picstar} is deterministic, the consistency condition for $\bar{Z}\in\mathcal{C}_T$ reduces to
\begin{equation}\label{eq:fixed-point2}
 \bar{Z}_{t} = e^{-\delta t}\left\{z_0 +\int_{0}^{t}\delta e^{\delta s}c_s^{*,\bar{Z}}\Ex\left[X_s^{*,\bar{Z}}\right]ds\right\},\quad t\in[0,T],
\end{equation}
which is equivalent to
\begin{align}\label{eq:fixed-point-3}
 d\bar{Z}_{t} & =  -\delta\{\bar{Z}_{t}- c_t^{*,\bar{Z}}\Ex[X_t^{*,\bar{Z} }]\}dt.
\end{align}

In order to simplify the consistency condition \eqref{eq:fixed-point2} or \eqref{eq:fixed-point-3}, for a given $\bar{Z}\in\mathcal{C}_T$,
we first compute $\Ex[X_t^{*,\bar{Z}}]$. By virtue of \eqref{eq:wealth-XstarbarZ}, it holds that
\begin{align*}
X_t^{*,\bar{Z}} &= x_0\exp\left\{\int_0^t \left(\pi^{*,\bar{Z}}_s\mu-\frac{\sigma^2}{2}(\pi^{*,\bar{Z}}_s)^2-c_s^{*,\bar{Z}}\right)ds+\pi^{*,\bar{Z}}_t\sigma W_t\right\}.
\end{align*}
Therefore, for $(t,\bar{Z})\in[0,T]\times{\cal C}_T$,
\begin{align}\label{eq:EXtstarZ}
f(t,\bar{Z}):=\Ex\left[X_t^{*,\bar{Z}}\right] 
&=x_0\Ex\left[\exp\left(\int_0^t \left(\pi^{*,\bar{Z}}_s\mu-c_s^{*,\bar{Z}}\right)ds\right)\right]\nonumber\\
&=x_0\exp\left(\int_0^t \left(\frac{\mu^2}{(1-p)\sigma^2}-\bar{Z}_s^{\frac{\alpha p}{p-1}}g^{m}(s)^{\frac{1}{p-1}}\right)ds\right).
\end{align}
Thus, we have from \eqref{eq:fixed-point-3} that
\begin{align}\label{eq:fixed-point-4}
 d\bar{Z}_{t} & =  -\delta\left(\bar{Z}_{t}- \bar{Z}_t^{\frac{\alpha p}{p-1}}g^{m}(t)^{\frac{1}{p-1}}f(t,\bar{Z})\right)dt.
\end{align}

The next main result provides the existence of a mean field equilibrium.

\begin{theorem}\label{prop:fixedpoint}
There exists a unique fixed point $\bar{Z}^m=(\bar{Z}^m_t)_{t\in[0,T]}\in{\cal C}_{T,\beta}$ with $\beta=\epsilon^{\frac{1}{1-p}}$ to Eq.~\eqref{eq:fixed-point-4}, and hence $(\pi^{m}, c^{m})=(\pi^{*,\bar{Z}^m}, c^{*,\bar{Z}^m})$ is a mean field equilibrium, where $(\pi^{*,\bar{Z}}, c^{*,\bar{Z}})$ is defined in  \eqref{eq:picstar} for a given $\bar{Z}\in{\cal C}_{T,\beta}$.
\end{theorem}

\begin{proof}
Let us define
\begin{align}\label{eq:hatZt}
\hat{Z}_t:= \exp\left(\frac{\delta}{1-p}t\right)\bar{Z}_t^{\frac{1}{1-p}},\quad t\in[0,T].
\end{align}
Then, we have from \eqref{eq:fixed-point-4} that
\begin{align}\label{eq:odehatZt}
d\hat{Z}_t&=\frac{\delta}{1-p}\exp\left(\frac{\delta}{1-p}t\right)\bar{Z}_t^{\frac{p(1-\alpha)}{1-p}}
g^{m}(t)^{\frac{1}{p-1}}f(t,\bar{Z})dt\notag\\
&=\frac{\delta}{1-p}\exp\left(\frac{\delta}{1-p}t\right)\hat{g}^{\hat{Z}}(t)^{\frac{1}{p-1}}\hat{f}(t,\hat{Z})dt,\quad \hat{Z}_0=z_0^{\frac{1}{1-p}},
\end{align}
where, for $t\in[0,T]$,
\begin{align}\label{eq:hatghatf}
 \hat{g}^{\hat{Z}}(t)&:= \left( e^{b(t-T)}+ e^{bt} \int_{t}^{T} e^{-bs}\exp\left(\frac{\alpha p\delta}{1-p}s\right) (\hat{Z}_s)^{-\alpha p}ds \right)^{1-p},\nonumber\\
 \hat{f}(t,\hat{Z}) &:= x_0\exp\left(\frac{\mu^2}{(1-p)\sigma^2}t\right)\exp\left(-\int_0^t \exp\left(\frac{\alpha p\delta}{1-p}s\right) \hat{Z}_s^{-\alpha p}\hat{g}^{\hat{Z}}(s)^{\frac{1}{p-1}}ds\right).
\end{align}

Now, it is enough to study the well-posedness of \eqref{eq:odehatZt}. To do it, for any $(t,Z)\in[0,T]\times{\cal C}_T$, let us define
\begin{align}\label{eq:H}
\Phi(t,Z) &:= z_0^{\frac{1}{1-p}} + \int_0^t \frac{\delta}{1-p}\exp\left(\frac{\delta}{1-p}s\right)\hat{g}^{Z}(s)^{\frac{1}{p-1}}\hat{f}(s,Z)ds=z_0^{\frac{1}{1-p}} + \frac{\delta}{1-p}\int_0^t \phi(s,Z)ds,
\end{align}
where, for $\kappa:=\frac{\alpha p\delta}{1-p}$,
\begin{align}\label{eq:htZ}
\phi(t,Z) &:= \frac{x_0\exp\left(\frac{\mu^2+\sigma^2\delta}{(1-p)\sigma^2}t\right)}{\left(e^{b(t-T)}
+\int_t^T\frac{e^{b(t-s)+\kappa s}}{Z_s^{{\alpha}p}}ds\right)
\exp\left(\int_0^t\frac{ds}{Z_s^{{\alpha}p}e^{-\kappa s}(e^{b(s-T)}+\int_s^Te^{b(s-v)+\kappa v}Z_v^{-{\alpha}p}dv)}\right)}.
\end{align}
Recall that $z_0>\epsilon>0$ by the assumption $\bm{(A_{h})}$. Then, for any $Z\in{\cal C}_{T,\beta}$ with $\beta=\epsilon^{\frac{1}{1-p}}$, the mapping $\Phi(\cdot,Z)\in{\cal C}_T$. Moreover, as $\phi$ is positive, we deduce from \eqref{eq:H} and $p\in(0,1)$ that $\Phi(t,Z)\geq z_0^{\frac{1}{1-p}}\geq\beta$ for all $t\in[0,T]$. Hence, it holds that $\Phi(\cdot,Z)\in {\cal C}_{T,\beta}$.

Thus, for any $Z^1,Z^2\in{\cal C}_{T,\beta}$, we have from \eqref{eq:H} that
{\small
\begin{align}\label{eq:esti1}
&\left|\phi(t,Z^1)-\phi(t,Z^2)\right|\leq\frac{x_0\exp\left(\frac{\mu^2+\sigma^2\delta}{(1-p)\sigma^2}t\right)}{\exp(2b(t-T))}
\Bigg\{e^{bt}\exp\left(\int_0^t\frac{ds}{(Z_s^1)^{{\alpha}p}e^{b(s-T)-\kappa s}}\right)
\int_t^Te^{-bs+\kappa s}|(Z_s^1)^{-{\alpha}p}-(Z_s^2)^{-{\alpha}p}|ds\nonumber\\
&\qquad+\left(e^{b(t-T)}+e^{bt}\int_t^Te^{-bs+\kappa s}(Z_s^1)^{-{\alpha}p}ds\right)\Bigg|\exp\left(\int_0^t\frac{ds}{(Z_s^1)^{{\alpha}p}e^{-\kappa s}(e^{b(s-T)}
+\int_s^Te^{b(s-v)+\kappa v}(Z_v^1)^{-{\alpha}p}dv)}\right)\nonumber\\
&\qquad-\exp\left(\int_0^t\frac{ds}{(Z_s^2)^{{\alpha}p}e^{-\kappa s}(e^{b(s-T)}
+\int_s^Te^{b(s-v)+\kappa v}(Z_v^2)^{-{\alpha}p}dv)}\right)\Bigg|\Bigg\}.
\end{align}
}Note that $Z^1,Z^2\in{\cal C}_{T,\beta}$. Then $\min\{Z_t^1,Z_t^2\}\geq\beta$ for all $t\in[0,T]$. Using the mean-value theorem, it follows that $|(Z_s^1)^{-{\alpha}p}-(Z_s^2)^{-{\alpha}p}|= {\alpha}p\xi^{-(p+1)}|Z_s^1-Z_s^2|$, where $\xi\geq\min\{Z_s^1,Z_s^2\}\geq\beta$. Let $\|Z\|_T:=\sup_{t\in[0,T]}|Z_t|$ for any $Z\in{\cal C}_{T,\beta}\subset{\cal C}_T$. This yields that, for all $t\in[0,T]$,
\begin{align}\label{eq:esti100}
&e^{bt}\exp\left(\int_0^t\frac{ds}{(Z_s^1)^{{\alpha}p}e^{b(s-T)-\kappa s}}\right)\int_t^Te^{-bs+\kappa s}|(Z_s^1)^{-{\alpha}p}-(Z_s^2)^{-{\alpha}p}|ds\nonumber\\
&\qquad\leq e^{bt}\exp\left(\int_0^t\frac{ds}{\beta^{{\alpha}p}e^{b(s-T)-\kappa s}}\right)\int_t^T {\alpha}p\beta^{-(p+1)}e^{-bs+\kappa s}|Z_s^1-Z_s^2|ds\nonumber\\
&\qquad\leq \frac{e^{bt}}{\kappa-b}\left[e^{(\kappa-b)T-e^{(\kappa-b)t}}\right]{\alpha}p\beta^{-(p+1)}
\exp\left(\int_0^t\frac{ds}{\beta^{{\alpha}p}e^{b(s-T)-\kappa s}}\right)\left\|Z^1-Z^2\right\|_{T}.
\end{align}
On the other hand, it follows from the mean-value theorem and $p\in(0,1)$ again that
{\small
\begin{align}\label{eq:esti200}
&\Bigg|\exp\left(\int_0^t\frac{ds}{(Z_s^1)^{{\alpha}p}e^{-\kappa s}(e^{b(s-T)}
+\int_s^Te^{b(s-v)+\kappa v}(Z_v^1)^{-p}dv)}\right)\nonumber\\
&\qquad-\exp\left(\int_0^t\frac{ds}{(Z_s^2)^{{\alpha}p}e^{-\kappa s}(e^{b(s-T)}
+\int_s^Te^{b(s-v)+\kappa v}(Z_v^2)^{-{\alpha}p}dv)}\right)\Bigg|\nonumber\\
&\quad\leq \exp\left(\int_0^t\frac{ds}{\beta^{{\alpha}p}e^{b(s-T)-\kappa s}}\right)\\
&\times\left|\int_0^t\frac{ds}{(Z_s^1)^{{\alpha}p}e^{-\kappa s}(e^{b(s-T)}
+\int_s^Te^{b(s-v)+\kappa v}(Z_v^1)^{- {\alpha}p}dv)}-\int_0^t\frac{ds}{(Z_s^2)^{{\alpha}p}e^{-\kappa s}(e^{b(s-T)}
+\int_s^Te^{b(s-v)+\kappa v}(Z_v^2)^{-{\alpha}p}dv)}\right|\nonumber\\
&\quad\leq\exp\left(\int_0^t\frac{ds}{\beta^{{\alpha}p}e^{b(s-T)-\kappa s}}\right)
\left(\int_0^t\frac{|(Z_s^1)^{{\alpha}p}-(Z_s^2)^{{\alpha}p}|}{(Z_s^1)^{{\alpha} p}(Z_s^2)^{{\alpha}p}e^{2b(s-T)-2\kappa s}}ds
+\int_0^t\frac{\int_s^Te^{b(s-v)}|(Z_v^1)^{-{\alpha}p}-(Z_v^2)^{-{\alpha}p}|dv}{(Z_s^1)^{{\alpha}p}e^{2b(s-T)-2\kappa s}}ds\right)\nonumber\\
&\quad\leq\exp\left(\int_0^t\frac{ds}{\beta^{{\alpha}p}e^{b(s-T)-\kappa s}}\right)
\left(\int_0^t\frac{p\beta^{{\alpha}p-1}|Z_s^1-Z_s^2|}{\beta^{2{\alpha}p}e^{2b(s-T)-2\kappa s}}ds
+\int_0^t\frac{\int_s^Te^{b(s-v)}p\beta^{-{\alpha}p-1}|Z_v^1-Z_v^2|dv}{\beta^{{\alpha}p}e^{2b(s-T)-2\kappa s}}ds\right)\nonumber\\
&\quad\leq\exp\left(\int_0^t\frac{ds}{\beta^{{\alpha}p}e^{b(s-T)-\kappa s}}\right)\left(\int_0^t\frac{p\beta^{-{\alpha}p-1}}{e^{2b(s-T)-2\kappa s}}ds
+\int_0^t\frac{\int_s^Te^{b(s-v)}p\beta^{-2{\alpha}p-1}dv}{e^{2b(s-T)-2\kappa s}}ds\right) \left\|Z^1-Z^2\right\|_T.\nonumber
\end{align}}In view of \eqref{eq:esti1} with the estimates \eqref{eq:esti100} and \eqref{eq:esti200}, there exists a positive continuous function $T \mapsto C(T)$ independent of $z_0$ that satisfies $\lim_{T\to0}C(T)=0$ such that
\begin{align}\label{eq:contracmap0}
\left\|\phi(\cdot,Z^1)-\phi(\cdot,Z^2)\right\|_T\leq C(T)\left\|Z^1-Z^2\right\|_T,\quad \forall Z^1,Z^2\in{\cal C}_{T,\beta}.
\end{align}

Now, we rewrite the equation \eqref{eq:odehatZt} as a fixed point problem on $t\in[0,T]$ given by
\begin{align}\label{eq:ficedprob10}
\hat{Z}_t &= \Phi(t,\hat{Z}).
\end{align}
We then consider the problem \eqref{eq:ficedprob10} on a time interval $[t_0,t_1]$ with $t_0=0$ and $t_1\in(0,T]$. In light of \eqref{eq:contracmap0}, we may take $t_1$ small enough such that $C(t_1)\in(0,1)$, and hence $\Phi$ is a contraction map on $C_{t_1,\beta}$. Thus, there exists a unique fixed point of \eqref{eq:ficedprob10} on $[t_0,t_1]$. Note that $T  \mapsto C(T)$ in \eqref{eq:contracmap0} is independent of $z$. Then, we can apply this similar argument to conclude that there exists a unique fixed point of \eqref{eq:ficedprob10} on $[t_1,t_2]$ for some $t_2>t_1$ small enough. Repeating this procedure, we conclude the existence of a unique fixed point of \eqref{eq:ficedprob10} on $[0,T]$.
\end{proof}

\section{Numerical Illustrations of Mean Field Equilibrium}\label{sec:num}
To numerically illustrate and compare the mean field equilibrium (MFE) under two types of external habit formation, we consider a constant type vector $o=(\mu, \sigma, p)$ in the mean field model. From the main results in Lemmas \ref{lem:solHJBlimit-two} and \ref{lem:solHJBlimit} and Theorems \ref{lem:fixedpoint-two} and \ref{prop:fixedpoint}, we can see that MFE controls depend on model parameters in the complicated manner due to the structure of the fixed points $\bar{Z}^l$ and $\bar{Z}^m$. Some sensitivity results with respect to model parameters can only be concluded within some reasonable parameter regimes.

From Figures $1$-$3$, we see that the mean field habit formation process $\bar{Z}^l_t$ under the linear (addictive) habit formation is always an increasing function of time $t$ under different choices of parameters. Similarly, the feedback function of the MFE $C^l(t,x)=xc^l(t,x)$ is also increasing in time $t$ with an increasing slope (i.e., $d^2C^l(t,x)/dt^2>0$) under different choices of parameters. These observations indicate that the fierce competition induced by addictive habits may force each agent to consume more aggressively. When the wealth level is adequate, each agent would increase her consumption rate drastically especially when it is close to the terminal time, not only to obtain the higher excessive consumption to outperform the benchmark $\bar{Z}^l_t$ from the society, but also will strategically increase her own habit level such that the population's average habit level can be lifted even higher that may restrain other competitor's expected utility.

In contrast, under the multiplicative (non-addictive) habit formation, the mean field habit formation process $\bar{Z}^m_t$ and the MFE $C^m(t,x)=xc^m(t,x)$ instead exhibit diverse trends over time, sensitively depending on different choices of model parameters. 
In \autoref{fig:MFE-delta}, when the initial habit level $z_0$ is large and the initial wealth level $x_0$ is relatively low (recall that there is no constraint between $x_0$ and $z_0$), the mean field habit formation process $\bar{Z}^m_t$ can be first decreasing and then increasing in time. One might interpret this pattern that the average habit of the population, in the mean field equilibrium state, satisfies a type of mean-reverting mechanism. That is, when the habit level of the society is too high, the multiplicative habit formation preference often pulls down $\bar{Z}^m_t$ to a sustainable level and continues with another wave of growth, which can not be observed in the case under the linear habit formation. As for the MFE $C^m(t,x)$, more subtle trends can be observed over time, heavily relying on the representative agent's risk preference, the habit intensity, the competition parameter and other model parameters. 

Let us first illustrate in {\color{red}\autoref{fig:MFE-p}} the sensitivity results of the MFE on the risk aversion parameter $p$. In both cases, we choose and fix the model parameters $T=2$,~$\delta=0.1$, $\mu=0.2$, $\sigma=0.6$, $\alpha=1$ and take different values $p = 0.2,~ 0.5,~0.7$. For the linear habit formation, we further choose $x_0=5$,~$z_0=1$ and $x=5$; and for multiplicative habit formation, we choose $x_0=5$, $z_0=10$, $x=1$. We first note that both MFE portfolio $\pi^l(t,x)$ and $\pi^m(t,x)$ are increasing in $p$, which are similar to the Merton's solution that an individual investor allocates less wealth in the risky asset when she is more risk averse. These results are reasonable because our relative performance is purely measured by the excessive consumption with respect to the average external habit, and hence the equilibrium portfolio behaves similarly to the one in Merton's problem when the wealth level is adequate.

From the top panel, it is interesting to observe that when the habit formation process $\bar{Z}_t^l$ becomes reasonably large (after the accumulation over some time period), $\bar{Z}_t^l$ turns to be increasing in the parameter $p$, indicating that the more risk averse the agent is, the lower average habit of the population is attained at the terminal time. The same conclusion also holds for $\bar{Z}^m_t$. These results are reasonable because the average habit level has adverse effect in the expected utility, the larger risk aversion (smaller $p$) would lead to a lower mean field equilibrium habit level at the terminal time. We also see that both $C^l(t,x)$ and $C^m(t,x)$ are roughly decreasing in $p$, indicating that the more risk averse representative agent chooses higher MFE consumption plan. This is consistent with the intuition that the smaller $p$ value indicates that the representative agent is more risk averse towards the difference between the consumption rate and the benchmark habit formation, i.e. the representative agent feels more painful when the MFE consumption rate is close to or lower than the external habit level and hence consumes more aggressively. In particular, under the multiplicative habit formation, the representative agent behaves more aggressively and would increase the MFE consumption over time similar to the behavior driven by the addictive habit constraint in the case of linear habit formation (see the plot when $p=0.2$). On the other hand, when $p$ is close to $1$ (see the plot when $p=0.7$) and the representative agent is more risk neutral towards the distance between the consumption rate and the benchmark habit formation, under the multiplicative habit formation, the representative agent may strategically decrease the MFE consumption because the resulting habit formation process of the population is also decreasing.


Next, we illustrate in \autoref{fig:MFE-alpha} how the competition parameter $\alpha$ affects the MFE. We fix model parameters $T=2$, $p =0.5$, $x_0=3$, $\mu = 0.2$, $\delta=0.2$, $\sigma=0.6$, $x=1$ and consider different values $\alpha=0.2,~0.5, ~1$. We observe that all $\bar{Z}^l_t$, $C^l(t,x)$, $\bar{Z}^m_t$ and $C^m(t,x)$ are increasing in the parameter $\alpha$, indicating that the more competitive the representative agent is, the higher equilibrium consumption she chooses and the average habit level of the population also gets larger.

\begin{figure}[htbp]
	\centering
		\includegraphics[width=5.5in]{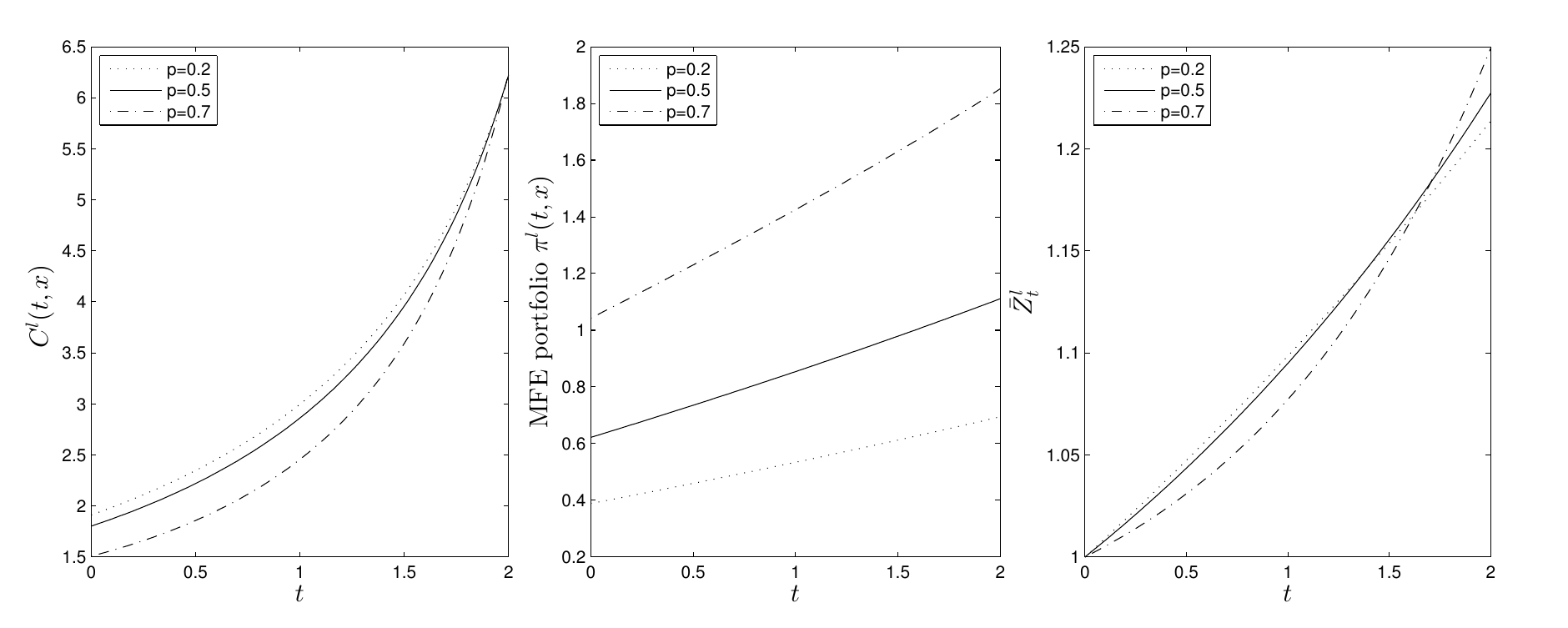}\\
		 \vspace{0.05cm}
		\includegraphics[width=5.5in]{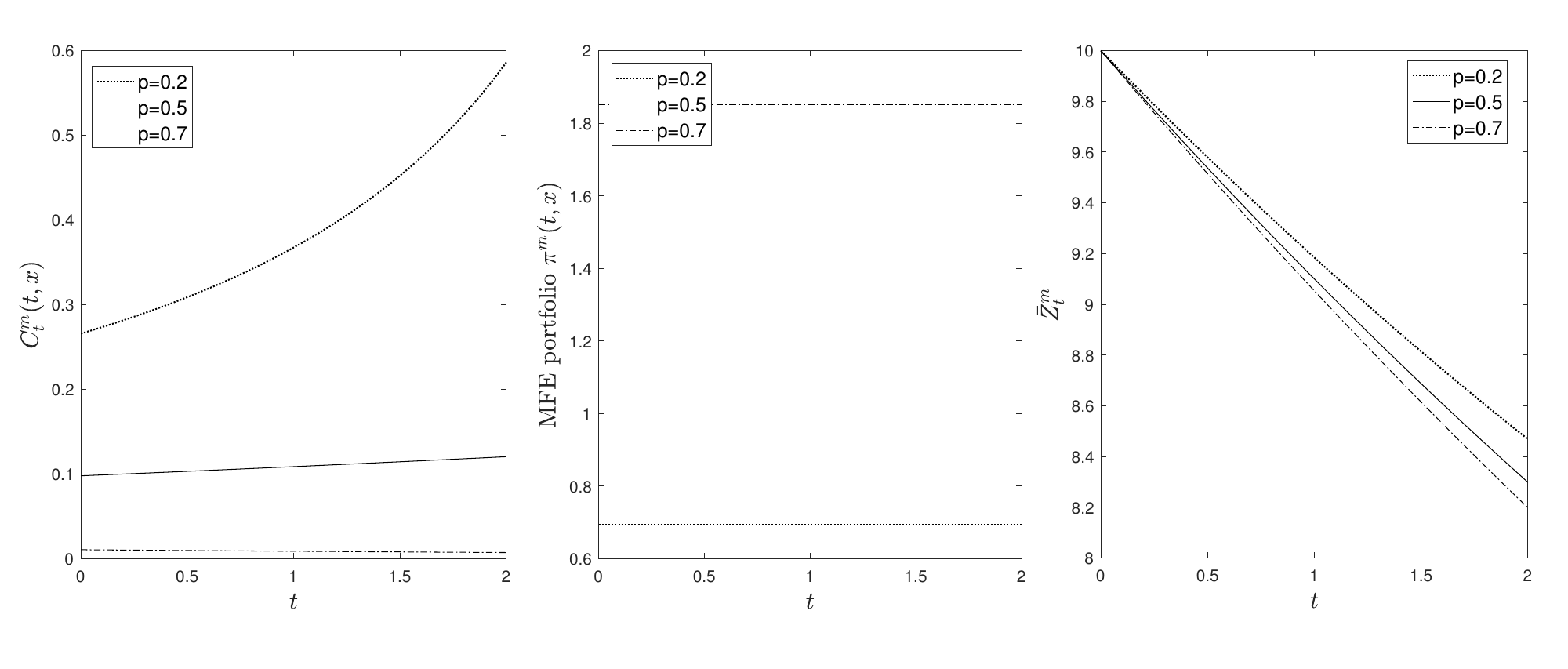}\\
	\centering
  \caption{
 {\small {\bf Top panel}: The MFE consumption rate $C^l(t,x)$, the MFE portfolio $\pi^{l}(t,x)$, and the habit formation process $\bar{Z}^l_t$ with risk aversion parameters $p=0.2,~0.5$ and $0.7$.
{\bf Bottom panel}:  The MFE consumption rate $C^m(t,x)$, the MFE portfolio $\pi^{m}(t,x)$, and the habit formation process $\bar{Z}^m_t$ with risk aversion parameters $p=0.2$, $0.5$ and $0.7$.   }}
 \label{fig:MFE-p}
\end{figure}

\begin{figure}[htbp]
\centering
\includegraphics[width=6 in]{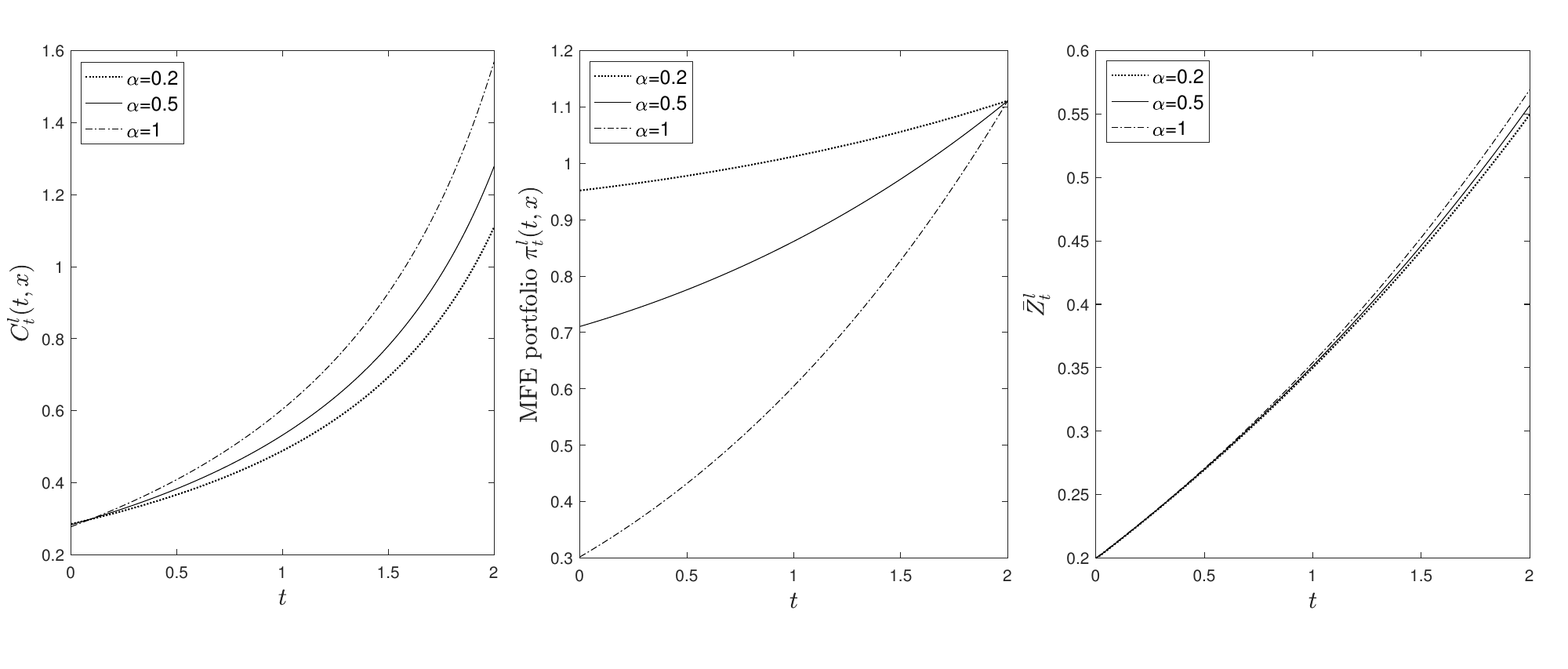}
\vspace{0.02cm}
\includegraphics[width=6 in]{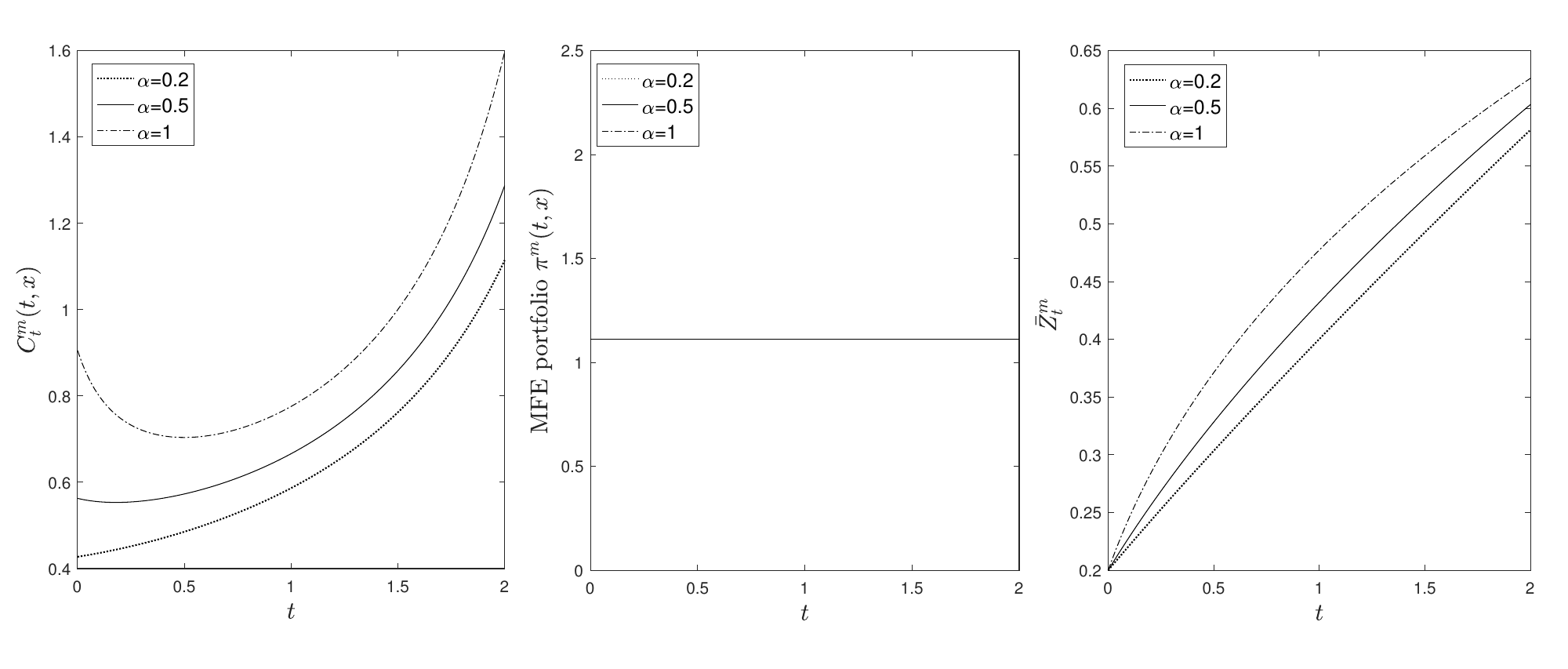}
\caption{
{\small{\bf Top panel}: The MFE consumption rate $C^l(t,x)$, the MFE portfolio $\pi^{l}(t,x)$, and the habit formation process $\bar{Z}^l_t$ with competition parameters $\alpha=0.2,~0.5$ and $1$.
 {\bf Bottom panel}:  The MFE consumption rate $C^m(t,x)$, the MFE portfolios $\pi^{m}(t,x)$, and the habit formation processes $\bar{Z}^m_t$ with competition parameters $\alpha=0.2,~0.5$ and $1$.
}}
\label{fig:MFE-alpha}
\end{figure}

\begin{figure}[htpb]
\centering
\includegraphics[width=5.5in]{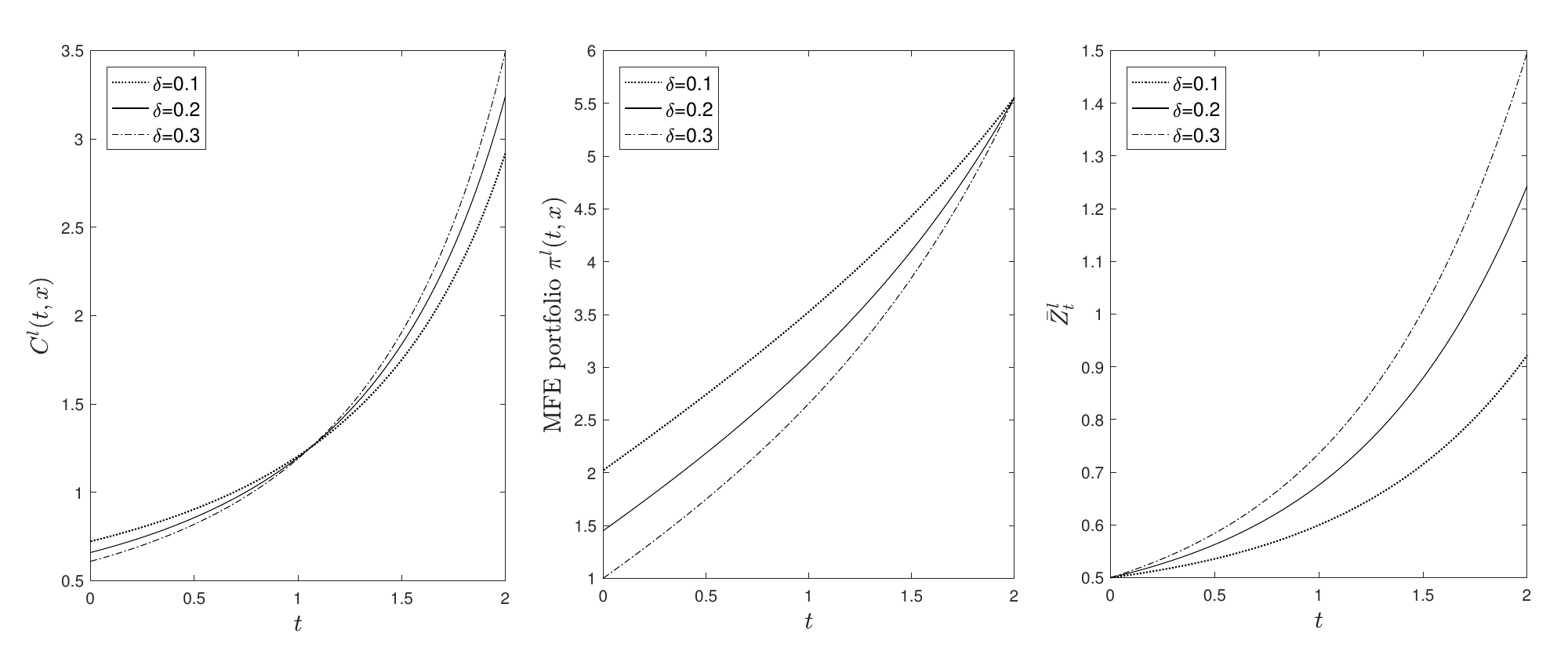}
\vspace{0.02cm}
\includegraphics[width=5.5in]{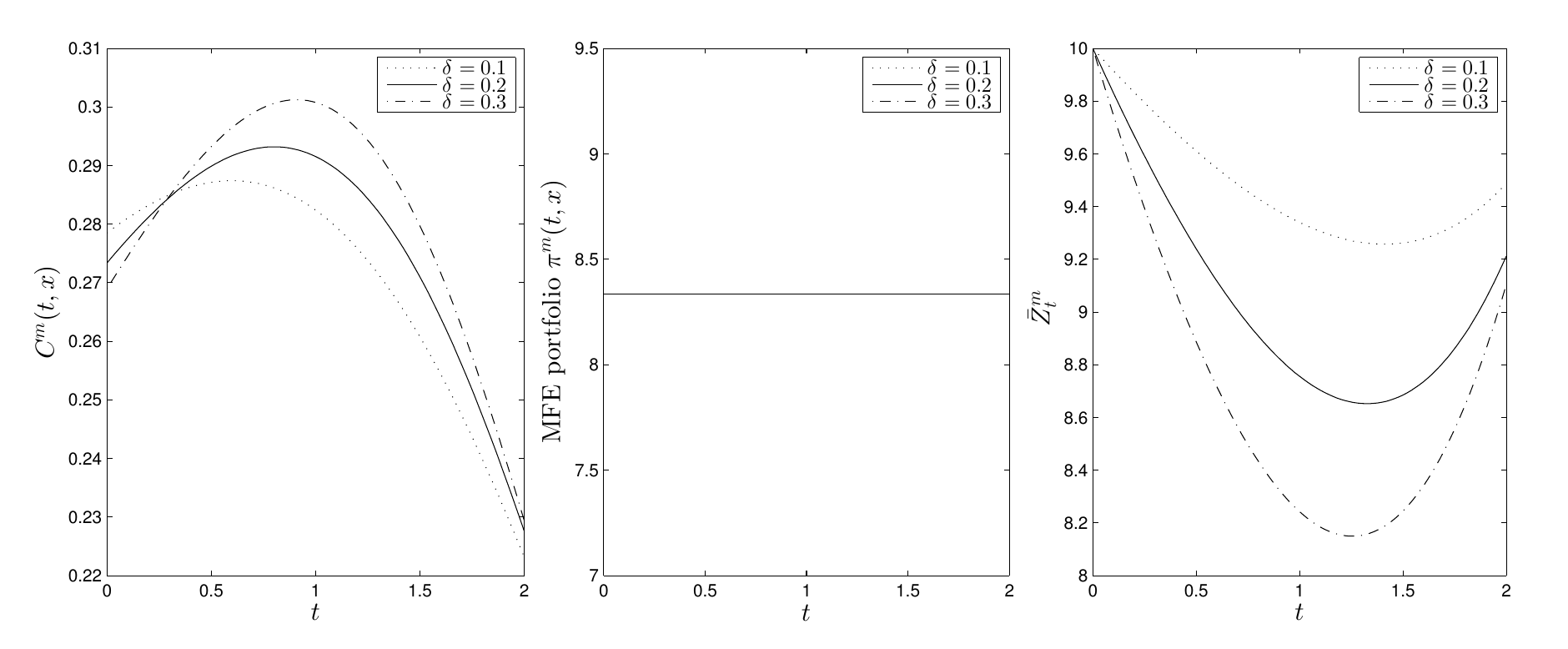}
\caption{
{\small {\bf Top panel}: The MFE consumption rate $C^l(t,x)$, the MFE portfolios $\pi^{l}(t,x)$, and the habit formation processes $\bar{Z}^l_t$ with habit persistence parameters $\delta = 0.1,~0.2$ and $0.3$.
{\bf Bottom panel}: The MFE consumption rate $C^m(t,x)$, the MFE portfolios $\pi^{m}(t,x)$, and the habit formation processes $\bar{Z}^m_t$ with habit persistence parameters $\delta=0.1,~0.2$ and $0.3$.
}}

\label{fig:MFE-delta}
\end{figure}

In \autoref{fig:MFE-delta}, we numerically illustrate the impact of the habit formation intensity parameter $\delta$. We choose and fix $T=2$ , $p=0.1$, $x_0=3$, $z_0=0.5$, $\mu=0.2$, $\sigma=0.2$, $x=2$, $\alpha=1$ under the linear habit formation with different values $\delta=0.1, ~0.2, ~0.3$, and choose and fix  $T=2$, $p=0.4$, $x_0=3$, $z_0=10$, $\mu=0.2$, $\sigma=0.2$, $x=1$ and $\alpha=1$ under the multiplicative habit formation with different values $\delta=0.1,~0.2,~0.3$. We observe that the MFE consumption rates $C^l(t,x)$ and $C^m(t,x)$ do not display monotonicity in $\delta$, which can be explained by the fact that the habit formation process of each agent is defined as the combination of the discounted initial value $z_0e^{-\delta t}$ (decreasing in $\delta$) and the weighted average $\int_0^t \delta e^{\delta (s-t)}C_sds$ (usually increasing in $\delta$), hence the dependence of $C^l(t,x)$ or $C^m(t,x)$ on the parameter $\delta$ is generally very subtle.

More importantly, we observe from the bottom-left panel of \autoref{fig:MFE-delta} that the representative agent displays the hump-shaped consumption pattern under some model parameters. That is, the consumption trajectory is first increasing in time and then decreasing in time, which matches with the well documented phenomenon in many empirical studies on individual's consumption behavior. \cite{Kraft} has proposed and verified the internal habit formation preference as an effective answer to support the hump-shaped consumption pattern of the individual agent. From the perspective of mean field competition, we can also give an explanation to the hump-shaped equilibrium consumption when the initial habit is very high. At the beginning of the time horizon, the population's average habit $\bar{Z}^m$ is high and decreasing in time $t$ due to the discount by $\delta$, and the competition mechanism takes the leading role in each agent's decision making. That is, each agent would increase her consumption rate to attain higher excessive consumption and larger expected utility. As time moves on, the past increasing consumption path takes more effect in each agent's habit formation process and the average habit level $\bar{Z}^m$ in the economy also starts to increase, indicating a high future standard of living in the society. As each agent has no obligation to consume above the high standard of living, the sense of competition is weakened and the self-satisfaction starts to take the leading role in decision making. Each agent is more likely to reduce her consumption rate after a period of time such that the growth rate of the population's habit level can start to slow down and everyone may suffer less from the high future benchmark.

\section{Approximate Nash Equilibrium in n-player Games}\label{sec:approx}
This section examines the approximate Nash equilibrium in finite
population games under two types of habit formation preferences when $n$ is sufficiently large. We construct the approximate Nash equilibrium using the obtained mean field equilibrium from the previous section. Let us recall two previous assumptions $\bm{(A_{h})}$ and $\bm{(A_{o})}$. In this section, in order to derive the explicit convergence rate in terms of $n$, we need to further assume that the mean field model is symmetric:

\begin{itemize}
\item[] $\bm{(A_{c})}$: there exists a constant vector $o=(\mu, \sigma, p)\in \mathcal{O} $ such that $ o_n \rightarrow o $ as $n\rightarrow \infty$ with the order of convergence $O(\frac{1}{\sqrt{n}})$.
\end{itemize}
Note that if the market model with $n$ agents is homogenous, i.e., the type vector $o_i=(\mu,\sigma,p)$ is symmetric for $i=1,\ldots, n$, the above assumption holds trivially. Similar assumption has been imposed in \cite{HuangNguyen2016} when they establish an approximate Nash equilibrium with an explicit convergence rate.

\subsection{Approximation under linear habit formation}

For $i=1,\ldots, n$, we recall that the objective functional \eqref{eq:Objective-i-two} of agent $i$ can be rewritten as: for $(\pi^i,c^i)\in{\cal A}^{l,i}(x_0)$,
\begin{align}\label{eq:Objective-i-n}
J^l_i((\pi^i,c^i),(\bm{\pi},\bm{c})^{-i}) = \Ex\left[\int_0^T  U_i\left(c^i_sX_s^i- \alpha\bar{Z}_s^n\right)ds + U_i(X_T^i) \right],
\end{align}
where the vector of policies $(\bm{\pi},\bm{c})^{-i}$ for $i=1,\ldots,n$ is defined by
\begin{align}\label{eq:policy-i}
(\bm{\pi},\bm{c})^{-i} :=( (\pi^1,c^1), \ldots, (\pi^{i-1},c^{i-1}), (\pi^{i+1},c^{i+1}),\ldots, (\pi^{n},c^{n})).
\end{align}
Then, the definition of an approximate Nash equilibrium for the model consisting of $n$ agents is defined as follows:
\begin{definition}[Approximate Nash equilibrium]\label{def:nashequilibrium-linear}
Let ${\cal A}^l(x_0):=\prod_{i=1}^n{\cal A}^{l,i}(x_0)$. An admissible strategy $(\bm{\pi}^{*,l},\bm{c}^{*,l}) =( (\pi^{*,l,1},c^{*,l,1}),\ldots, (\pi^{*,l,n},c^{*,l,n}) )\in{\cal A}^l(x_0)$ is called an $\epsilon$-Nash equilibrium to the $n$-player game problem \eqref{eq:Objective-i-two}  if, for any $((\pi^{i},c^{i}))_{i=1}^n \in{\cal A}^l(x_0)$, it holds that
\begin{equation}\label{eq:def-epsilon-Nash-linear}
   \sup_{(\pi^i,c^i)\in{\cal A}^{l,i}(x_0)} J^l_i((\pi^{i},c^{i}) , (\bm{\pi}^{*,l},\bm{c}^{*,l})^{-i})  \leq  J^l_i( (\bm{\pi}^{*,l},\bm{c}^{*,l}))+\epsilon,\quad\forall i=1,\ldots,n.
\end{equation}
\end{definition}

In what follows, we plan to construct and verify the closed-loop approximate Nash equilibrium for the $n$-player game. However, we note that the addictive habit constraint $C_t^i\geq \alpha\bar{Z}^n_t$, $t\in[0,T]$, needs to be guaranteed. Let us define $\overline{m}_p:=\sup_{i\in\N}p_i$ and $\underline{m}_p:=\inf_{i\in\N}p_i$. It is clear that $0<\underline{m}_p\leq\overline{m}_p<1$. For $i=1,\ldots,n$, we now give a careful construction of the candidate investment and consumption pair $(\pi^{*,l,i}, c^{*,l,i}) = (\pi_t^{*,l,i}, c_t^{*,l,i})_{t\in[0,T]}$ for agent $i$ in the following form
\begin{equation}\label{eq:pi-c-star-i-two}
\left\{
  \begin{aligned}
   & \pi_t^{*,l,i} :=  \frac{\mu_i}{(1-p_i)\sigma_i^2X_t^{*,l,i}}\left(X_t^{\bar{Z}^{l},i}-\int_{t}^{T}\alpha \bar{Z}^{l}_sds\right), \\[0.4em]
   & c_t^{*,l,i} := \frac{1}{X_t^{*,l,i}}\left(\alpha\bar{Z}_t^{*,l, n} +\left(X_t^{\bar{Z}^{l},i}-\int_{t}^{T}\alpha\bar{Z}^{l}_s ds\right)g^l_i(t)^{\frac{1}{p_i-1}}\right),
\end{aligned}
\right.
\end{equation}
where $X^{*,l,i}=(X_t^{*,l,i})_{t\in[0,T]}$ is the wealth process defined in \eqref{eq:X-i} for agent $i$ under the control $(\pi^{*,l,i}, c^{*,l,i})$ that
\begin{equation}\label{eq:X-star-i-two}
\frac{dX_t^{*,l,i}}{X_t^{*,l,i}} = \pi_t^{*,l,i}\mu_idt +\pi_t^{*,l,i}\sigma_i dW^i_t - c_t^{*,l,i}dt,\quad X_0^{*,l,i}=x_0,
\end{equation}
and the average habit formation process from $n$ agents is defined by
\begin{equation}\label{eq:averageZ-star}
  \bar{Z}_t^{*,l,n}:= \frac{1}{n}\sum_{i=1}^{n} Z_t^{*,l,i} =  e^{-\delta t} z_0+  e^{-\delta t} \int_{0}^{t} \delta e^{\delta s }\frac{1}{n}\sum_{i=1}^{n} c_s^{*,l, i}X_s^{*,l,i} ds,
\end{equation}
with habit formation process of $i$ agent satisfying
\begin{equation}\label{eq:Z-star-i-2}
  Z_t^{*,l,i} = e^{-\delta t} \left( z_0+ \int_{0}^{t}\delta e^{\delta s }c_s^{*,l, i}X_s^{*,l,i} ds \right).
\end{equation}
The positive function $t  \mapsto  g^l_i(t)$ is given by
\begin{equation}\label{eq:solution-g-i-two}
  g_i^l(t)
  := \left[\left(1+\frac{1}{a_i}\right)e^{a_i(T-t)}-\frac{1}{a_i} \right]^{1-p_i},
\end{equation}
Here $a_i = \frac{p_i\mu_i^2}{2(1-p_i)^2\sigma_i^2} $.
Note that $\bar{Z}^{l}=(\bar{Z}^{l}_t)_{t\in[0,T]}$ in \eqref{eq:pi-c-star-i-two} differs from the average habit formation process $\bar{Z}^{*,l,n}$ in the $n$-player game. Indeed, $\bar{Z}^{l}$ is the unique fixed point established in Proposition~\ref{lem:fixedpoint-two} in the mean field game problem, and the auxiliary process $X^{\bar{Z}^{l},i} = (X_t^{\bar{Z}^{l},i})_{t\in[0,T]}$ in \eqref{eq:pi-c-star-i-two}  is defined by the SDE
\begin{equation}\label{eq:X-starbar-i}
\left\{
\begin{aligned}
 &\frac{d(X_t^{\bar{Z}^{l},i} -\int_{t}^{T}\alpha \bar{Z}^{l}_sds)} {X_t^{\bar{Z}^{l},i}- \int_{t}^{T}\alpha \bar{Z}^{l}_sds }=  \left(
 \frac{\mu_i^2}{(1-p_i)\sigma_i^2} -g^l_i(t)^{\frac{1}{p_i-1}}\right)dt + \frac{\mu_i}{(1-p_i)\sigma_i}dW^i_t, \\
 & X_0^{\bar{Z}^{l},i}  = x_0.
\end{aligned}
\right.
\end{equation}
It is clear that $(X_t^{\bar{Z}^{l},i} -\int_{t}^{T}\alpha \bar{Z}^{l}_sds)_{t\in[0,T]}$ is a GBM for each $i$ that is nonnegative in view of Lemma~\ref{lem:fixedpoint-two}. Moreover, we can write
\begin{align}\label{eq:XbarZiintbarZs}
X_t^{\bar{Z}^{l},i}-\int_{t}^{T}\alpha \bar{Z}_s^{l}ds= \left(x_0-\int_0^T \alpha \bar{Z}^{l}_sds\right) \exp\left(\int_{0}^{t}G_i(s)ds +\frac{\mu_i}{(1-p_i)\sigma_i}W^i_t\right)
\end{align}
with $G_i(t) := \frac{\mu_i^2}{(1-p_i)\sigma_i^2} -g^l_i(t)^{\frac{1}{p_i-1}} -\frac{\mu_i^2}{2(1-p_i)^2\sigma_i^2}$ for all $t\in[0,T]$. As $X_t^{\bar{Z}^{l},i}>\int_{t}^{T}\alpha\bar{Z}_s^{l}ds$ and $g_i^l(t)>0$ for $t\in[0,T]$, by definition in \eqref{eq:pi-c-star-i-two}, it is clear that the addictive habit constraint is satisfied. In fact, we can show in the next result that $C^{*,l,i}$ can be fully expressed by the given $\bar{Z}^l$. Let us denote the average aggregate consumption rate
\begin{equation}\label{eq:average-star-2}
\bar{C}_t^{*,l,n} := \frac{1}{n}\sum_{i=1}^{n} C_t^{*,l,i} = \frac{1}{n}\sum_{i=1}^{n} c_t^{*,l,i}X_t^{*,l,i}.
\end{equation}
We also recall that $d \bar{Z}_t^{*,l,n} = -\delta( \bar{Z}_t^{*,l,n} -\bar{C}_t^{*,l,n})dt$.
\begin{lemma}\label{lem:CbarnZbarnlinear}
Let $\bar{Z}^{l}=(\bar{Z}_t^{l})_{t\in[0,T]}\in{\cal C}_{T, x_0}$ be the unique fixed point in Proposition~\ref{lem:fixedpoint-two}. Then, we can express the average consumption rate $\bar{C}_t^{*,l,n} $ and the average habit formation process $\bar{Z}_t^{*,l,n} $ explicitly in terms of $\bar{Z}^{l}$ that
\begin{equation}\label{eq:ODE-barZn}
\left\{
\begin{aligned}
 &\bar{C}_t^{*,l,n}= \alpha\bar{Z}_t^{*,l,n} + \left(x_0-\int_0^T\alpha\bar{Z}^{l}_sds\right) \frac{1}{n}\sum_{i=1}^{n}\exp\left(\int_{0}^{t}G_i(s)ds +\frac{\mu_i}{(1-p_i)\sigma_i}W^i_t\right) g_i(t)^{\frac{1}{p_i-1}}, \\
 &\bar{Z}_t^{*,l,n}= e^{\delta(\alpha-1)t} z_0+ \left(x_0-\int_0^T\alpha\bar{Z}^{l}_sds\right)\int_0^t e^{\delta(\alpha-1)(s-t)} \frac{\delta}{n}\sum_{i=1}^{n}\exp\left(\int_{0}^{s}G_i(v)dv  \right.\\
 &\qquad \quad  \left. +\frac{\mu_i}{(1-p_i)\sigma_i}W^i_s\right)g^l_i(s)^{\frac{1}{p_i-1}} ds.
\end{aligned}
\right.
\end{equation}
Moreover, for any $q\geq1$, it holds that $\sup_{t\in[0,T]}\Ex[(\bar{Z}_t^{*,l,n})^q]\leq C_{q,T}$ for some constant $C_{q,T}>0$ that is independent of $n$.
\end{lemma}

\begin{proof}
In view of \eqref{eq:pi-c-star-i-two} and \eqref{eq:average-star-2}, we have that
\begin{align}\label{eq:Cbartstranrep0}
 \bar{C}_t^{*,l,n} &= \alpha\bar{Z}_t^{*,l,n} + \frac{1}{n}\sum_{i=1}^{n}\left(X_t^{\bar{Z}^l,i}-\int_{t}^{T}\alpha\bar{Z}^l_s ds \right) g^l_i(t)^{\frac{1}{p_i-1}}.
\end{align}
Then, it holds from \eqref{eq:Cbartstranrep0} that
\begin{align}\label{eq:ODE-Z-n}
d \bar{Z}_t^{*,l,n} &= -\delta( \bar{Z}_t^{*,l,n} -\bar{C}_t^{*,l,n})dt = \delta(\alpha-1)\bar{Z}_t^{*,l,n}dt +  \frac{\delta}{n}\sum_{i=1}^{n}\left(X_t^{\bar{Z}^l,i}-\int_{t}^{T}\alpha\bar{Z}^l_s ds \right) g^l_i(t)^{\frac{1}{p_i-1}}dt.
\end{align}

The first equality in \eqref{eq:ODE-barZn} follows from \eqref{eq:Cbartstranrep0} and \eqref{eq:XbarZiintbarZs}, while the second equality results from \eqref{eq:ODE-Z-n} and \eqref{eq:XbarZiintbarZs}.
By using the second equality in \eqref{eq:ODE-barZn}, it follows from Jensen's inequality that
\begin{align*}
 &\sup_{t\in[0,T]}\Ex\left[ \left(\bar{Z}_t^{*,l,n}\right)^q \right]
  = \sup_{t\in[0,T]}\Ex\left[\left(e^{\delta(\alpha-1)t}z_0+ \left(x_0-\int_0^T\alpha\bar{Z}^l_sds\right)\int_0^t e^{\delta(\alpha-1)(t-s)}   \right.\right. \\
 &\qquad\qquad\qquad \qquad\qquad \left.\left. \times \frac{\delta}{n}\sum_{i=1}^{n}\exp\left(\int_{0}^{s}G_i(v)dv +\frac{\mu_i}{(1-p_i)\sigma_i}W^i_s\right)g^l_i(s)^{\frac{1}{p_i-1}} ds\right)^q\right]\nonumber\\
 & \quad\leq C_q z_0^q+ C_{q,T}\left(x_0-\int_0^{T}\alpha\bar{Z}^l_sds\right)^q \int_{0}^{T}\frac{1}{n}\sum_{i=1}^{n}\Ex\left[\exp\left(\int_{0}^{s}qG_i(v)dv +\frac{q\mu_i}{(1-p_i)\sigma_i}W^i_s\right)\right] ds  \nonumber \\
 &\quad  =C_q z_0^q+ C_{q,T}\left(x_0-\int_0^{T}\alpha\bar{Z}^l_sds\right)^q\int_0^T \frac{1}{n}\sum_{i=1}^{n} \exp\left(\int_{0}^{s}qG_i(v)dv+\frac{q^2\mu_i^2}{2(1-p_i)^2\sigma_i^2}s\right)ds.
\end{align*}
Then, the desired estimate on $\sup_{t\in[0,T]}\Ex[(\bar{Z}_t^{*,l,n})^q]$ follows from the assumption $\bm{(A_{c})}$. Thus, we complete the proof of the lemma.
\end{proof}

We now present the main result of this subsection on the existence of an approximate Nash equilibrium under linear external habit formation.
\begin{theorem}\label{thm:epsilon-Nash-equilibrium-two}
The control pair $(\bm{\pi}^{\ast,l},\bm{c}^{\ast,l})=((\pi^{*,l,1}, c_t^{*,l,1}),\ldots, (\pi^{*,l,n},c_t^{*,l,n}))_{t\in[0,T]}$ given in \eqref{eq:pi-c-star-i-two} is an $\epsilon_n$-Nash equilibrium for the $n$-player game problem \eqref{eq:Objective-i-two} with the explicit order $\epsilon_n= O(n^{-\frac{\underline{m}_p}{2}})$.
\end{theorem}
To prove Theorem~\ref{thm:epsilon-Nash-equilibrium-two}, we need the following auxiliary results. 

\begin{lemma}\label{lem:momentestiXbarZi}
Let $\bar{Z}^{l}=(\bar{Z}_t^{l})_{t\in[0,T]}\in{\cal C}_{T, x_0}$ be the unique fixed point in Proposition~\ref{lem:fixedpoint-two}. Then, for any $q>0$, there exists a constant $C_q>0$ independent of $i$ such that
\begin{equation}\label{eq:lem-X-star-E-two}
 \sup_{t\in[0,T]}\Ex\left[\left(X_t^{\bar{Z}^{l},i}- \int_{t}^{T}\alpha\bar{Z}_s^{l}ds\right)^q\right] \leq C_{q}.
\end{equation}
\end{lemma}
\begin{proof}
Recall the dynamics of $X_t^{\bar{Z}^l,i}$ satisfying
\begin{align*}
 X_t^{\bar{Z}^{l},i}-\int_{t}^{T}\alpha \bar{Z}_s^{l}ds = \left(x_0-\int_0^T \alpha \bar{Z}^{l}_sds\right) \exp\left(\int_{0}^{t}G_i(s)ds +\frac{\mu_i}{(1-p_i)\sigma_i}W^i_t\right)
\end{align*}
with $G_i(t) := \frac{\mu_i^2}{(1-p_i)\sigma_i^2} -g^l_i(t)^{\frac{1}{p_i-1}} -\frac{\mu_i^2}{2(1-p_i)^2\sigma_i^2}$ for all $t\in[0,T]$. It clearly holds that $X_t^{\bar{Z}^l,i} -\int_{t}^{T}\alpha\bar{Z}_s^lds >0 $ for any $t\in[0,T]$, as $ x_0 >\int_{0}^{T}\alpha\bar{Z}_s^lds$. It follows that for $q >0$
\begin{align}\label{eq:Ex-Y-star-i}
 \Ex\left[\left( X_t^{\bar{Z}^{l},i}-\int_{t}^{T}\alpha \bar{Z}_s^{l}ds \right)^q\right]
 & = \left(x_0 - \int_{0}^{T}\alpha\bar{Z}_s^lds\right)^q \Ex\left[\exp \left\{q \int_{0}^{t}G_i(s)ds + \frac{q \mu_i}{(1-p_i)\sigma_i} W_t^i \right\}\right]      \nonumber \\
  & = \left(x_0-\int_{0}^{T}\alpha\bar{Z}_s^lds\right)^q  \exp \left\{ \int_{0}^{t} \left(qG_i(s)+ \frac{1}{2}\frac{q^2\mu_i^2}{(1-p_i)^2\sigma_i^2} \right) ds  \right\}.
\end{align}
It then follows from \eqref{eq:solution-g-i-two} that $g_i^l(t) >0 $ for $\forall t\in[0,T]$, which implies that for  all $t\in[0,T]$
\begin{equation*}
  G_i(t)= \frac{\mu_i^2}{(1-p_i)\sigma_i^2} -g^l_i(t)^{\frac{1}{p_i-1}} -\frac{\mu_i^2}{2(1-p_i)^2\sigma_i^2} < \frac{\mu_i^2}{(1-p_i)\sigma_i^2}.
\end{equation*}
Thus, \eqref{eq:Ex-Y-star-i} implies that
\begin{align*}
\Ex\left[\left( X_t^{\bar{Z}^{l},i}-\int_{t}^{T}\alpha \bar{Z}_s^{l}ds \right)^q\right]
 & = \left(x_0-\int_{0}^{T}\alpha\bar{Z}_s^lds\right)^q  \exp \left\{ \int_{0}^{t} \left(qG_i(s)+ \frac{1}{2}\frac{q^2\mu_i^2}{(1-p_i)^2\sigma_i^2} \right) ds  \right\}   \\
 & \leq \left(x_0- \int_{0}^{T}\alpha\bar{Z}_s^lds\right)^q \exp \left\{\int_{0}^{t}\left(\frac{q\mu_i^2}{(1-p_i)\sigma_i^2} + \frac{1}{2}\frac{q^2\mu_i^2}{(1-p_i)^2\sigma_i^2} \right)ds \right\}   \\
 & \leq \left(x_0- \int_{0}^{T}\alpha\bar{Z}_s^lds\right)^q \exp \left\{\left(\frac{q\mu_i^2}{(1-p_i)\sigma_i^2} +\frac{1}{2}\frac{q^2\mu_i^2}{(1-p_i)^2\sigma_i^2} \right)T \right\}   \\
 &\leq x_0^q \exp \left\{\left(\frac{q\mu_i^2}{(1-p_i)\sigma_i^2} +\frac{1}{2}\frac{q^2\mu_i^2}{(1-p_i)^2\sigma_i^2} \right)T \right\}  \\
 & := C_{q},
\end{align*}
which yields \eqref{eq:lem-X-star-E-two}.

\end{proof}

The next result follows from Lemmas~\ref{lem:momentestiXbarZi} and \ref{lem:CbarnZbarnlinear}.
\begin{lemma}\label{lem:estimZtnilinear}
For $q\geq 1$, there exists a $C_q>0$ independent of $i$ such that $\sup_{t\in[0,T]}\Ex[(Z_t^{*,l,i})^q] \leq C_{q}$. Moreover, it holds that
       \begin{align}\label{eq:lem-Z-n-star-two}
   \sup_{t\in[0,T]}\Ex\left[\left|\bar{Z}_t^{*,l,n} -\bar{Z}^l_t\right|^2 \right] = O\left(n^{-1}\right).
 \end{align}
\end{lemma}
\begin{proof} By \eqref{eq:averageZ-star} and \eqref{eq:fixed-point-2}, we have that
\begin{align*}
  & \bar{Z}_t^{*,l,n} - \bar{Z}^l_t=  e^{-\delta t} \int_{0}^{t}\delta e^{\delta s} \left(\frac{1}{n}\sum_{i=1}^{n}c_s^{*,l,i}X_s^{*,l,i} - \Ex[c_s^l X_s^{l,\bar{Z}^l}]\right)ds \\
  &\quad = e^{-\delta t}\int_{0}^{t}\delta e^{\delta s} \left\{\frac{1}{n}\sum_{i=1}^{n}\left(\alpha\bar{Z}_s^{*,l,n} + \left( X_s^{\bar{Z}^l,i} - \int_{s}^{T}\alpha\bar{Z}^l_udu\right)g^l_i(s)^{\frac{1}{p_i-1}}\right)  \right.\\
  &\quad \qquad \left.- \alpha\bar{Z}^l_s - \Ex\left[\left( X_s^{l,\bar{Z}^l} - \int_{s}^{T}\alpha\bar{Z}_u^ldu\right)g^l(s)^{\frac{1}{p-1}} \right] \right\}ds \\
  &\quad  =  \int_{0}^{t}\alpha\delta e^{\delta(s-t)}\left(\bar{Z}_s^{*,l,n} -\bar{Z}^l_s\right)ds  +\int_{0}^{t}\delta e^{\delta(s-t)} \left[\frac{1}{n}\sum_{i=1}^{n} X_s^{\bar{Z}^l,i} - \Ex\left[X_s^{l,\bar{Z}^l}\right] \right]g^l(s)^{\frac{1}{p-1}}ds \\
   &\qquad \quad+ \int_{0}^{t}\delta e^{\delta (s-t)} \left[\frac{1}{n}\sum_{i=1}^{n}\left( X_s^{\bar{Z}^l,i} - \int_{s}^{T}\alpha\bar{Z}^l_udu\right)\left(g^l_i(s)^{\frac{1}{p_i-1}} - g^l(s)^{\frac{1}{p-1}} \right) \right]ds.
\end{align*}
Then, by applying Cauchy inequality, it holds that
\begin{align*}
 \left|\bar{Z}_t^{*,l,n} -\bar{Z}^l_t\right|^2 & \leq
  4\left(\int_{0}^{t}\delta\alpha\left|\bar{Z}_s^{*,l,n}-\bar{Z}^l_s\right|ds\right)^2  + 4\left(\int_{0}^{t}\delta\left|\frac{1}{n}\sum_{i=1}^{n}X_s^{\bar{Z}^l,i}- \Ex\left[X_s^{l,\bar{Z}^l}\right]\right|ds \right)^2 \\
  &\quad + 4\left(\int_{0}^{t}\delta e^{\delta(s-t)}\frac{1}{n}\sum_{i=1}^{n} \left|g^l_i(s)^{\frac{1}{p_i-1}} -g^l(s)^{\frac{1}{p-1}}\right|\left(X_s^{\bar{Z}^l,i}-\int_{s}^{T}\bar{Z}^l_u du\right)ds\right)^2    \\
  &\leq  4\delta^2 T\int_{0}^{t}\left|\bar{Z}_s^{*,l,n}-\bar{Z}^l_s\right|^2ds + 4\delta^2T\int_{0}^{t}\left|\frac{1}{n}\sum_{i=1}^{n} X_s^{\bar{Z}^l,i}-\Ex\left[X_s^{l,\bar{Z}^l}\right]\right|^2ds \\
  &\quad + 4\delta^2 T \int_{0}^{t} \left(\frac{1}{n}\sum_{i=1}^{n}\left|g^l_i(s)^{\frac{1}{p_i-1}} -g^l(s)^{\frac{1}{p-1}}\right| \left(X_s^{\bar{Z}^l,i}-\int_{s}^{T}\alpha\bar{Z}^l_u du\right) \right)^2ds.
 \end{align*}
There exists a constant $C_T>0$ depending on $T$ only, which may be different from line to line that
\begin{align}\label{eq:Ex-Z-n-Z-two}
  & \mathbb{E}\left[\left|\bar{Z}_t^{*,l,n} -\bar{Z}^l_t\right|^2 \right]    \nonumber \\
  &\qquad\leq  C_T\int_{0}^{t} \Ex\left[\left|\bar{Z}_s^{*,l,n}-\bar{Z}^l_s\right|^2\right]ds  + C_T \int_{0}^{t}\Ex\left[\left|\frac{1}{n}\sum_{i=1}^{n}X_s^{\bar{Z}^l,i}- \Ex\left[X_s^{l,\bar{Z}^l}\right] \right|^2 \right]ds  \nonumber \\
  &\qquad \quad + C_T \int_{0}^{t}\left(\frac{1}{n}\sum_{i=1}^{n} \left(g^l_i(s)^{\frac{1}{p_i-1}}-g^l(s)^{\frac{1}{p-1}}\right)^{2}\right) \Ex\left[\frac{1}{n}\sum_{i=1}^{n}\left(X_s^{\bar{Z}^l,i}- \int_{s}^{T}\alpha\bar{Z}^l_u du\right)\right]ds \nonumber\\
  &\qquad\leq  C_T \int_{0}^{t}\Ex\left[\left|\frac{1}{n}\sum_{i=1}^{n} X_s^{\bar{Z}^l,i} - \Ex\left[X_s^{l,\bar{Z}^l}\right]\right|^2 \right]ds \nonumber \\
  &\qquad\quad + C_T \int_{0}^{t} \frac{1}{n} \sum_{i=1}^{n}\left(g^l_i(s)^{\frac{1}{p_i-1}} -g^l(s)^{\frac{1}{p-1}}\right)^{2} ds
   + C_T\int_{0}^{t}\Ex\left[\left|\bar{Z}_s^{*,l,n}-\bar{Z}^l_s\right|^2\right]ds \nonumber \\
 &\qquad:= L_1^{(n)}(t) + L_2^{(n)}(t) + C_T \int_{0}^{t} \Ex\left[\left| \bar{Z}_s^{*,l,n}-\bar{Z}^l_s\right|^2 \right]ds.
\end{align}
For the estimate of $L_1^{(n)}(t)$, we introduce the following SDE, for $i=1,\ldots,n$,
\begin{equation}\label{eq:X-tilde-i}
 \frac{d\tilde{X}_t^i}{\tilde{X}_t^i} = \pi^{l}(t,\tilde{X}_t^i)\mu dt +\pi^{l}(t,\tilde{X}_t^i)\sigma dW_t^i - c^{l}(t,\tilde{X}_t^i)dt,\quad \tilde{X}_0^i=x_0,
\end{equation}
where $\pi^l(t,x)=\pi^{*,\bar{Z}^l}(t,x)$ and $c^l(t,x)=c^{*,\bar{Z}^l}(t,x)$ are defined in \eqref{eq:picstar-2} as the (feedback) mean field equilibrium with the given $\bar{Z}^l$. We can rewrite \eqref{eq:X-tilde-i} as 
\begin{equation*}
\left\{
\begin{aligned}
 &\frac{d(\tilde{X}_t^i -\int_{t}^{T}\alpha \bar{Z}^{l}_sds)} {\tilde{X}_t^i- \int_{t}^{T}\alpha \bar{Z}^{l}_sds }=  \left(
 \frac{\mu^2}{(1-p)\sigma^2} -g^l(t)^{\frac{1}{p-1}}\right)dt + \frac{\mu}{(1-p)\sigma}dW^i_t, \\
 & \tilde{X}_0^i  = x_0.
\end{aligned}
\right.
\end{equation*}
Hence, $(\tilde{X}^i)_{i=1}^n$ are i.i.d. such that $\Ex[\tilde{X}_t^i]=\Ex[X_t^{l,\bar{Z}^l}]$ for all $i=1,\ldots,n$. Here, we recall that $X^{l,\bar{Z}^l}=(X_t^{l,\bar{Z}^l})_{t\in[0,T]}$ satisfies the dynamics
\begin{equation}\label{eq:wealth-Xstar}
\frac{dX_t^{l,\bar{Z}^l}}{X_t^{l,\bar{Z}^l}}=\pi^l(t,X_t^{l,\bar{Z}^l})\mu dt+\pi^l(t,X_t^{l,\bar{Z}^l})\sigma dW_t - c^{l}(t,X_t^{l,\bar{Z}^l})dt,\quad X_0^{l,\bar{Z}^l}=x_0.
\end{equation}
It follows from Jensen's inequality that
\begin{align}\label{eq:diffXXstar00-two}
\mathbb{E}\left[\left|\frac{1}{n}\sum_{i=1}^{n}X_t^{\bar{Z}^l,i}-\mathbb{E}[X_t^{l,\bar{Z}^l}]\right|^2\right]&\leq 2 \mathbb{E}\left[\left|\frac{1}{n}\sum_{i=1}^{n}X_t^{\bar{Z}^l,i} - \frac{1}{n}\sum_{i=1}^{n}\tilde{X}_t^{i}\right|^2 \right] + 2 \mathbb{E}\left[\left|\frac{1}{n}\sum_{i=1}^{n}\tilde{X}_t^{i} - \mathbb{E}[X_t^{l,\bar{Z}^l}]\right|^2 \right]\nonumber\\
 &\leq \frac{2}{n}\sum_{i=1}^{n}\mathbb{E}\left[\left|X_t^{\bar{Z}^l,i} - \tilde{X}_t^i\right|^2 \right] + \frac{2}{n^2}\mathbb{E}\left[\left|\sum_{i=1}^{n} (\tilde{X}_t^i-\mathbb{E}[X_t^{l,\bar{Z}^l}])\right|^2 \right] \nonumber \\
 &= \frac{2}{n}\sum_{i=1}^{n}\mathbb{E}\left[\left|X_t^{\bar{Z}^l,i} - \tilde{X}_t^i\right|^2 \right] + \frac{2}{n}\Ex\left[\left|\tilde{X}_t^1-\mathbb{E}[X_t^{l,\bar{Z}^l}]\right|^2 \right].
\end{align}
Denote by $ \theta_i :=\frac{\mu_i}{(1-p_i)\sigma_i^2}$ and $\theta := \frac{\mu}{(1-p)\sigma^2}$. We deduce from \eqref{eq:X-starbar-i} and \eqref{eq:X-tilde-i} that
\begin{align*}
 & X_t^{\bar{Z}^l,i} - \tilde{X}_t^i \\
 & = \int_{0}^{t}\left[\left(\mu_i\theta_i- g^l_i(s)^{\frac{1}{p_i-1}}\right) \left(X_s^{\bar{Z}^l,i}-\tilde{X}_s^i \right) + \left(\mu_i\theta_i - g^l_i(s)^{\frac{1}{p_i-1}} - \mu\theta + g^l(s)^{\frac{1}{p-1}}\right)\left(\tilde{X}_s^{i}-\int_{s}^{T}\alpha\bar{Z}^l_sds\right) \right]ds \\
 &\quad + \int_{0}^{t}\left[\sigma_i\theta_i\left(X_s^{\bar{Z}^l,i} - \tilde{X}_s^{i}\right) + \left(\sigma_i\theta_i - \sigma\theta\right)\left(\tilde{X}_s^{i} -\int_{s}^{T}\alpha\bar{Z}^l_sds\right) \right]dW_s^i .
\end{align*}
In the sequel, let $C_T>0$ be a generic constant depending on $T$ only, which may be different from line to line. By applying Lemma~\ref{lem:momentestiXbarZi} and BDG inequality, we deduce that
\begin{align*}
 \E\left[\left|X_t^{\bar{Z}^l,i} - \tilde{X}_t^i\right|^2 \right]&\leq C_{T} \int_{0}^{t}\E\left[\left|X_s^{\bar{Z}^l,i}-\tilde{X}_s^i\right|^2 \right]ds
 + C_{T} \int_{0}^{t} \left|\mu_i\theta_i - g^l_i(s)^{\frac{1}{p_i-1}}-\mu\theta + g^l(s)^{\frac{1}{p-1}}\right|^2 ds  \\
 &\quad  + C_{T}(\sigma_i\theta_i- \sigma\theta)^{2}.
\end{align*}
Thus, we arrive from \eqref{eq:pi-c-star-i-two} at
\begin{align}\label{eq:estimation}
 \frac{1}{n}\sum_{i=1}^{n}\E\left[\left|X_t^{\bar{Z}^l,i} - \tilde{X}_t^i\right|^2\right]
 &\leq  C_{T}\int_{0}^{t}\frac{1}{n}\sum_{i=1}^{n}\E\left[\left|X_s^{\bar{Z}^l,i}-\tilde{X}_s^i\right|^2\right]ds + \frac{C_{T}}{n}\sum_{i=1}^{n}\left|\sigma_i\theta_i -\sigma\theta \right|^2 \nonumber \\
 &\quad + C_{T}\int_{0}^{t} \frac{1}{n}\sum_{i=1}^{n} \left|\mu_i\theta_i- g^l_i(s)^{\frac{1}{p_i-1}} - \mu\theta + g^l(s)^{\frac{1}{p-1}}\right|^2 ds .
\end{align}
It follows from the assumption $\bm{(A_c)}$ that
\begin{align}\label{eq:parameterconerror-2}
\begin{cases}
 \displaystyle \frac{1}{n}\sum_{i=1}^{n} \left|\sigma_i\theta_i- \sigma\theta\right|^2 = O\left(n^{-1}\right), \\
 \displaystyle \sup_{s\in[0,T]}\frac{1}{n}\sum_{i=1}^{n}\left|\mu_i\theta_i - g^l_i(s)^{\frac{1}{p_i-1}}  - \mu\theta + g^l(s)^{\frac{1}{p-1}}\right|^2 = O\left(n^{-1}\right).
\end{cases}
\end{align}
Then, the Gronwall's lemma with \eqref{eq:estimation} yields that
\begin{align}\label{eq:X-i-star-tilde}
 \sup_{t\in[0,T]}\frac{1}{n}\sum_{i=1}^{n} \E\left[\left|X_t^{\bar{Z}^l,i} - \tilde{X}_t^i\right|^2 \right]  = O\left(n^{-1}\right).
\end{align}
Moreover, we also have from \eqref{eq:parameterconerror-2} that $ L_2^{(n)}(t)= O(n^{-1})$. Finally, it obviously holds that
\begin{align}\label{eq:conerror000-two}
 \frac{2}{n}\sup_{t\in[0,T]}\Ex\left[\left|\tilde{X}_t^1-\mathbb{E}[X_t^{l,\bar{Z}^l}]\right|^2 \right] = O\left(n^{-1}\right).
\end{align}
Therefore, by \eqref{eq:diffXXstar00-two} and \eqref{eq:conerror000-two}, it follows that $ L_1^{(n)}(t)= O(n^{-1})$. Then, applying Gronwall's inequality to \eqref{eq:Ex-Z-n-Z-two} yields \eqref{eq:lem-Z-n-star-two}.
\end{proof}

We then provide the proof of Theorem \ref{thm:epsilon-Nash-equilibrium-two}.
\begin{proof}[Proof of Theorem \ref{thm:epsilon-Nash-equilibrium-two}] For $i=1,\ldots,n$, let $Z^i=(Z_t^i)_{t\in[0,T]}$ and $Z^{*,l,i}=(Z_t^{*,l,i})_{t\in[0,T]}$ be the habit formation processes of agent $i$ under an arbitrary admissible strategy $(\pi^i,c^i)\in {\cal A}^{l,i}(x_0)$ and under the strategy $(\pi^{*,l,i},c^{*,l,i})\in{\cal A}^{l,i}(x_0)$ given in~\eqref{eq:pi-c-star-i-two} respectively. We denote
\begin{align}\label{eq:barZ-iast-two}
  \bar{Z}^{*, l,n,-i}_t := \frac{1}{n}\sum_{j\neq i} Z_t^{\ast,l,j},\quad t\in[0,T].
\end{align}
In terms of \eqref{eq:Objective-i-two}, we have that, for $i = 1,\ldots,n$,
\begin{align*}
 J^l_i( (\pi^i,c^i),(\bm{\pi}^{*,l},\bm{c}^{*,l})^{-i}) &= \Ex\left[\int_0^T U_i\left(c^i_sX_s^i - \alpha\bar{Z}_s^{*,l,n,-i}-\frac{\alpha}{n}Z_s^i\right)ds + U_i(X_T^i) \right],\nonumber\\
 J^l_i((\pi^{*,l,i},c^{*,l,i}),(\bm{\pi}^{*,l},\bm{c}^{*,l})^{-i}) &= \Ex\left[\int_0^T U_i\left(c^{*,l,i}_sX_s^{*,l,i} - \alpha\bar{Z}_s^{*,l,n}\right)ds + U_i(X_T^{*,l,i})\right].
\end{align*}
where $X^{*,l,i}=(X_t^{*,l,i})_{t\in[0,T]}$ obeys the dynamics \eqref{eq:X-star-i-two}, and for an admissible control $(\pi^i,c^i)\in{\cal A}^{l,i}(x_0)$, the process $X^i=(X_t^i)_{t\in[0,T]}$ satisfies
\begin{equation}\label{eq:Xti0x}
 \frac{dX_t^i}{X_t^i} = \pi^i_t\mu_idt +\pi^i_t\sigma_i dW^i_t -c^i_tdt,\quad X_0^i=x_0.
\end{equation}

In order to prove \eqref{eq:def-epsilon-Nash-linear} in Definition \ref{def:nashequilibrium-linear}, we also introduce an auxiliary optimal control problem ($\mathbf{P}^l$): for $\bar{Z}^{l}=(\bar{Z}^{l}_t)_{t\in[0,T]}$ being the unique fixed point in \autoref{lem:fixedpoint-two}, let us consider
\begin{align}
 \sup_{(\pi^i,c^i) \in \mathcal{A}^{l,i}(x_0)} \bar{J}_i((\pi^i,c^i);\bar{Z}^{l}) := \sup_{(\pi^i,c^i) \in \mathcal{A}^{l,i}(x_0)} \Ex\left[\int_0^T U_i\left(c^i_s{X}_s^i-\alpha\bar{Z}^{l}_s\right)ds + U_i({X}_T^i)\right] {\color{red}.} \label{eq:objective-J-bar-P-two}
\end{align}
By the MFG results of linear habit formation in Section 3.1, we obtain that the optimal strategy of the auxiliary control problem ($\mathbf{P}^l$) is
\begin{equation}\label{eq:pi-c-i-bar}
\left\{
  \begin{aligned}
   & \pi_t^{*,\bar{Z}^{l},i} =  \frac{\mu_i}{(1-p_i)\sigma_i^2X_t^{\bar{Z}^{l},i}}\left(X_t^{\bar{Z}^{l},i}-\int_{t}^{T}\alpha\bar{Z}^{l}_sds\right), \\[0.4em]
   & c_t^{*,\bar{Z}^{l},i} = \frac{1}{X_t^{\bar{Z}^{l},i}}\left(\alpha\bar{Z}^{l}_t +\left(X_t^{\bar{Z}^{l},i}-\int_{t}^{T}\alpha\bar{Z}^{l}_s ds\right)g^l_i(t)^{\frac{1}{p_i-1}}\right),
\end{aligned}
\right.
\end{equation}
where the controlled wealth process $X^{\bar{Z}^{l},i}=(X^{\bar{Z}^{l},i}_t)_{t\in[0,T]}$ is given by \eqref{eq:XbarZiintbarZs}.

We then focus on the verification of \eqref{eq:def-epsilon-Nash-linear} by using the auxiliary problem ($\mathbf{P}^l$) in \eqref{eq:objective-J-bar-P-two}. We have that
\begin{align}\label{eq:J-J-two}
 & \sup_{(\pi^i,c^i) \in \mathcal{A}^{l,i}(x_0)} J^l_i\left( (\pi^i,c^i),(\bm{\pi}^{*,l},\bm{c}^{*,l})^{-i} \right) - J^l_i\left(\bm{\pi}^{*,l},\bm{c}^{*,l}\right)\nonumber  \\
 &\qquad = \left(\sup_{(\pi^i,c^i) \in \mathcal{A}^{l,i}(x_0)}J^l_i\left((\pi^i,c^i),(\bm{\pi}^{*,l},\bm{c}^{*,l})^{-i} \right) - \sup_{(\pi^i,c^i) \in \mathcal{A}^l(x_0)}\bar{J}^l_i\left((\pi^i,c^i);\bar{Z}^{l}\right) \right) \nonumber \\
  &\qquad \quad + \sup_{(\pi^i,c^i) \in \mathcal{A}^{l,i}(x_0)} \bar{J}^l_i\left((\pi^i,c^i);\bar{Z}^{l}\right) - J_i^l\left(\bm{\pi}^{*,l},\bm{c}^{*,l}\right) \notag\\
  &\qquad\leq \sup_{(\pi^i,c^i)\in\mathcal{A}^{l,i}(x_0)}\left(J_i^l\left((\pi^i,c^i),(\bm{\pi}^{*,l},\bm{c}^{*,l})^{-i}\right)  -\bar{J}^l_i\left((\pi^i,c^i);\bar{Z}^{l}\right)\right) \nonumber \\
  &\qquad\quad + \sup_{(\pi^i,c^i)\in \mathcal{A}^{l,i}(x_0)}  \bar{J}^l_i\left((\pi^i,c^i);\bar{Z}^{l}\right)-J^l_i\left(\bm{\pi}^{*,l},\bm{c}^{*,l}\right).
\end{align}
We first evaluate the first term of RHS of \eqref{eq:J-J-two} that, for all $(\pi^i,c^i)\in\mathcal{A}^{l,i}(x_0)$,
\begin{align*}
 & J^l_i((\pi^i,c^i),(\bm{\pi}^{*,l,i},\bm{c}^{*,l,i})^{-i}) - \bar{J}^l_i((\pi^i,c^i);\bar{Z}^{l}) \\
 &\quad = \left\{ \Ex\left[\int_0^TU_i\left(c^i_sX_s^i-\alpha\bar{Z}_s^{*,l,n,-i}-\frac{\alpha}{n}Z_s^i\right)ds\right]
 -  \Ex\left[\int_0^T U_i\left(c^i_sX_s^i-\alpha\bar{Z}_s^{*,l,n}\right)ds \right] \right\}\\
 & \qquad + \left\{\Ex\left[\int_0^T U_i\left(c^i_sX_s^i-\alpha\bar{Z}_s^{*,l,n}\right)ds \right] -\Ex\left[\int_0^T U_i\left(c^i_sX_s^i-\alpha\bar{Z}^{l}_s\right)ds\right] \right\}\\
 &\quad  := I_i^{(1)} + I_i^{(2)}.
\end{align*}
For the term $I_i^{(1)}$, we have that
\begin{align*}
 I_i^{(1)} 
= \frac{1}{p_i}\Ex\left[\int_0^T \left[\left(c^i_sX_s^i-\alpha\bar{Z}_s^{*,l,n,-i}-\frac{\alpha}{n}Z_s^i\right)^{p_i} - \left(c^i_sX_s^i-\alpha\bar{Z}_s^{*,l,n}\right)^{p_i}\right]ds\right].
\end{align*}
 Using the inequality $(a+b)^{p_i}\leq a^{p_i} + b^{p_i} $  for all $a,b >0$, $p_i\in (0,1)$, we can derive on the event $\{Z_s^{*,l,i} >Z_s^i\}$ that
\begin{align}\label{eq:I-1-i-1-two}
 &\left(c^i_sX_s^i-\alpha\bar{Z}_s^{*,l,n,-i}-\frac{\alpha}{n}Z_s^i\right)^{p_i} - \left(c^i_sX_s^i-\alpha\bar{Z}_s^{*,l,n}\right)^{p_i}
  \leq \alpha^{p_i}\left(\bar{Z}_s^{*,l,n} - \bar{Z}_s^{*,l,n,-i}-\frac{1}{n}Z_s^i\right)^{p_i}   \nonumber \\
 &\quad = \frac{1}{n^{p_i}}\alpha^{p_i}\left(Z_s^{\ast,l,i}- Z_s^{i}\right)^{p_i}
 \leq \frac{1}{n^{p_i}}\alpha^{p_i}\left(Z_s^{\ast,l,i}\right)^{p_i}.
\end{align}
With the help of Lemma~\ref{lem:estimZtnilinear}, the Jensen's inequality with \eqref{eq:I-1-i-1-two} and $p_i\in(0,1)$, we arrive at
\begin{align*}
  I_i^{(1)} &\leq  \frac{1}{p_i} \mathbb{E}\left[\int_0^T  \frac{1}{n^{p_i}}\alpha^{p_i}\left(Z_s^{\ast,l,i}\right)^{p_i}  ds\right] 
 \leq \frac{1}{p_i}\frac{1}{n^{p_i}}\alpha^{p_i} \int_0^T \left(\mathbb{E}\left[Z_s^{\ast,l,i}\right]\right)^{p_i} dt
  \leq  \frac{C_T}{n^{p_i}}  = O(n^{-\underline{m}_p}),
\end{align*}
for some constant $C_T>0$ independent of $i$. Here, we select all the elements of $\{p_i\}_{i\geq1}$ belonging to $(0,1)$ as the subsequence $\{p_{i_k}\}_{k\geq 1}$, and we use the fact that  $0<\underline{m}_p:=\inf_{k\in\N}p_{i_k}\leq \bar{m}_p:=\sup_{k\in\N}p_{i_k}<1$. On the other hand, on the event $\{Z_s^{*,l,i} \leq Z_s^i\}$, we have that
$(c^i_sX_s^i-\alpha\bar{Z}_s^{*,l,n,-i}-\frac{\alpha}{n}Z_s^i)^{p_i} -(c^i_sX_s^i-\alpha\bar{Z}_s^{*,l,n})^{p_i} \leq 0$.
This yields that
\begin{align}\label{eq:I-1-i-order-two}
 I_i^{(1)} &=  \frac{1}{p_i}\Ex\left[\int_0^T \left[\left(c^i_sX_s^i-\alpha\bar{Z}_s^{*,l,n,-i}-\frac{\alpha}{n}Z_s^i\right)^{p_i} - \left(c^i_sX_s^i-\alpha\bar{Z}_s^{*,l,n}\right)^{p_i}\right]ds\right] = O(n^{-\underline{m}_p}).
\end{align}

Similarly, for the term $I_i^{(2)}$, it suffices to consider its estimate on the event $\{\bar{Z}^{l}_s> \bar{Z}_s^{*,l,n}\}$, on which we have
\begin{align*}
  (c^i_sX_s^i-\alpha\bar{Z}_s^{*,l,n})^{p_i} - (c^i_sX_s^i-\alpha\bar{Z}^{l}_s)^{p_i}
 &\leq  \alpha^{p_i}(\bar{Z}_s^{l}- \bar{Z}_s^{\ast,l,n})^{p_i}.
\end{align*}
By applying H\"{o}lder inequality and Lemma~\autoref{lem:estimZtnilinear}, it holds that
\begin{align}\label{eq:I-i-2-two}
 I_i^{(2)}
 & = \mathbb{E}\left[\int_0^T \frac{1}{p_i}\left[ \left(c^i_sX_s^i-\alpha\bar{Z}_s^{*,l,n}\right)^{p_i} -  \left(c^i_sX_s^i-\alpha\bar{Z}^{l}_s\right)^{p_i} \right]ds \right]   \nonumber \\
 &\leq \frac{1}{p_i} \Ex\left[\int_0^T \alpha^{p_i}\left|\bar{Z}^{l}_s-\bar{Z}_s^{*,l,n}\right| ^{p_i}ds\right]\nonumber\\
 &\leq \frac{1}{p_i}\int_0^T\Ex\left[\left|\bar{Z}^{l}_s-\bar{Z}_s^{*,l,n}\right|^2\right]^{\frac{p_i}{2}}ds
  =  O\left(n^{-\frac{\underline{m}_p}{2}}\right).
\end{align}
We then conclude that
\begin{align}\label{eq:estimate0000}
\sup_{(\pi^i,c^i)\in\mathcal{A}^{l,i}(x_0)}\left(J^l_i\left((\pi^i,c^i),(\bm{\pi}^{*,l},\bm{c}^{*,l})^{-i}\right)
-\bar{J}^l_i\left((\pi^i,c^i);\bar{Z}^{l}\right)\right)=O\left(n^{-\frac{\underline{m}_p}{2}}\right).
\end{align}

We next focus on the second term of RHS of \eqref{eq:J-J-two}. We emphasize that the optimal solution $(\pi^{*,\bar{Z}^{l},i},c^{*,\bar{Z}^{l},i})$  of the auxiliary control problem ($\mathbf{P}^l$) defined in \eqref{eq:pi-c-i-bar} differs from the control pair $(\pi^{*,l,i},c^{*,l,i})$ constructed in \eqref{eq:pi-c-star-i-two}.
Note that
\begin{align*}
 \sup_{(\pi^i,c^i)\in\mathcal{A}^{l,i}(x_0)}\bar{J}^l_i((\pi^i,c^i);\bar{Z}^{l})-J^l_i(\bm{\pi}^{*,l},\bm{c}^{*,l})& = \bar{J}^l_i((\pi^{*,\bar{Z}^{l},i},c^{*,\bar{Z}^{l},i});\bar{Z}^{l}) - J^l_i(\bm{\pi}^{*,l},\bm{c}^{*,l}).
\end{align*}
From the construction of $(\pi^{*,l,i},c^{*,l,i})$ in \eqref{eq:pi-c-star-i-two}, it follows that $ c_t^{*,\bar{Z}^{l},i}X_t^{\bar{Z}^{l},i}-\alpha\bar{Z}^{l}_t = c_t^{*,l,i}X_t^{*,l,i}-\alpha\bar{Z}_t^{*,l,n}$ for all $t\in[0,T]$. As a consequence, we deduce that
\begin{align}\label{eq:J-bar-J-22}
 &\bar{J}^l_i((\pi^{*,\bar{Z}^l,i},c^{*,\bar{Z}^l,i});\bar{Z}^l) - J^l_i(\bm{\pi}^{*,l},\bm{c}^{*,l})
= \frac{1}{p_i}\Ex \left[\left(X_T^{\bar{Z}^l,i}\right)^{p_i} -\left(X_T^{*,l,i}\right)^{p_i}\right].
\end{align}
It is sufficient to analyze \eqref{eq:J-bar-J-22} on the event $\{X_T^{\bar{Z}^l,i} > X_T^{*,l,i} \}$ because $\Ex[(X_T^{\bar{Z^l},i})^{p_i} -(X_T^{*,l,i})^{p_i}] \leq 0 $ on the event $ \{X_T^{\bar{Z^l},i} \leq X_T^{*,l,i}\}$. Then, by applying Jesen inequality,  we can derive that
\begin{equation}\label{eq:J-bar-J-33}
 \bar{J}^l_i((\pi^{*,\bar{Z}^l,i},c^{*,\bar{Z}^l,i});\bar{Z}) - J^l_i(\bm{\pi}^{*,l},\bm{c}^{*,l} )
 \leq \frac{1}{p_i}\Ex \left[\left(X_T^{\bar{Z}^l,i} - X_T^{*,l,i}\right)^{p_i}\right]
\leq \frac{1}{p_i}\Ex \left[\left|X_T^{\bar{Z}^l,i} - X_T^{*,l,i}\right|^2\right]^{\frac{p_i}{2}}{\color{red}.}
\end{equation}
It follows from \eqref{eq:X-starbar-i} and \eqref{eq:X-star-i-two} that
\begin{equation*}
X_T^{\bar{Z}^l,i} - X_T^{*,l,i} = \int_{0}^{T}\alpha\left(\bar{Z}^l_t - \bar{Z}_t^{*,l,n}\right)dt.
\end{equation*}
By Fubini theorem, there exists a constant $C_T>0$ independent of $i$ such that
\begin{align*}
 \Ex\left[\left|X_T^{\bar{Z}^l,i} - X_T^{*,l,i}\right|^2\right] 
 &\leq C_T\int_{0}^{T}\Ex\left[\left|\bar{Z}^l_t - \bar{Z}_t^{*,l,n}\right|^2\right]dt.
\end{align*}
Then Lemma~\ref{lem:estimZtnilinear} yields that
\begin{equation}\label{eq:ExdiffXTi2}
  \Ex\left[\left|X_T^{\bar{Z}^l,i} - X_T^{*,l,i}\right|^2\right] = O\left(n^{-1}\right).
\end{equation}
Thus, the estimates \eqref{eq:J-bar-J-33} and \eqref{eq:ExdiffXTi2} imply that
\begin{align}\label{eq:J-bar-J-1-two}
 \bar{J}^l_i((\pi^{*,i},c^{*,i});\bar{Z}^l) - J^l_i(\bm{\pi}^{*,l,i},\bm{c}^{*,l,i}) = O\left(n^{-\frac{\underline{m}_p}{2}}\right).
\end{align}
We obtain from \eqref{eq:J-J-two}, \eqref{eq:estimate0000} and \eqref{eq:J-bar-J-1-two} that
\begin{align*}
\sup_{(\pi^i,c^i) \in \mathcal{A}^{l,i}(x_0)} J^l_i\left( (\pi^i,c^i),(\bm{\pi}^{*,l},\bm{c}^{*,l})^{-i} \right) - J^l_i\left(\bm{\pi}^{*,l},\bm{c}^{*,l}\right)= O\left(n^{-\frac{\underline{m}_p}{2}}\right).
\end{align*}
Thus, we get the desired result with $\epsilon_n= O(n^{-\frac{\underline{m}_p}{2}})$.
\end{proof}

\subsection{Approximation under multiplicative habit formation}

We next construct and verify an approximate Nash equilibrium to the $n$-player game under the multiplicative habit formation preference. Again, for $i=1,\ldots,n$, we recall that the objective functional \eqref{eq:Objective-i} of agent $i$ can be rewritten as: for $(\pi^i,c^i)\in{\cal A}^{m,i}(x_0)$,
\begin{align}\label{eq:Objective-i-n}
J^m_i((\pi^i,c^i),(\bm{\pi},\bm{c})^{-i}) = \Ex\left[\int_0^T U_i\left(\frac{c^i_sX_s^i}{(\bar{Z}_s^n)^{{\alpha}} }\right)ds +  U_i(X_T^i) \right].
\end{align}
The definition of an approximate Nash equilibrium under the multiplicative habit formation is given below.
\begin{definition}[Approximate Nash equilibrium]\label{def:nashequilibrium}
Let ${\cal A}^m(x):=\prod_{i=1}^n{\cal A}^{m,i}(x)$. An admissible strategy $(\bm{\pi}^{*,m},\bm{c}^{*,m}) =( (\pi^{*,m,1},c^{*,m,1}),\ldots, (\pi^{*,m,n},c^{*,m,n}) )\in{\cal A}^m(x)$ is called an $\epsilon$-Nash equilibrium to the $n$-player game problem \eqref{eq:Objective-i}  if, for all $(\pi^{i},c^{i}) \in{\cal A}(x)$ with $i = 1,\ldots,n$, it holds that
\begin{equation}\label{eq:def-epsilon-Nash}
   \sup_{(\pi^i,c^i)\in{\cal A}^{m,i}(x)} J_i((\pi^{i},c^{i}) , (\bm{\pi}^{*,m},\bm{c}^{*,m})^{-i})  \leq  J_i( (\bm{\pi}^{*,m},\bm{c}^{*,m}))+\epsilon.
\end{equation}
\end{definition}

Thanks to the non-addictive nature of the habit formation in \eqref{eq:Objective-i-n}, the construction of the closed-loop approximate Nash equilibrium in the $n$-player game becomes much easier than the case of linear habit formation. For $i=1,\ldots,n$, let us consider $(\pi^{*,m,i}, c^{*,m, i}) = (\pi_t^{*,m,i}, c_t^{*,m,i})_{t\in[0,T]}$ that, for $t\in[0,T]$,
\begin{align}\label{eq:pi-c-star-i}
\left\{
  \begin{aligned}
     &\pi_t^{*,m, i} := \frac{\mu_i}{(1- p_i)\sigma_i^2},\\
     &c_t^{*,m, i} := g^m_i(t)^{\frac{1}{p_i-1}} (\bar{Z}_t^m)^{\frac{{\alpha}p_i}{p_i-1}},
\end{aligned}
\right.
\end{align}
where $\bar{Z}^m=(\bar{Z}^m_t)_{t\in[0,T]}$ is the unique fixed point of \eqref{eq:fixed-point-4} established in Theorem~\ref{prop:fixedpoint} for the mean field game problem. The function $t\to g^m_i(t)$ is given by, for $t\in[0,T]$,
\begin{align}\label{eq:solution-g-i}
 g^m_i(t) = \left( e^{b_i(t-T)}+ e^{b_it} \int_{t}^{T} e^{-b_is} \left(\bar{Z}_s\right)^{\frac{\alpha p_i}{p_i-1}}ds \right)^{1-p_i},\quad b_i :=  -\frac{\mu_i^2}{2\sigma_i^2}\frac{p_i}{(1-p_i)^2}.
\end{align}
Denote by $X^{*,m,i} = (X_t^{*,m,i})_{t\in[0,T]} $ the wealth process of agent $i$ under the investment and consumption strategy pair $(\pi^{*,m,i}, c^{*,m,i})$ in \eqref{eq:pi-c-star-i} that
\begin{equation}\label{eq:X-star-i}
 \frac{dX_t^{*,m,i}}{X_t^{*,m,i}} = \pi_t^{*,m, i}\mu_idt +\pi_t^{*,m, i}\sigma_i dW^i_t - c_t^{*,m, i}dt,\quad X_0^{*,i}=x_0>0.
\end{equation}

For $ i =1,\ldots, n$,  the $i$th agent's habit formation process is given by
\begin{equation}\label{eq:Z-star-i}
  Z_t^{*,m,i} = e^{-\delta t} \left( z_0 + \int_{0}^{t}\delta e^{\delta s }c_s^{*,m,i}X_s^{*,m,i} ds \right).
\end{equation}
Let us also denote
\begin{equation}\label{eq:average-star}
  \bar{C}_t^{*,m,n} := \frac{1}{n}\sum_{i=1}^{n} C_t^{*,m,i} = \frac{1}{n}\sum_{i=1}^{n} c_t^{*,m,i}X_t^{*,m,i},\quad \bar{Z}_t^{*,m,n} := \frac{1}{n}\sum_{i=1}^{n} Z_t^{*,m,i},\quad t\in[0,T].
\end{equation}
It follows that
\begin{equation}\label{eq:bar-Z-n-star}
  \bar{Z}_t^{*,m,n} = e^{-\delta t} \left(z_0 + \int_{0}^{t}\delta e^{\delta s }\bar{C}_s^{*,m,n} ds \right),\quad t\in[0,T].
\end{equation}

Next, we introduce the main result of this section on the existence of an approximate Nash equilibrium under the multiplicative external habit formation.
\begin{theorem}\label{thm:epsilon-Nash-equilibrium}
The control pair $(\bm{\pi}^{\ast,m},\bm{c}^{\ast,m})=((\pi^{*,m,1}, c_t^{*,m,1}),\ldots, (\pi^{*,m,n},c_t^{*,m,n}))_{t\in[0,T]}$ given in \eqref{eq:pi-c-star-i} is an $\epsilon_n$-Nash equilibrium for the $n$-player game problem \eqref{eq:Objective-i} with the explicit order $\epsilon_n=O(n^{-\frac{1}{2}})$.
\end{theorem}

To prove Theorem~\ref{thm:epsilon-Nash-equilibrium}, we need the following auxiliary results.

\begin{lemma}\label{lem:auxiliary-results}
For any $n\geq1$, it holds that
\begin{itemize}
  \item [{\bf(i)}] For $i=1,\ldots,n$, let the wealth process $X^{*,m,i} = (X_t^{*,m,i})_{t\in[0,T]}$ be defined by \eqref{eq:X-star-i}. Then, for any $q>1$, there exists a constant $C_q>0$ independent of $i$ such that
      \begin{equation}\label{eq:lem-X-star-E}
        \sup_{t\in[0,T]}\Ex\left[\left(X_t^{*,m,i}\right)^q \right] \leq C_q.
      \end{equation}
  \item [{\bf(ii)}] For $i=1,\ldots,n$, let $Z^{*,m,i}=(Z_t^{*,m,i})_{t\in[0,T]}$ be defined by \eqref{eq:Z-star-i}. Then, for any $q\in(1,\infty)\cup(-\infty,0)$, there exists a constant $C_q>0$ independent of $i$ such that
      \begin{equation}\label{eq:lem-Z-star-i-E}
      \sup_{t\in[0,T]} \Ex\left[\left(Z_t^{*,m,i}\right)^q \right] \leq C_q.
      \end{equation}
  \item [{\bf(iii)}] Let $\bar{Z}^m=(\bar{Z}^m_t)_{t\in[0,T]}$ be the unique fixed point to \eqref{eq:fixed-point-4}, and let $\bar{Z}^{*,m,n}=(\bar{Z}_t^{*,m,n})_{t\in[0,T]}$ be defined by \eqref{eq:bar-Z-n-star}.  Then, for any even number $q\geq 2$,
 \begin{align}\label{eq:lem-Z-n-star}
   \sup_{t\in[0,T]}\Ex\left[\left|\bar{Z}_t^{*,m,n} -\bar{Z}^m_t\right|^q \right] = O\left(n^{-\alpha_q}\right),\quad \alpha_q :=\frac{q}{2}\wedge(q-1).
 \end{align}
\end{itemize}
\end{lemma}

\begin{proof}
{\bf(i)} Given $(\pi^{*,m,i},c^{*,m,i})$ in \eqref{eq:pi-c-star-i}, by applying It\^o's formula to $(X_t^{*,m,i})^q$ with $q>1$, we have from the dynamics \eqref{eq:X-star-i} that
\begin{align}\label{eq:X-i-star-q}
 \left(X_t^{*,m,i}\right)^q & = x_0^q \exp\left\{\int_{0}^{t} q\left(\mu_i\pi_s^{*,m,i}-c_s^{*,m,i} -\frac{\sigma_i^2}{2} (\pi_s^{*,m,i})^2\right)ds + \int_{0}^{t} q\sigma_i\pi_s^{*,m,i} dW_s^i\right\}.
\end{align}
It follows from \eqref{eq:pi-c-star-i} that
\begin{align}\label{eq:Ex-X-star-i-2}
\Ex\left[\left(X_t^{*,m,i}\right)^q\right]
 & = x_0^q \exp\left(\int_0^t \left(\frac{q\mu_i^2(1+q-2p_i)}{2(1-p_i)^2\sigma_i^2} - 2(\bar{Z}^m_s)^{\frac{\alpha p_i}{p_i-1}} g^m_i(s)^{\frac{1}{p_i-1}}\right)ds\right)    \nonumber \\
 &\leq \sup_{i\in\mathbb{N}}x_0^q \exp\left( \frac{qT(1+q-2p_i)\mu_i^2}{2(1-p_i)^2\sigma_i^2}\right):=C_q,
 \end{align}
where we have used the positivity of $\bar{Z}^m_t$ and $g^m_i(t)$ for $t\in[0,T]$ for the first inequality above. Note that, as $i\to\infty$, $\exp\left( \frac{qT(1+q-2p_i)\mu_i^2}{2(1-p_i)^2\sigma_i^2}\right)\to \exp\left( \frac{qT(1+q-2p)\mu^2}{2(1-p)^2\sigma^2}\right)$ by the assumption $\bm{(A_{c})}$. Then, the sequence $\exp\left( \frac{qT(1+q-2p_i)\mu_i^2}{2(1-p_i)^2\sigma_i^2}\right)$ for $i\in\mathbb{N}$ is bounded, and hence $C_q$ is finite.

{\bf(ii)} Recall the inequality $(a+b)^r \leq 2^{r-1}(a^r+b^r)$ for $a,b >0$ and $r>1$.  From  \eqref{eq:Z-star-i}, \eqref{eq:lem-X-star-E} and H\"{o}lder inequality with exponents $(q,q_0)\in(1,+\infty)^2$ satisfying $\frac{1}{q_0} +\frac{1}{q}=1$, it follows that
\begin{align}\label{eq:Z-star-i-bound}
 \sup_{t\in[0,T]}\Ex \left[\left(Z_t^{*,m,i}\right)^q\right]
 &\leq 2^{q-1} z_0^q + 2^{q-1}\delta^q \Ex\left[\left(\int_{0}^{T}c_s^{*,m,i}X_s^{*,m,i} ds\right)^q \right]     \nonumber \\
 & \leq 2^{q-1}z_0^q + 2^{q-1}\delta^q\left(\int_{0}^{T}\left(c_s^{*,m,i}\right)^{q_0}ds\right)^{\frac{q}{q_0}} \Ex\left[\int_{0}^{T}\left(X_s^{*,m,i}\right)^q ds \right]      \nonumber \\
 & \leq 2^{q-1} z_0^q + 2^{q-1} \delta^q C_qT\sup_{i\in\mathbb{N}}\left(\int_{0}^{T}\left(c_s^{*,m,i}\right)^{q_0}ds\right)^{\frac{q}{q_0}},
\end{align}
where $C_q>0$ is given in \eqref{eq:lem-X-star-E}. On the other hand, in view of \eqref{eq:pi-c-star-i}, we have that, for $q_0>1$,
\begin{align}\label{eq:c-star-i-bound}
 \sup_{i\in\mathbb{N}}\sup_{t\in[0,T]} \left(c_t^{*,m,i}\right)^{q_0} & = \sup_{i\in\mathbb{N}}\sup_{t\in[0,T]}\left(\frac{ (\bar{Z}^m_t)^{\frac{{\alpha}p_i}{p_i-1}}  }{ e^{a_i(t-T)}+ e^{a_it} \int_{t}^{T} e^{-a_is} \left(\bar{Z}_s^m\right)^{\frac{{\alpha}p_i}{p_i-1}}ds }\right)^{q_0}\nonumber\\
 &\leq \sup_{i\in\mathbb{N}}\sup_{t\in[0,T]}\left(\frac{ (\bar{Z}^m_t)^{\frac{{\alpha}p_i}{p_i-1}}}{ e^{-a_iT}}\right)^{q_0}\leq \sup_{i\in\mathbb{N}}\left(e^{a_iT}(e^{-\delta t}z_0)^{\frac{{\alpha}p_i}{p_i-1}}\right)^{q_0}\nonumber\\
 &\leq \sup_{i\in\mathbb{N}}e^{(a_i+\frac{\delta {\alpha}p_i}{1-p_i})q_0T}z_0^{\frac{q_0{\alpha}p_i}{p_i-1}}:=K,
\end{align}
where we used that fact $\bar{Z}^m_t\geq e^{-\delta t}z_0$ from Theorem~\ref{prop:fixedpoint} with \eqref{eq:fixed-point-4} and $\frac{{\alpha}p_i}{p_i-1}<0$ for the second inequality. In addition, it follows from the assumption $\bm{(A_{c})}$ that $K$ is a finite (positive) constant. Then, using \eqref{eq:Z-star-i-bound} and \eqref{eq:c-star-i-bound}, we obtain that
\begin{align*}
\sup_{t\in[0,T]}\Ex \left[\left(Z_t^{*,m,i}\right)^q\right]\leq 2^{q-1}z_0^q + 2^{q-1} \delta^q C_qT(KT)^{\frac{q}{q_0}},
\end{align*}
which proves the estimate \eqref{eq:lem-Z-star-i-E} under $q>1$.

For $q<0$, it follows from \eqref{eq:Z-star-i} that, for all $t\in[0,T]$,
\begin{align*}
  Z_t^{*,m,i} = e^{-\delta t} \left( z_0+ \int_{0}^{t}\delta e^{\delta s }c_s^{*,m,i}X_s^{*,m,i} ds \right)\geq e^{-\delta t}z_0.
\end{align*}
This yields from the assumption $\bm{(A_{c})}$ that
\begin{align*}
\sup_{t\in[0,T]} \Ex\left[ \left(Z_t^{*,m,i}\right)^q\right] &\leq  e^{-\delta q t}z_0^q \leq e^{-\delta q T}z_0^q := C_q<+\infty.
\end{align*}
Thus, we complete the proof of the estimate {\bf(ii)}.

{\bf(iii)}
In light of the consumption strategy $c^{*,m,i}$ given by \eqref{eq:pi-c-star-i} and the (feedback) consumption strategy given by \eqref{eq:picstar}, we have that
\begin{align*}
  \bar{Z}_t^{*,m,n} -\bar{Z}^m_t 
  & = e^{-\delta t} \int_{0}^{t}\delta e^{\delta s} c_s^m \left(\frac{1}{n}\sum_{i=1}^{n}X_s^{*,m,i} - \mathbb{E}[X_s^{m,\bar{Z}^m}]\right)ds \\
  &\quad + e^{-\delta t} \int_{0}^{t}\delta e^{\delta s} \frac{1}{n}\sum_{i=1}^{n}X_s^{*,m,i}( c_s^{*,m,i}-c_s^m) ds.
\end{align*}
Note that, it follows from Jensen's inequality that $(a+b)^q \leq 2^{q-1}(a^q + b^q)$, for any $a,b>0$ and $q >1$. Then, by applying H\"{o}lder inequality with arbitrary $q_0>1$ satisfying $\frac{1}{q} + \frac{1}{q_0} =1$, we arrive at
\begin{align*}
  \left|\bar{Z}_t^{*,m,n} -\bar{Z}^m_t\right|^q
  &\leq  2^{q-1}\left(\int_{0}^{t}\delta e^{\delta (s-t)} c_s^m \left|\frac{1}{n}\sum_{i=1}^{n}X_s^{*,m,i}-\mathbb{E}[X_s^{m,\bar{Z}^m}]\right|ds \right)^q  \\
  & \quad + 2^{q-1}\left(\int_{0}^{t}\delta e^{\delta (s-t)} \frac{1}{n}\sum_{i=1}^{n}X_s^{*,m,i}|c_s^{*,m,i}-c_s^m| ds \right)^q    \\
  &\leq 2^{q-1}\left(\int_{0}^{t}\delta^{q_0} e^{q_0\delta (s-t)} (c_s^m)^{q_0}ds\right)^{\frac{q}{q_0}} \int_{0}^{t}\left|\frac{1}{n}\sum_{i=1}^{n}X_s^{*,m,i} - \mathbb{E}[X_s^{m,\bar{Z}^m}]\right|^qds\\
  &\quad + 2^{q-1}\left(\int_{0}^{t}\delta^{q_0} e^{q_0\delta (s-t)}ds \right)^{\frac{q}{q_0}} \int_{0}^{t} \left(\frac{1}{n}\sum_{i=1}^{n}X_s^{*,m,i}| c_s^{*,m,i}-c_s^m|\right)^q ds.
\end{align*}
Take the expectation of both sides of above inequality, it deduces that, for all $t\in[0,T]$,
\begin{align}\label{eq:Ex-Z-n-Z}
  \mathbb{E}\left[\left|\bar{Z}_t^{*,m,n} -\bar{Z}^m_t\right|^q\right]
  &\leq 2^{q-1}\left(\int_{0}^{t}\delta^{q_0} e^{q_0\delta (s-t)} (c_s^m)^{q_0}ds\right)^{\frac{q}{q_0}}\int_{0}^{t}\mathbb{E}\left[\left|\frac{1}{n}\sum_{i=1}^{n}X_s^{*,m,i}-\mathbb{E}[X_s^{m,\bar{Z}^m}]\right|^q\right] ds   \nonumber \\
  &\quad + 2^{q-1}\left(\int_{0}^{t}\delta^{q_0} e^{q_0\delta (s-t)}ds \right)^{\frac{q}{q_0}} \mathbb{E}\left[\int_{0}^{t} \left(\frac{1}{n}\sum_{i=1}^{n}X_s^{*,m,i}|c_s^{*,m,i}-c_s^m|\right)^q ds \right]       \nonumber \\
  &\leq  2^{q-1}\delta^{q}\left(\int_{0}^{t}(c_s^m)^{q_0} ds\right)^{\frac{q}{q_0}} \int_{0}^{t}\mathbb{E}\left[\left|\frac{1}{n}\sum_{i=1}^{n}X_s^{*,m,i} - \mathbb{E}[X_s^{m,\bar{Z}^m}]\right|^q \right]ds \nonumber\\
  & \quad + 2^{q-1}\delta^{q}T^{\frac{q}{q_0}} \mathbb{E}\left[\int_{0}^{t}\frac{1}{n}\sum_{i=1}^{n}\left(X_s^{*,i}\right)^{q}\left|c_s^{*,m,i}-c_s^m\right|^q ds \right] \nonumber\\
  &=  2^{q-1}\delta^q \left(\int_{0}^{t} (c_s^m)^{q_0}ds\right)^{\frac{q}{q_0}}  \int_{0}^{t} \mathbb{E}\left[\left|\frac{1}{n}\sum_{i=1}^{n}X_s^{*,m,i}-\mathbb{E}[X_s^{m,\bar{Z}^m}]\right|^q \right]ds \nonumber\\
  & \quad + 2^{q-1}\delta^{q}T^{\frac{q}{q_0}} \int_{0}^{t} \frac{1}{n}\sum_{i=1}^{n}  \left|c_s^{*,m,i}-c_s^m \right|^q \mathbb{E}\left[\left(X_s^{*,m,i}\right)^{q} \right]ds .
\end{align}

Recall $X^{*,m,i}=(X_t^{*,m,i})_{t\in[0,T]}$ satisfies \eqref{eq:X-star-i}. It follows from Lemma~\ref{lem:auxiliary-results}-{\bf(i)} that, there exists a constant $C_q >0$ independent of $i$ such that
\begin{align}\label{eq:Ex-X-i-star}
 \sup_{t\in[0,T]}\Ex\left[\left|X_t^{*,m,i}\right|^q\right] \leq  C_q.
\end{align}
In view of \eqref{eq:picstar}, we have from Theorem~\ref{prop:fixedpoint} that, for all $t\in[0,T]$,
\begin{align}\label{eq:c-star-bound}
  c_t^{m}& =  g^{m}(t)^{\frac{1}{p-1}} (\bar{Z}^m_t)^{\frac{{\alpha}p}{p-1}}=  \frac{(\bar{Z}_t^m)^{\frac{{\alpha}p}{p-1}}}{e^{a(t-T)}+ e^{at} \int_{t}^{T} e^{-as} \left(\bar{Z}^m_s\right)^{\frac{{\alpha}p}{p-1}}ds }\leq e^{aT} z_0^{\frac{{\alpha}p}{p-1}}:= C_0.
\end{align}
Thus, we combine \eqref{eq:Ex-X-i-star} and \eqref{eq:c-star-bound} to have the following estimation that
\begin{align}\label{eq:Ex-Z-n-Z-2}
  \mathbb{E}\left[\left|\bar{Z}_t^{*,m,n} -\bar{Z}^m_t\right|^q\right]
  &\leq   2^{q-1} \delta^q C_0^qT^{\frac{q}{q_0}} \int_{0}^{t} \mathbb{E}\left[\left|\frac{1}{n}\sum_{i=1}^{n}X_s^{*,m,i} - \mathbb{E}[X_s^{m,\bar{Z}^m}]\right|^q \right]ds \nonumber\\
  & \quad + 2^{q-1}\delta^{q}T^{\frac{q}{q_0}} C_q \int_{0}^{t} \frac{1}{n}\sum_{i=1}^{n} \left|c_s^{*,m,i}-c_s^m\right|^q ds \nonumber \\
  &:= I_1^{(n)}(t) + I_2^{(n)}(t).
\end{align}
For the estimate of $I_1^{(n)}(t)$, let us consider the auxiliary process $\hat{X}^i$, for $i=1,\ldots,n$, satisfying 
\begin{equation}\label{eq:X-hat-i}
 \frac{d\hat{X}_t^i}{\hat{X}_t^i} = \pi_t^{m}\mu dt +\pi_t^{m}\sigma dW_t^i - c_t^{m}dt,\quad \hat{X}_0^i=x_0,
\end{equation}
where $(\pi_t^{m},c_t^{m})$, $t\in[0,T]$, is defined by \eqref{eq:picstar} with the given $\bar{Z}^m$. Therefore, we have from \eqref{eq:X-hat-i} that, for $t>0$, the sequence $(\hat{X}_t^i)_{i=1}^n$ is i.i.d., and $\Ex[\hat{X}_t^i]=\Ex[X_t^{m,\bar{Z}^m}]$ for all $i=1,\ldots,n$. Then, it follows from Jensen's inequality that
\begin{align}\label{eq:diffXXstar00}
 &\quad \mathbb{E}\left[\left|\frac{1}{n}\sum_{i=1}^{n}X_t^{*,m,i} - \mathbb{E}[X_t^{m,\bar{Z}^m}]\right|^q \right] \nonumber \\
 &\leq 2^{q-1}\mathbb{E}\left[\left|\frac{1}{n}\sum_{i=1}^{n}X_t^{*,m,i} - \frac{1}{n}\sum_{i=1}^{n}\hat{X}_t^{i}\right|^q \right] + 2^{q-1} \mathbb{E}\left[\left|\frac{1}{n}\sum_{i=1}^{n}\hat{X}_t^{i} - \mathbb{E}[X_t^{m,\bar{Z}^m}]\right|^q \right]  \nonumber \\
 & \leq \frac{2^{q-1}}{n}\sum_{i=1}^{n}\mathbb{E}\left[\left|X_t^{*,m,i} - \hat{X}_t^i\right|^q \right] + \frac{2^{q-1}}{n^q}\mathbb{E}\left[\left|\sum_{i=1}^{n} \left(\hat{X}_t^i-\mathbb{E}[X_t^{m,\bar{Z}^m}] \right)\right|^q \right] \nonumber \\
 & =  \frac{2^{q-1}}{n}\sum_{i=1}^{n}\mathbb{E}\left[\left|X_t^{*,m,i} - \hat{X}_t^i\right|^q \right] + \frac{2^{q-1}}{n^{q-1}}\Ex\left[\left|\hat{X}_t^1-\mathbb{E}[X_t^{m,\bar{Z}^m}]\right|^q \right].
\end{align}
In view of \eqref{eq:X-star-i} and \eqref{eq:X-hat-i}, it holds that
\begin{align*}
 X_t^{*,m,i} - \hat{X}_t^i
 & =  \int_{0}^{t} \left(\pi_s^{m}\mu-c_s^{m}\right)\left(X_s^{*,m,i}-\hat{X}_s^i\right)ds  + \int_{0}^{t} \left(\pi_s^{*,m, i}\mu_i-c_s^{*,m,i}+\pi_s^{m}\mu-c_s^{m}\right)X_s^{*,m,i} ds   \\
 & \quad + \int_{0}^{t} (\sigma_i\pi_s^{*,m, i}-\sigma\pi_s^{m})X_s^{*,m,i} dW_s^i + \int_{0}^{t} \sigma\pi_s^{m}\left(X_s^{*,m,i}-\hat{X}_s^i\right) dW_s^i.
\end{align*}
For $q \geq 1$, it follows from Jensen's inequality that
\begin{align*}
 \left|X_t^{*,m,i} - \hat{X}_t^i\right|^q
 &\leq 4^{q-1}\left(\int_{0}^{t} \left|\pi_s^{m}\mu-c_s^{m}\right|\left|X_s^{*,m,i}-\hat{X}_s^i\right|ds\right)^q\\
 &\quad  + 4^{q-1}\left(\int_{0}^{t} \left|\pi_s^{*,m,i}\mu_i-c_s^{*,m,i}+\pi_s^{m}\mu-c_s^{m}\right|X_s^{*,m,i} ds\right)^q   \\
 & \quad + 4^{q-1}\left|\int_{0}^{t} (\sigma_i\pi_s^{*,m, i}-\sigma\pi_s^{m})X_s^{*,m,i} dW_s^i\right|^q + 4^{q-1}\left|\int_{0}^{t} \sigma\pi_s^{m}\left(X_s^{*,m,i}-\hat{X}_s^i\right) dW_s^i\right|^q.
\end{align*}
We first take the expectation on both sides of above inequality.  By applying Burkholder-Davis-Gundy inequality, and H\"{o}lder inequality for $q_0>1$ and $q_1>1$ respectively satisfying $\frac{1}{q}+\frac{1}{q_0}=1$ and $\frac{2}{q}+\frac{1}{q_1}=1$ (for the case with $q=2$, we don't need to apply H\"older inequality), there exists a constant $K_q>0$ that might be different from line to line such that, for all $t\in[0,T]$,
{\small
\begin{align*}
& \E\left[\left|X_t^{*,m,i} - \hat{X}_t^i\right|^q \right]\leq 4^{q-1}\left(\int_{0}^{t} \left|\pi_s^{m}\mu-c_s^{m}\right|^{q_0}ds\right)^{\frac{q}{q_0}} \E\left[\int_{0}^{t}\left|X_s^{*,m,i}-\hat{X}_s^i\right|^q ds\right]\nonumber \\
&\qquad + 4^{q-1}\left(\int_{0}^{t} \left|\pi_s^{*,m,i}\mu_i-c_s^{*,m,i}+\pi^{m}_s\mu-c_s^{m}\right|^{q_0} ds\right)^{\frac{q}{q_0}} \int_{0}^{t}\E\left[\left(X_s^{*,m,i}\right)^q\right]ds   \\
 &\qquad + K_q \E\left[\left(\int_{0}^{t} \left(\sigma_i\pi_s^{*,m,i}-\sigma\pi^{m}_s\right)^2 \left(X_s^{*,m,i}\right)^2 ds\right)^{\frac{q}{2}} \right]+ K_q\E\left[\left(\int_{0}^{t}(\sigma\pi_s^{m})^2\left|X_s^{*,m,i}-\hat{X}_s^i\right|^2 ds\right)^{\frac{q}{2}} \right] \\
&\quad \leq K_q \int_{0}^{t} \E\left[\left|X_s^{*,m,i}-\hat{X}_s^i\right|^q \right]ds + 4^{q-1}C_qT \left(\int_{0}^{t} \left|\pi_s^{*,m,i}\mu_i-c_s^{*,m,i}+\pi_s^{m}\mu-c_s^{m}\right|^{q_0} ds\right)^{\frac{q}{q_0}}  \\
 &\quad + K_q \left(\int_{0}^{t}|\sigma_i\pi_s^{*,m,i}-\sigma\pi_s^{m}|^{2q_1}ds\right)^{\frac{q}{2q_1}} \int_{0}^{t} \E\left[\left(X_s^{*,m,i}\right)^q \right]ds \\
 &\quad + K_q\left(\int_{0}^{t}(\sigma\pi_s^{m})^{2q_1}ds\right)^{\frac{q}{2q_1}} \int_{0}^{t} \E\left[\left|X_s^{*,m,i}-\hat{X}_s^i\right|^q \right]ds  \\
 &\quad \leq K_q\int_{0}^{t} \E\left[\left|X_s^{*,m,i}-\hat{X}_s^i\right|^q \right]ds + 4^{q-1}C_qT \left(\int_{0}^{t}\left|\pi_s^{*,m,i}\mu_i-c_s^{*,m,i}+\pi_s^{m}\mu-c_s^{m}\right|^{q_0} ds\right)^{\frac{q}{q_0}}
 \\
 &\qquad  + K_q C_q T\left(\int_{0}^{t}|\sigma_i\pi_s^{*,m,i}-\sigma\pi_s^{m}|^{2q_1}ds\right)^{\frac{q}{2q_1}}
 + K_q\left(\int_{0}^{t}(\sigma\pi_s^{m})^{2q_1}ds\right)^{\frac{q}{2q_1}} \int_{0}^{t} \E\left[\left|X_s^{*,m,i}-\hat{X}_s^i\right|^q \right]ds \\
 &\quad \leq K_q \int_{0}^{t} \E\left[\left|X_s^{*,m,i}-\hat{X}_s^i\right|^q \right]ds + K_qT^{1+\frac{q}{q_0}}\sup_{t\in[0,T]} \left|\pi^{*,m,i}_t\mu_i-\pi^{m}_t\mu\right|^q\\ &
 \quad+ K_q C_qT^{1+\frac{q}{2q_1}}\sup_{t\in[0,T]}\left|\sigma_i\pi^{*,m,i}_t-\sigma\pi^{m}_t\right|^{q} + K_q T \left(\int_{0}^{t} \left|c_s^{*,m,i}-c_s^{m}\right|^{q_0} ds\right)^{\frac{q}{q_0}},
\end{align*}}
where the constant $C_q >0$ is given in Lemma~\ref{lem:auxiliary-results}-{\bf(i)}.

Note that $\frac{1}{q}+\frac{1}{q_0}=1$. Then $q_0=\frac{q}{q-1}\leq q$ since $q\geq2$, and hence $\frac{q}{q_0}=q-1\geq1$. This yields from H\"older's inequality that
\begin{align*}
\left(\int_{0}^{t} \left|c_s^{*,m,i}-c_s^{m}\right|^{q_0} ds\right)^{\frac{q}{q_0}}\leq K_{q,T}\left(\int_{0}^{t} \left|c_s^{*,m,i}-c_s^{m}\right|^{q} ds\right),\quad \forall~t\in[0,T],
\end{align*}
for some constant $K_{q,T}$ depending on $(q,T)$. Thus, by \eqref{eq:picstar} in Lemma~\ref{lem:solHJBlimit}, and \eqref{eq:pi-c-star-i}, we arrive at, for all $t\in[0,T]$,
\begin{align}\label{eq:X-star-hat-i}
&\frac{1}{n}\sum_{i=1}^{n}\E\left[\left|X_t^{*,m,i} - \hat{X}_t^i\right|^q\right]
 \leq  K_{q,T}\int_{0}^{t} \frac{1}{n}\sum_{i=1}^{n} \E\left[\left|X_s^{*,m,i}-\hat{X}_s^i\right|^q\right]ds + K_{q,T}\int_{0}^{t} \frac{1}{n}\sum_{i=1}^{n} \left|c_s^{*,m,i}-c_s^{m}\right|^{q}ds\nonumber\\
 &\qquad\qquad + \frac{K_{q,T}}{n}\sum_{i=1}^{n}\left[\left|\frac{\mu_i^2}{(1-p_i)\sigma_i^2} - \frac{\mu^2}{(1-p)\sigma^2}\right|^q + \left|\frac{\mu_i}{(1-p_i)\sigma_i}-\frac{\mu}{(1-p)\sigma}\right|^{q} \right].
\end{align}
It follows from the assumption $\bm{(A_c)}$ that
\begin{align}\label{eq:parameterconerror}
\begin{cases}
\displaystyle \frac{1}{n}\sum_{i=1}^{n}\left[\left|\frac{\mu_i^2}{(1-p_i)\sigma_i^2} - \frac{\mu^2}{(1-p)\sigma^2}\right|^q + \left|\frac{\mu_i}{(1-p_i)\sigma_i}-\frac{\mu}{(1-p)\sigma}\right|^q\right] = O\left(n^{-\frac{q}{2}}\right), \\ \\
\displaystyle   \sup_{t\in[0,T]}\frac{1}{n}\sum_{i=1}^{n}\left|c_t^{*,m,i}-c_t^{m}\right|^{q} = O\left(n^{-\frac{q}{2}}\right).
\end{cases}
\end{align}
By using \eqref{eq:X-star-hat-i}, the Gronwall's lemma yields that
\begin{align}\label{eq:X-i-star-hat}
 \sup_{t\in[0,T]}\frac{1}{n}\sum_{i=1}^{n} \E\left[\left|X_t^{*,m,i} - \hat{X}_t^i\right|^q \right]  = O\left(n^{-\frac{q}{2}}\right).
\end{align}
Meanwhile, we also have from \eqref{eq:parameterconerror} that $ I_2^{(n)}(t)= O(n^{-\frac{q}{2}})$. Finally, it is obvious to have that
\begin{align}\label{eq:conerror00}
\frac{2^{q-1}}{n^{q-1}}\sup_{t\in[0,T]}\Ex\left[\left|\hat{X}_t^1-\mathbb{E}[X_t^{*,m,1}]\right|^q \right] =O\left(n^{-(q-1)}\right).
\end{align}
Therefore, using \eqref{eq:diffXXstar00} and \eqref{eq:conerror00}, it holds that $ I_1^{(n)}(t)= O(n^{-\alpha_q})$, with $\alpha_q = \frac{q}{2}\wedge(q-1)$. This gives \eqref{eq:lem-Z-n-star}.
\end{proof}

We are now ready to prove Theorem~\ref{thm:epsilon-Nash-equilibrium}.

\begin{proof}[Proof of Theorem~\ref{thm:epsilon-Nash-equilibrium}]
 For $i=1,\ldots,n$, let $Z^i=(Z_t^i)_{t\in[0,T]}$ and $Z^{*,m,i}=(Z_t^{*,m,i})_{t\in[0,T]}$ be the habit formation processes of agent $i$ under an arbitrary admissible strategy $(\pi^i,c^i)\in {\cal A}^{m,i}(x_0)$ and under the strategy $(\pi^{*,m,i},c^{*,m,i})\in{\cal A}^{m,i}(x_0)$ given in~\eqref{eq:pi-c-star-i} respectively. For ease of presentation, let us denote $\bar{Z}^{*, m, n,-i}_t := \frac{1}{n}\sum_{j\neq i} Z_t^{\ast,m,j}$, $t\in[0,T]$. Recall the objective functionals $J^m_i$ in \eqref{eq:Objective-i}. Then, we have that, for $i = 1,\ldots,n$,
\begin{align*}
J^m_i( (\pi^i,c^i),(\bm{\pi}^{*,m},\bm{c}^{*,m})^{-i}) = \Ex\left[\int_0^T U_i\left(\frac{c^i_sX_s^i}{\left(\bar{Z}_s^{*,m,n,-i}+\frac{1}{n}Z_s^i\right)^{{\alpha}}}\right)ds + U_i(X_T^i) \right],
\end{align*}
and
\begin{align*}
 J_i^m((\pi^{*,m,i},c^{*,m,i}),(\bm{\pi}^{*,m},\bm{c}^{*,m})^{-i}) = \Ex\left[\int_0^T U_i\left(\frac{c^{*,m,i}_sX_s^{*,m,i}}{\left(\bar{Z}_s^{*,m,n}\right)^{{\alpha}}}\right)ds +U_i(X_T^{*,m,i})\right],
\end{align*}
where $X^{*,m,i}=(X_t^{*,m,i})_{t\in[0,T]}$ (resp. $X^i = (X_t^i)_{t\in[0,T]}$) obeys the dynamics \eqref{eq:X-star-i} (resp.~\eqref{eq:X-i} under an arbitrary admissible strategy $(\pi^i,c^i)\in \mathcal{A}^{m,i}(x_0)$).

To show \eqref{eq:def-epsilon-Nash} in Definition \autoref{def:nashequilibrium}, let us introduce an auxiliary optimal control problem ($\mathbf{P}^m$): for $\bar{Z}^m=(\bar{Z}^m_t)_{t\in[0,T]}$ being the unique fixed point to \eqref{eq:fixed-point-4} in \autoref{prop:fixedpoint}, let us consider
\begin{equation}\label{eq:objective-J-bar-P}
 \sup_{(\pi^i,c^i) \in \mathcal{A}^{m,i}(x_0)} \bar{J}_i^m((\pi^i,c^i);\bar{Z}^m) := \sup_{(\pi^i,c^i) \in \mathcal{A}^{m,i}(x_0)} \Ex\left[\int_0^T U_i\left(\frac{c^i_sX_s^i}{(\bar{Z}^m_s)^{{\alpha}}}\right)ds + U_i(X_T^i)\right].
\end{equation}
We get that the optimal strategy of the auxiliary problem ($\mathbf{P}^m$) coincides with $(\pi^{*,m,i}, c^{*,m,i})$ constructed in \eqref{eq:pi-c-star-i} for the given $\bar{Z}^m$.

Similar to \eqref{eq:J-J-two}, we have that
\begin{align}\label{eq:J-J}
 & \sup_{(\pi^i,c^i) \in \mathcal{A}^{m,i}(x_0)} J^m_i\left( (\pi^i,c^i),(\bm{\pi}^{*,m},\bm{c}^{*,m})^{-i} \right) - J^m_i\left(\bm{\pi}^{*,m},\bm{c}^{*,m}\right)\nonumber  \\
  &\qquad\leq \sup_{(\pi^i,c^i)\in\mathcal{A}^{m,i}(x_0)}\left(J^m_i\left((\pi^i,c^i),(\bm{\pi}^{*,m},\bm{c}^{*,m})^{-i}\right)  -\bar{J}^m_i\left((\pi^i,c^i);\bar{Z}^m\right)\right) \nonumber \\
  &\qquad\quad + \sup_{(\pi^i,c^i)\in \mathcal{A}^{m,i}(x_0)}  \bar{J}^m_i\left((\pi^i,c^i);\bar{Z}^m\right)-J^m_i\left(\bm{\pi}^{*,m},\bm{c}^{*,m}\right).
\end{align}
We first evaluate the first term of RHS of \eqref{eq:J-J}, we have that
\begin{align*}
 & J_i((\pi^i,c^i),(\bm{\pi}^{*,m,i},\bm{c}^{*,m,i})^{-i}) - \bar{J}_i((\pi^i,c^i);\bar{Z}^m) \\
 &\quad = \left\{ \Ex\left[\int_0^T U_i\left(\frac{c^i_sX_s^i}{(\bar{Z}_s^{*,m,n,-i})^{{\alpha}}}\right)ds\right]
 -  \Ex\left[\int_0^T U_i\left(\frac{c^i_sX_s^i}{(\bar{Z}_s^{*,m,n})^{{\alpha}}}\right)ds\right] \right\}\\
 & \qquad + \left\{\Ex\left[\int_0^T U_i\left(\frac{c^i_sX_s^i}{(\bar{Z}_s^{*,m,n})^{{\alpha}}}\right)ds\right] -\Ex\left[\int_0^T U_i\left(\frac{c^i_sX_s^i}{(\bar{Z}^m_s)^{{\alpha}}}\right)ds\right] \right\}\\
 &\quad := I_i^{(1)} + I_i^{(2)}.
\end{align*}
First, it is clear that $I_i^{(1)}= \Ex\left[\int_0^T\frac{1}{p_i}\left(c^i_sX_s^i\right)^{p_i} \left[\left(\bar{Z}_s^{*,m,n,-i}\right)^{- {\alpha}p_i} - \left(\bar{Z}_s^{*,m,n}\right)^{-{\alpha}p_i}\right]ds\right]$. Note that $p_i\in(0,1)$ and $\alpha\in(0,1]$. Using the inequality $|a^{{\alpha}p_i}- b^{{\alpha}p_i}| \leq {\alpha}p_i |a-b|\max\{a^{{\alpha}p_i-1} ,b^{{\alpha}p_i-1}\}$ for all $a,b >0$, we can derive on $\{Z_t^{*,m,i} > Z_t^i\}$ that
\begin{align}\label{eq:I-1-i-1}
  &\left(\bar{Z}_t^{*,m,n,-i}\right)^{-{\alpha}p_i} - \left(\bar{Z}_t^{*,m,n}\right)^{-{\alpha}p_i}
 = \frac{1}{(\bar{Z}_t^{*,m,n,-i})^{{\alpha}p_i}\left(\bar{Z}_t^{\ast,m,n}\right)^{{\alpha}p_i}}
  \left[ \left(\bar{Z}_t^{*,m,n}\right)^{{\alpha}p_i} - \left(\bar{Z}_t^{\ast,m,n,-i}\right)^{{\alpha}p_i}\right] \nonumber\\
 &\qquad \leq \frac{{\alpha}p_i}{n}\frac{1}{(\bar{Z}_t^{\ast,m,n,-i})^{{\alpha}p_i} \left(\bar{Z}_t^{\ast,m,n}\right)^{{\alpha}p_i}} \left(Z_t^{\ast,m,i}- Z_t^{i}\right) \max\left\{\left(Z_t^{*,m,n}\right)^{{\alpha}p_i-1} , \left(Z_t^{*,m,n,-i} \right)^{{\alpha}p_i-1}\right\} \nonumber  \\
 &\qquad \leq \frac{{\alpha}p_i}{n}\frac{1}{\left(\frac{1}{n}\sum_{j\neq i}Z_t^{*,m,i}\right)^{2{\alpha}p_i} }Z_t^{\ast,m,i}\left(Z_t^{*,m,n,-i}\right)^{{\alpha}p_i-1}   \leq \frac{{\alpha}p_i}{n}Z_t^{*,m,i} \left(\frac{1}{n}\sum_{j\neq i}Z_t^{*,m,j}\right)^{-{\alpha}p_i-1}.
\end{align}
Obviously, the inequality \eqref{eq:I-1-i-1} trivially holds on $\{Z_t^{*,m,i} \leq Z_t^i\}$. Note that $Z^{*,m,i}=(Z_t^{*,m,i})_{t\in[0,T]}$ is positive. Then, it follows from Jensen's inequality that
\begin{equation}\label{eq:I-1-i-2}
   \left(\frac{1}{n}\sum_{j\neq i}Z_t^{*,m,j}\right)^{-{\alpha}p_i-1}  \leq \left(\frac{n-1}{n}\right)^{-{\alpha}p_i-1} \frac{1}{n-1}
   \sum_{j\neq i} \left(Z_t^{*,m,j}\right)^{-{\alpha}p_i-1}.
\end{equation}
Combining \eqref{eq:I-1-i-1} and \eqref{eq:I-1-i-2}, and using the generalized H\"{o}lder inequality for any $q_1, q_2 >1$ satisfying $\frac{1}{q_1} + \frac{1}{q_2} + p_i = 1$, we derive from Lemma~\ref{lem:auxiliary-results} (c.f.~{\bf(ii)} and {\bf(iv)}) at
\begin{align}\label{eq:I-1-i-order}
  I_i^{(1)} & \leq \mathbb{E}\left[\int_0^T \frac{1}{p_i}\left(c^i_tX_t^i\right)^{p_i} \frac{{\alpha}p_i}{n}Z_t^{*,m,i} \left(\frac{1}{n}\sum_{j\neq i}Z_t^{*,m,j}\right)^{-{\alpha}p_i-1} dt\right]  \\
  & \leq \frac{{\alpha}}{n} \left(\frac{n-1}{n}\right)^{-{\alpha}p_i-1} \frac{1}{n-1} \sum_{j\neq i} \Ex\left[\int_0^T \left(c^i_tX_t^i\right)^{p_i} Z_t^{*,m,i} \left(Z_t^{*,m,j}\right)^{-{\alpha}p_i-1} dt\right] \nonumber \\
  &\leq \frac{{\alpha}}{n} \left(\frac{n-1}{n}\right)^{-{\alpha}p_i-1} \frac{1}{n-1} \sum_{j\neq i} \int_0^T \mathbb{E}\left[c^i_tX_t^i\right]^{p_i} \mathbb{E}\left[\left(Z_t^{*,m,i}\right)^{q_1}\right]^{\frac{1}{q_1}} \mathbb{E}\left[ \left(Z_t^{*,m,j}\right)^{-({\alpha}p_i+1)q_2}\right]^{\frac{1}{q_2}} dt  \nonumber \\
  &\leq \frac{{\alpha}}{n} \left(\frac{n-1}{n}\right)^{-{\alpha}p_i-1} \frac{1}{n-1}\sum_{j\neq i} C_{q_1}^{\frac{1}{q_1}} C_{-({\alpha}p_i+1)q_2}^{\frac{1}{q_2}}  \left(\int_0^T \mathbb{E}\left[c^i_tX_t^i\right]^{p_i}dt\right)   \nonumber \\
  &\leq \frac{{\alpha}}{n} \left(\frac{n-1}{n}\right)^{-{\alpha}p_i-1} C_{q_1}^{\frac{1}{q_1}} C_{-({\alpha}p_i+1)q_2}^{\frac{1}{q_2}}  T^{1-p_i}\left(\int_0^T \mathbb{E}\left[c^i_tX_t^i\right]dt \right)^{p_i}= O( n^{-1} ),\nonumber
\end{align}
where the constant $C_q$ with $q\in(1,\infty)\cup(-\infty,0)$ is given in Lemma~\ref{lem:auxiliary-results}-{\bf(ii)}. In the last inequality, we used H\"older inequality that $\int_0^T \mathbb{E}\left[c^i_tX_t^i\right]^{p_i}dt\leq (\int_0^T 1^{\frac{1}{1-p_i}}dt)^{1-p_i}(\int_0^T \Ex[c^i_tX_t^i]dt)^{p_i} $.

Similarly, for the term $I_i^{(2)}$, we can apply H\"{o}lder inequality and the estimate \eqref{eq:lem-Z-n-star} in Lemma \ref{lem:auxiliary-results} to get that
\begin{align}\label{eq:I-i-2}
 I_i^{(2)} 
 &\leq \frac{1}{p_i} \Ex\left[ \int_0^T \left(c^i_sX_s^i\right)^{p_i} \frac{1}{ \left(\bar{Z}^m_s\right)^{{\alpha}p_i}\left(\bar{Z}_s^{*,m,n}\right)^{{\alpha}p_i}} \left[ \left(\bar{Z}^m_s\right)^{{\alpha}p_i} - \left(\bar{Z}_s^{*,m,n}\right)^{{\alpha}p_i} \right] ds\right] \nonumber \\
 & \leq \alpha\Ex\left[\int_0^T \left(c^i_sX_s^i\right)^{p_i} \frac{1}{ \left(\bar{Z}^m_s\right)^{{\alpha}p_i}\left(\bar{Z}_s^{*,m,n}\right)^{{\alpha}p_i}} \left|\bar{Z}^m_s-\bar{Z}_s^{*,m,n}\right| \max\left\{\left(\bar{Z}^m_s\right)^{{\alpha}p_i-1}, \left(\bar{Z}_s^{*,m,n}\right)^{{\alpha}p_i-1} \right\} ds\right] \nonumber \\
 &\leq {\alpha} \Ex\left[\int_0^T \left(c^i_sX_s^i\right)^{p_i}\left|\bar{Z}^m_s-\bar{Z}_s^{*,m,n}\right| \max\left\{(\bar{Z}_s^m)^{-{\alpha}p_i-1}, (\bar{Z}_s^{*,m,n})^{-{\alpha}p_i-1} \right\} ds\right].
\end{align}
Using H\"{o}lder inequality, Lemma \ref{lem:auxiliary-results} (c.f. {\bf(iii)} and {\bf(iv)}) and \autoref{prop:fixedpoint} with~\eqref{eq:fixed-point-4}, we have that, for all $(\pi^i,c^i)\in{\cal A}(x_0)$,
\begin{align}\label{eq:I-i-2-1}
  & \Ex\left[\int_0^T \left(c^i_sX_s^i\right)^{p_i}\left|\bar{Z}^m_s-\bar{Z}_s^{*,m,n}\right| \left(\bar{Z}_s^m\right)^{-{\alpha}p_i-1} ds\right]
  \leq z^{-{\alpha}p_i-1}\int_0^T \Ex\left[\left(c^i_sX_s^i\right)^{2p_i}\right]^{\frac{1}{2}}\Ex\left[\left|\bar{Z}^m_s-\bar{Z}_s^{*,m,n}\right|^2\right]^{\frac{1}{2}}ds\nonumber\\
  &\qquad\quad\leq z^{-\alpha p_i-1}K^{p_i}\int_0^T \Ex\left[\left|X_s^i\right|^{2p_i}\right]^{\frac{1}{2}}\Ex\left[\left|\bar{Z}^m_s-\bar{Z}_s^{*,m,n}\right|^2\right]^{\frac{1}{2}} ds=  O(n^{-\frac{1}{2}}).
\end{align}
Note that, in view of \eqref{eq:X-i}, we have that, for all $(\pi^i,c^i)\in{\cal A}^{m,i}(x_0)$, and $q\geq2$,
\begin{align}\label{eq:Ex-X-i-q}
 \Ex\left[\left|X_t^i\right|^q\right] 
& \leq x^q\Ex^{\mathbb{Q}_i}\left[\exp\left(\int_0^t\left(q\mu_i\pi_s^i + \frac{q(q-1)}{2}\sigma_i^2(\pi_s^i)^2\right)ds\right)\right]\nonumber\\
 & \leq x^q \Ex^{\mathbb{Q}_i}\left[\exp\left(\int_0^t\left(qK\mu_i + \frac{q(q-1)}{2}K^2\sigma_i^2\right)ds\right)\right],
\end{align}
where we defined the probability measure $\mathbb{Q}_i\sim\mathbb{P}$ with $\frac{d\mathbb{Q}_i}{d\mathbb{P}}|_{{\cal F}_t}=\exp(q\int_0^t\sigma_i\pi_s^idW_s^i-\frac{q^2}{2}\int_0^t\sigma_i^2|\pi_s^i|^2ds)$ for $t\in[0,T]$. By applying the assumption $\bm{(A_c)}$, $(\mu_i,\sigma_i)\to(\mu,\sigma)$ as $i\to\infty$, and hence the sequence $((\mu_i,\sigma_i))_{i\in\N}$ is bounded (denote by $C$ the bound of this sequence). This implies that $\Ex^{\mathbb{Q}_i}\left[\exp\left(\int_0^t\left(qK\mu_i + \frac{q(q-1)}{2}K^2\sigma_i^2\right)ds\right)\right]\leq \exp\left(\left(qKC + \frac{q(q-1)}{2}K^2C^2\right)T\right)$ for all $t\in[0,T]$. This yields that $M_q:=\sup_{i\in\N}\sup_{t\in[0,T]}\Ex[|X_t^i|^q]<\infty$ for any $q\geq2$. Hence, by the assumption $\bm{(A_c)}$, it follows from \eqref{eq:I-i-2-1} and Lemma~\ref{lem:auxiliary-results}-{\bf(ii)} with $q=2$ that
{\small
\begin{align*}
  \Ex\left[\int_0^T \left(c^i_sX_s^i\right)^{p_i}\left|\bar{Z}^m_s-\bar{Z}_s^{*,m,n}\right| \left(\bar{Z}_s^m\right)^{-\alpha p_i-1} ds\right]
  &\leq z^{{-p_i-1}}K^{p_i}M_{2}^{1/2}T\sup_{t\in[0,T]}\Ex\left[\left|\bar{Z}_t^{*,m,n} -\bar{Z}^m_t\right|^2\right]^{\frac{1}{2}} =  O\left(n^{-\frac{1}{2}}\right).
\end{align*}}On the other hand, by Jensen's inequality and the generalized H\"{o}lder inequality with $q_1,q_2\geq2$ satisfying $q_1^{-1}+q_2^{-1}=\frac{1}{2}$, it follows that, for all $(\pi^i,c^i)\in{\cal A}^{m,i}(x_0)$,
\begin{align}\label{eq:I-i-2-2}
  &\Ex\left[\int_0^T \left(c^i_sX_s^i\right)^{p_i}\left|\bar{Z}^m_s-\bar{Z}_s^{*,m,n}\right| \left(\bar{Z}_s^{*,m,n}\right)^{-{\alpha}p_i-1} ds\right]\nonumber\\
  &\qquad\leq  \frac{1}{n}\sum_{j=1}^{n} \int_0^T \Ex\left[\left(c^i_sX_s^i\right)^{p_i}\left|\bar{Z}^m_s-\bar{Z}_s^{*,m,n}\right|(Z_s^{*,m,j})^{-{\alpha}p_i-1}\right]ds\nonumber\\
   &\qquad \leq \frac{1}{n}\sum_{j=1}^{n} \int_0^T \Ex\left[\left(c^i_sX_s^i\right)^{q_1p_i}\right]^{\frac{1}{q_1}}\Ex\left[\left|\bar{Z}^m_s-\bar{Z}_s^{*,m,n}\right|^2\right]^{\frac{1}{2}}
   \Ex\left[\left|Z_s^{*,m,j}\right|^{-({\alpha}p_i+1)q_2}\right]^{\frac{1}{q_2}}ds \nonumber\\
  &\qquad \leq K^{p_i}C_{-(p_i+1)q_2}^{\frac{1}{q_2}}\int_0^T \Ex\left[\left|X_s^i\right|^{q_1}\right]^{\frac{1}{q_1}}\Ex\left[\left|\bar{Z}^m_s-\bar{Z}_s^{*,m,n}\right|^2\right]^{\frac{1}{2}}ds\nonumber\\
  &\qquad \leq K^{p_i}C_{-({\alpha}p_i+1)q_2}^{\frac{1}{q_2}}M_{q_1}^{\frac{1}{q_1}}T\sup_{t\in[0,T]}\Ex\left[\left|\bar{Z}^m_t-\bar{Z}_t^{*,m,n}\right|^2\right]^{\frac{1}{2}}=O\left(n^{-\frac{1}{2}}\right),
\end{align}
where $C_{-(p_i+1)q_2}$ is the constant given in Lemma~\ref{lem:auxiliary-results}-{\bf(ii)} with $q=-(\alpha p_i+1)q_2<0$, and we also used Lemma~\ref{lem:auxiliary-results}-{\bf(iii)} with $q=2$ therein. By combining \eqref{eq:I-i-2-1} and \eqref{eq:I-i-2-2}, it yields that
\begin{align}\label{eq:I-i-2-order}
  I_i^{(2)} &= \mathbb{E}\left[\int_0^T U_i\left(\frac{c^i_sX_s^i}{(\bar{Z}_s^{*,m,n})^{{\alpha}}}\right)ds \right] -\mathbb{E}\left[\int_0^T U_i\left(\frac{c^i_sX_s^i}{(\bar{Z}^m_s)^{{\alpha}}}\right)ds\right]= O\left(n^{-\frac{1}{2}}\right).
\end{align}
For the second term of RHS of \eqref{eq:J-J}, it follows from \eqref{eq:pi-c-star-i} that
\begin{align*}
 \sup_{(\pi^i,c^i)\in\mathcal{A}(x_0)}\bar{J}^m_i((\pi^i,c^i);\bar{Z}^m)-J^m_i(\bm{\pi}^{*,m},\bm{c}^{*,m})& = \bar{J}^m_i((\pi^{*,m,i},c^{*,m,i});\bar{Z}^m) - J^m_i(\bm{\pi}^{*,m},\bm{c}^{*,m}).
\end{align*}
We can therefore follow the similar argument in showing the convergence error of $I_i^{(2)}$ with the fixed $(\bm{\pi}^{*,m},\bm{c}^{*,m})$ to get that
\begin{align}\label{eq:J-bar-J-1}
 \bar{J}^m_i((\pi^{*,m,i},c^{*,m,i});\bar{Z}^m) - J^m_i(\bm{\pi}^{*,m},\bm{c}^{*,m}) = O\left(n^{-\frac{1}{2}}\right).
\end{align}
Finally, the estimates \eqref{eq:I-1-i-order}, \eqref{eq:I-i-2-order} and \eqref{eq:J-bar-J-1} jointly yield \eqref{eq:def-epsilon-Nash} with $\epsilon_n=O(n^{-\frac{1}{2}})$.
\end{proof}

\section{Conclusions}\label{sec:con}
We study the equilibrium consumption under external habit formation in the MFG and n-player game framework. By assuming the asset specialization for each agent, the external habit formation preference can be naturally regarded as the relative performance, where the interaction of agents occurs via the average external habit process. Both linear (addictive) habit formation and multiplicative (non-addictive) habit formation are considered in the present work, and one mean field equilibrium can be characterized in analytical form in each MFG problem. For each preference, we also establish the connection to the n-player game by constructing its approximate Nash equilibrium using the mean field equilibrium.

For the future research, it will be interesting to extend our current work to MFGs and n-player games with common shock and random market model parameters, in which the mean field habit formation process becomes a stochastic process instead of a deterministic function and the FBSDE approach needs to be developed. It is also appealing to investigate the MFGs and n-player games when the external habit formation is defined by the average of the past spending maximum from all peers in the linear form as studied by \cite{DLPY22} and \cite{LiYuZ21} and in the multiplicative form as studied by \cite{Guasoni}. New difficulties arise in the verification of consistency condition and the approximation of the $n$-player Nash equilibrium due to the structure of the running maximum process.

\ \\
{\small
\textbf{Acknowledgements}: L. Bo and S. Wang are supported by National Key R\&D Program of China (No. 2022YFA1000033), National Natural Science Foundation of China (No. 11971368), Natural Science Basic Research Program of Shaanxi (No. 2023-JC-JQ-05) and the Fundamental Research Funds for the Central Universities (No. 20199235177). X. Yu is supported by the Hong Kong Polytechnic University research grant under no. P0039251.}

\ \\

\end{document}